%% file: MACjrnl.tex
\documentclass[journal,onecolumn]{IEEEtran}
\usepackage{amsmath,amsthm,amssymb,color,graphicx,caption,subcaption,comment,tikz,mathtools}
\usepackage{pgfplots}
\usepackage{verbatim}
\usetikzlibrary{decorations.markings}
\usepgfplotslibrary{units}
\usetikzlibrary{spy,backgrounds}
\usepackage{pgfplotstable}

\newtheorem{theorem}{Theorem}

\newtheorem{lemma}{Lemma}
\newtheorem{proposition}{Proposition}
\newtheorem{corollary}{Corollary}

\newtheorem{remark}{Remark}

\ifCLASSINFOpdf
\else
\fi

\begin{document}
%
% paper title
% Titles are generally capitalized except for words such as a, an, and, as,
% at, but, by, for, in, nor, of, on, or, the, to and up, which are usually
% not capitalized unless they are the first or last word of the title.
% Linebreaks \\ can be used within to get better formatting as desired.
% Do not put math or special symbols in the title.
\title{Capacity Bounds for Diamond Networks with an Orthogonal Broadcast Channel}
%
%
% author names and IEEE memberships
% note positions of commas and nonbreaking spaces ( ~ ) LaTeX will not break
% a structure at a ~ so this keeps an author's name from being broken across
% two lines.
% use \thanks{} to gain access to the first footnote area
% a separate \thanks must be used for each paragraph as LaTeX2e's \thanks
% was not built to handle multiple paragraphs
%
\author{Shirin~Saeedi Bidokhti,~\IEEEmembership{Member,~IEEE,}
      and  Gerhard~Kramer,~\IEEEmembership{Fellow,~IEEE,}% <-this % stops a space
\thanks{S. Saeedi Bidokhti and G. Kramer are with the Department for Electrical and Computer Engineering, Technische Universit\"{a}t M\"{u}nchen, Germany (shirin.saeedi@tum.de, gerhard.kramer@tum.de).}% <-this % stops a space
\thanks{The authors are supported by the German Federal Ministry of Education and Research in the framework of the Alexander von Humboldt-Professorship. The work of S. Saeedi Bidokhti was also supported by the Swiss National Science Foundation Fellowship no. 146617.}% <-this % stops a space
%\thanks{Manuscript received April 19, 2005; revised September 17, 2014.}
\thanks{The material in this paper was presented in part at the 2014 IEEE International Symposium on Information Theory, Honolulu, USA.}}

\maketitle

% As a general rule, do not put math, special symbols or citations
% in the abstract or keywords.
\begin{abstract}
A class of diamond networks is studied where the broadcast component is orthogonal and modeled by two independent bit-pipes. New upper and lower bounds on the capacity are derived. The proof technique for the upper bound generalizes bounding techniques of Ozarow for the Gaussian multiple description problem (1981) and Kang and Liu for the Gaussian diamond network (2011). The lower bound is based on Marton's coding technique and superposition coding. The bounds are evaluated for Gaussian and binary adder multiple access channels (MACs). For Gaussian MACs, both the lower and upper bounds strengthen the Kang-Liu bounds and establish capacity for interesting ranges of bit-pipe capacities. {For binary adder MACs, the capacity is established for all ranges of bit-pipe capacities.} 
\end{abstract}

% Note that keywords are not normally used for peerreview papers.
%\begin{IEEEkeywords}
%IEEEtran, journal, \LaTeX, paper, template.
%\end{IEEEkeywords}

% For peer review papers, you can put extra information on the cover
% page as needed:
% \ifCLASSOPTIONpeerreview
% \begin{center} \bfseries EDICS Category: 3-BBND \end{center}
% \fi
%
% For peerreview papers, this IEEEtran command inserts a page break and
% creates the second title. It will be ignored for other modes.
\IEEEpeerreviewmaketitle

\section{Introduction}
% The very first letter is a 2 line initial drop letter followed
% by the rest of the first word in caps.
% 
% form to use if the first word consists of a single letter:
% \IEEEPARstart{A}{demo} file is g..
% 
% form to use if you need the single drop letter followed by
% normal text (unknown if ever used by IEEE):
% \IEEEPARstart{A}{}demo file is ....
% 
% Some journals put the first two words in caps:
% \IEEEPARstart{T}{his demo} file is ....
% 
% Here we have the typical use of a "T" for an initial drop letter
% and "HIS" in caps to complete the first word.
%\IEEEPARstart{T}{his} demo file is intended to serve as a ``starter file''
%for IEEE journal papers produced under \LaTeX\ using
%IEEEtran.cls version 1.8a and later.
% You must have at least 2 lines in the paragraph with the drop letter
% (should never be an issue)
%I wish you the best of success.

%\subsection{Introduction}
The diamond network \cite{Schein01} is a two-hop network that is a cascade of a broadcast channel (BC) and a multiple access channel (MAC). The  two-relay diamond network has a source communicate with a sink through two relay nodes that do not have information of their own to communicate. The underlying challenge may be described as follows. In order to fully utilize the MAC to the receiver, we would  like to achieve full cooperation at the relay nodes. On the other hand,  to better use the diversity that is offered by the relays, we would like to send independent information to the relay nodes over the BC.

The  problem of finding the capacity of this network is unresolved. Lower and upper bounds  on the capacity are given in \cite{Schein01}. An interesting class of networks is when the  BC and/or MAC are modelled via orthogonal links \cite{TraskovKramer07, KangLiu11,KangUlukus11,SanderovichShamaiSteinbergKramer08}. The problem is solved for linear deterministic relay networks, and the capacity of Gaussian relay networks has been approximated within a constant number of bits \cite{AvestimehrDiggaviTse11}. The capacity of Gaussian diamond networks with $n$ relays is studied in \cite{NiesenDiggavi11,SenguptaHsiangFragouli12,ChernOzgur12}. 
These works propose relaying strategies that achieve the cut-set upper bound up to an additive (or multiplicative) gap. %The cut-set bound is the best known upper bound on the capacity of diamond networks, but it is not tight in general. 

In this paper, we study capacity bounds when there are two relays and the BC is orthogonal, which means that the BC may as well have two independent bit-pipes.
This problem was  studied in \cite{TraskovKramer07} where lower and upper bounds were derived on the capacity. Recently, \cite{KangLiu11} studied  a Gaussian MAC and  derived a new upper bound that constrains the mutual information between the MAC inputs. The bounding technique in \cite{KangLiu11} is motivated by \cite{Ozarow80} that treats the Gaussian multiple description problem. Unfortunately, neither result seems to apply to discrete memoryless channels. 

This paper is organized as follows. We state the problem setup in Section \ref{prel}. In Section \ref{secach}, we improve the achievable rates of \cite{TraskovKramer07} by communicating a common piece of information from the source to both relays using superposition coding and Marton's coding. In Section \ref{secupp}, we prove new capacity upper bounds by generalizing and improving the bounding technique of \cite{KangLiu11}.
Our upper bounds apply to the general class of discrete memoryless MACs, and strictly improve the cut-set bound. We study the bounds for networks with a Gaussian MAC (Section \ref{exGauss}) and a binary adder MAC (Section \ref{exAdd}). For  networks with a Gaussian MAC, we find conditions on the bit-pipe capacities such that the upper and lower bounds meet. {For networks with a binary adder MAC, we find the capacity for all ranges of bit-pipe capacities.}

% needed in second column of first page if using \IEEEpubid
%\IEEEpubidadjcol

%\newpage
\section{Preliminaries}
\label{prel}
\subsection{Notation}
Random variables are denoted by capital letters, e.g. $X$, and their realizations are denoted by small letters, e.g. $x$. The probability mass function (pmf) describing $X$ is denoted by $p_X(x)$ or $p(x)$. The entropy of $X$ is denoted by $H(X)$,  the conditional entropy of $X$ given $Y$ is denoted by $H(X|Y)$, and the mutual information between $X$ and $Y$ is denoted by $I(X;Y)$. Differential entropies are denoted by $h(X)$ and conditional differential entropies are denoted by $h(X|Y)$. Sets are denoted by script letters and matrices are denoted by bold capital letters.  The random sequence $X_1,\ldots,X_n$ is denoted by $X^n$. $\mathcal{T}^n_\epsilon(X)$ denotes the set of sequences $x^n$ that are $\epsilon-$ typical with respect to $P_X(.)$ \cite{ElGamalKim}. When $P_X(.)$ is clear from the context we write $\mathcal{T}^n_\epsilon$.

When $X$ is a Bernoulli random variable with $p_X(0)=q$, its entropy in bits is  $h_2(q)=-q\log_2(q)-(1-q)\log_2(1-q)$. The pair of random variables $(X,Y)$ is said to be a doubly symmetric binary source with parameter $p$ if $p_{XY}(0,0)=p_{XY}(1,1)=\frac{1-p}{2}$, and  $p_{XY}(0,1)=p_{XY}(1,0)=\frac{p}{2}$. Throughout this paper, all logarithms are to the base $2$. For a real number $x$, we denote $\max(x,0)$  by $x^+$. %The capacity of the diamond network is denoted by $C^{\diamond}$. 

\subsection{Model}
\input{two-user-mac}
Consider the diamond network in Fig. \ref{two-user-mac}.  A source communicates a message $W$ with $nR$ bits to a sink.  
The source encodes $W$ into the sequence $V_1^n$, which is available at relay $1$, and the sequence $V_2^n$, which is available at relay $2$. $V_1^n$ and $V_2^n$ are such that $H(V_1^n)\leq nC_1$ and $H(V_2^n)\leq nC_2$.
Each relay $i$, $i=1,2$,  maps its received sequence $V_i^n$ into a sequence $X_i^n$ which is sent over a MAC with transition probabilities $p(y|x_1,x_2)$, for each $x_1\in\mathcal{X}_1$, $x_2\in\mathcal{X}_2$, $y\in\mathcal{Y}$.
From the received sequence $Y^n$, the sink decodes an estimate $\hat{W}$ of $W$.

A coding scheme consists of an encoder, two relay mappings, and a decoder, and is said to achieve
the rate $R$ if, by choosing $n$ sufficiently large, we can make the error probability $\Pr(\hat{W}\neq W)$ as small as desired. We are interested in characterizing the largest achievable rate $R$. We refer to the maximum achievable rate as the capacity $C^\diamond$ of the network.

\input{seclowerbound}
\section{An upper bound}
\label{secupp}
The idea behind our upper bound is motivated by \cite{KangLiu11,Ozarow80}. The proposed bound applies not only to Gaussian channels, but also to general discrete memoryless channels. It strictly improves the cut-set bound as we show via two examples.
The cut-set bound \cite[Theorem 15.10.1]{CoverThomas} is given by the following Lemma.
\begin{lemma}[Cut-Set Bound]
The capacity $C^{\diamond}$ satisfies
\begin{align}
C^{\diamond}\leq \max_{p(x_1,x_2)}\min\left\{\begin{array}{l} C_1+C_2,\\
 C_1+I(X_2;Y|X_1),\\
 C_2+I(X_1;Y|X_2),\\
 I(X_1X_2;Y)\end{array}\right\}.\label{disregard}
\end{align}
\end{lemma}
\noindent The cut-set bound disregards the potential correlation between the inputs in the first term of \eqref{disregard}. More precisely, we have
\begin{align}
nR&\leq H(V_1^n,V_2^n)\nonumber\\&= H(V_1^n)+H(V_1^n)-I(V_1^n;V_2^n)\nonumber\\&\leq nC_1+nC_2-I(X_1^n;X_2^n).\label{using}
\end{align}
It is noted in \cite{TraskovKramer07} that optimizing the following $n$-letter characterization gives the capacity of the network when  $n\to \infty$:
\begin{align}
nR&\leq nC_1+nC_2-I(X^n_1;X^n_2)\label{disregardmulti}\\
nR&\leq nC_1+I(X^n_2;Y^n|X^n_1)\label{secondineq}\\
nR&\leq nC_2+I(X^n_1;Y^n|X^n_2)\label{thirdineq}\\
nR&\leq I(X^n_1X^n_2;Y^n).\label{sumwith}
\end{align}
But, infinite letter characterizations are usually non-computable and we would like to find computable bounds.

We prove the following upper bound.
\begin{theorem}
\label{thmupp}
The capacity $C^{\diamond}$ satisfies
\begin{align}
C^{\diamond}\leq \max_{p(x_1,x_2)}\min_{p(u|x_1,x_2,y)=p(u|y)}\min\left\{\begin{array}{l}
C_1+C_2,\\
C_1+I(X_2;Y|X_1),\\
C_2+I(X_1;Y|X_2),\\
 I(X_1X_2;Y),\\
\frac{1}{2}\left( C_1+C_2+I(X_{1}X_{2};Y|U)+I(X_{1};U|X_{2})+I(X_{2};U|X_{1})\right)\end{array}\right\}.\label{thmupp15}
\end{align}
\end{theorem}

%\textcolor{red}{\begin{remark}
%The upper bound given by Theorem \ref{thmupp} is in the form of a max-min problem, where the maximization is over $p(x_1,x_2)$ and the minimization is over all distributions $p(u|x_1,x_2,y)$ where $p(u|x_1,x_2,u)=p(u|y)$. \end{remark}}
\begin{remark}
\label{concave}
For a fixed auxiliary channel $p(u|x_1,x_2,y)$ and a fixed MAC $p(y|x_1,x_2)$, all RHS terms in \eqref{thmupp15} are concave in $p(x_1,x_2)$. See Appendix \ref{app1}. 
\end{remark}
\begin{remark}
The last term of the  minimum in \eqref{thmupp15} may be written as
\begin{align}
\label{equivthm5}
R\leq \frac{1}{2}\left(C_1+C_2+I(X_{1}X_{2};YU)-I(X_{1};X_2)+I(X_{1};X_2|U)\right).
\end{align}
Since we choose $p(u|x_1,x_2,y)=p(u|y)$,  the bound \eqref{equivthm5} becomes
\begin{align}
\label{equivthm6}
2R\leq  C_1+C_2+I(X_{1}X_{2};Y)-I(X_{1};X_2)+I(X_{1};X_2|U).
\end{align}
\end{remark}
\begin{proof}[Proof of Theorem \ref{thmupp}]
It is observed in \cite{KangLiu11} that $I(X_1^n;X_2^n)$ may be written in the following form for  \textit{any} integer $n$, and \textit{any} random sequence $U^n$:
\begin{align}
I(X_1^n;X_2^n)=&  I(X_1^nX_2^n;U^n)-I(X_1^n;U^n|X_2^n)-I(X_2^n;U^n|X_1^n)+I(X_1^n;X_2^n|U^n).\label{identityI}
\end{align}
Therefore, using \eqref{using} and the non-negativity of mutual information we have
\begin{align}
nR\leq& nC_1+nC_2-I(X_1^nX_2^n;U^n)+I(X_1^n;U^n|X_2^n)+I(X_2^n;U^n|X_1^n).\label{eqref}
\end{align}
To see the usefulness of \eqref{eqref}, we proceed as follows.
First, note that
\begin{align}
nR&\leq I(X_1^nX_2^n;Y^n)\leq I(X_1^nX_2^n;Y^nU^n).\label{eq4}
\end{align}
Combining inequalities  \eqref{eqref} and \eqref{eq4}, we have
\begin{align}
\label{sumbound}
2nR\leq& nC_1+nC_2+I(X_1^nX_2^n;Y^n|U^n)+I(X_1^n;U^n|X_2^n)+I(X_2^n;U^n|X_1^n).
\end{align}
Define $U_i$ from $X_{1i}, X_{2i}, Y_i$ through the channel $p_{U|X_{1} X_{2} Y}(u_i|x_{1i},x_{2i},y_i)$, $i=1,2,\ldots,n$. With this choice of $U_i$ we have the following chain of inequalities:
\allowdisplaybreaks
\begin{align}
2nR&\leq nC_1+nC_2+I(X_1^nX_2^n;Y^n|U^n)+I(X_1^n;U^n|X_2^n)+I(X_2^n;U^n|X_1^n)\nonumber\\
&= nC_1+nC_2+\sum_{i=1}^nI(X_1^nX_2^n;Y_i|U^nY^{i-1})+\sum_{i=1}^nI(X_1^n;U_i|X_2^nU^{i-1} ) + \sum_{i=1}^nI(X_2^n;U_i|X_1^nU^{i-1} )\nonumber\\
&\stackrel{(a)}{\leq} nC_1+nC_2+\sum_{i=1}^nI(X_{1i}X_{2i};Y_i|U_i)+\sum_{i=1}^nI(X_{1i};U_i|X_{2i})+\sum_{i=1}^nI(X_{2i};U_i|X_{1i})\nonumber\\
&\leq nC_1+nC_2+nI(X_{1I},X_{2I};Y_I|U_I)+nI(X_{1I};U_I|X_{2I})+nI(X_{2I};U_I|X_{1I}),
\end{align}
where $I$ is a time-sharing random variable with $p_I(i)=\frac{1}{n}$ for all $i=1,\ldots,n$. Step $(a)$ holds because of the following two Markov chains:
\begin{align}
&(X_1^nX_2^nU^nY^{i-1})-(X_{1i}X_{2i}U_i)-Y_i\\
&(X_1^nX_2^nU^{i-1})-(X_{1i}X_{2i})-U_i.\label{Markovchain}
\end{align}
From here on, for simplicity we restrict the auxiliary channel to satisfy $$p_{U|X_{1} X_{2} Y}(u|x_{1},x_{2},y)=p_{U| Y}(u|y),\quad \forall x_1\in \mathcal{X}_1,\ x_2\in \mathcal{X}_2,\ u\in\mathcal{U},\ y\in\mathcal{Y}.$$  
This proves Theorem \ref{thmupp}.
\end{proof}
%We study the upper bound of Theorem \ref{thmupp} with a Gaussian MAC and a binary adder MAC in Sections \ref{exGauss} and \ref{exAdd} respectively. In both cases, we show that the upper bound of Theorem \ref{thmupp} is tight for a wide range of bit-pipe capacities $C_1,C_2$. 

We now refine our bounding technique to derive a stronger upper bound in Theorem \ref{theoremupp}.

\begin{theorem}
\label{theoremupp}
The capacity $C^{\diamond}$ satisfies
\begin{align}
C^{\diamond}\leq \max_{p(x_1,x_2)}\min_{p(u|x_1,x_2,y)=p(u|y)}\max_{p(q|x_1,x_2,y,u)=p(q|x_1,x_2)} \!\!\min\left\{\!\!\!\!\begin{array}{l}
C_1+C_2,\\
C_1+I(X_2;Y|X_1Q),\\
C_2+I(X_1;Y|X_2Q),\\
I(X_1X_2;Y|Q),\\
C_1+C_2-I(X_1X_2;U|Q)+I(X_2;U|X_1Q)+I(X_1;U|X_2Q)\end{array}\!\!\!\!\right\}.\label{21}
\end{align}
\end{theorem}

\begin{remark}
In the above characterization, it suffices to consider $|\mathcal{Q}|\leq |\mathcal{X}_1||\mathcal{X}_2|+3$. See Appendix \ref{cardin1}.
\end{remark}
\begin{remark}
\label{remarkcardin}
The upper bound in \eqref{21} may be loosened by exchanging the order of the fist maximization and the second minimization. In this case, it suffices to consider $|\mathcal{Q}|\leq 4$. See Remark \ref{cardin2} in Appendix \ref{cardin1}.
\end{remark}
%
%\begin{remark}
%The bound of Theorem \ref{thmupp} is in the form of a max-min-max problem, where the first maximization is over $p(x_1,x_2)$, the minimization is over $p(u|x_1,x_2,y)=p(u|y)$, and the last  maximisation is over $p(q|x_1,x_2,y,u)=p(q|x_1,x_2)$. \end{remark}

\begin{remark}
The last term of the minimum  in \eqref{21} may be re-written as
\begin{align}
C_1+C_2-I(X_1;X_2|Q)+I(X_1;X_2|UQ).\label{diffinf}
\end{align}
Observe that the difference of mutual information terms in \eqref{diffinf} also appears in the Hekstra-Willems dependence balance bound~\cite{HekstraWillems89}.
\end{remark}

\begin{remark}
\label{uptighter}
The upper bound given in Theorem \ref{theoremupp} is tighter than Theorem \ref{thmupp} (see Appendix \ref{tighter}). We show  through the examples of Sections \ref{exGauss} and \ref{exAdd} that  Theorem \ref{theoremupp} can strictly improve on Theorem~\ref{thmupp}.
\end{remark}

\begin{proof}[Proof of Theorem \ref{theoremupp}] 
We start  with the multi-letter bound in \eqref{disregardmulti}-\eqref{sumwith}. We use the identity in \eqref{identityI} to expand inequality \eqref{disregardmulti} for \textit{any random sequence $U^n$} as follows:
\begin{align}
nR&\leq nC_1+nC_2 -I(X_1^n;X_2^n)\nonumber\\
&= nC_1+nC_2 -I(X_1^nX_2^n;U^n)+I(X_2^n;U^n|X_1^n)+I(X_1^n;U^n|X_2^n)-I(X_1^n;X_2^n|U^n).
\end{align}
In particular, we choose $U^n$ to be such that each symbol $U_i$ is the output of the channel $p_{U|Y}(u_i|y_i)$ with input $y_i$, $i=1,\ldots,n$. Thus, we have the functional dependence graph (FDG) depicted in Fig. \ref{dg}. Furthermore, we have
\begin{align}
nR
&\stackrel{(a)}{\leq} nC_1+nC_2 -\sum_i \left[I(X_{1i}X_{2i};U_i|U^{i-1})+I(X_{2i};U_i|U^{i-1}X_{1i})+I(X_{1i};U_i|U^{i-1}X_{2i})\right]\nonumber\\
&= nC_1+nC_2 -n I(X_{1I}X_{2I};U_I|U^{I-1}I)+nI(X_{2I};U_I|U^{I-1}X_{1I}I)+nI(X_{1I};U_I|U^{I-1}X_{2I}I)\nonumber\\
&\stackrel{(b)}{=} nC_1+nC_2 -n I(X_{1I}X_{2I};U_I|Q)+nI(X_{2I};U_I|X_{1I}Q)+nI(X_{1I};U_I|X_{2I}Q).\label{boundminusmi}
\end{align}
Step $(a)$ follows because $U_i-X_{1i}X_{2i}U^{i-1}-X_1^nX_2^n$ forms a Markov chain (see \eqref{Markovchain}). Step $(b)$ follows by defining $Q=U^{I-1}I$.

\input{dg}

%Using a conditional version of \eqref{identityI}, i.e.,  $$I(X_{1I};X_{2I}|Q)=I(X_{1I}X_{2I};U_I|Q)-I(X_{1I};U_I|X_{2I}Q)-I(X_{2I};U_I|X_{1I}Q)+I(X_{1I};X_{2I}|U_IQ),$$  \eqref{boundminusmi} becomes
%\begin{align}
%nR&\leq nC_1+nC_2-nI(X_{1I};X_{2I}|Q)+nI(X_{1I};X_{2I}|U_IQ).\label{18b}
%\end{align}

We single-letterize \eqref{secondineq}-\eqref{sumwith} next:
\begin{align}
nR&\leq nC_1+I(X_2^n;Y^n|X_1^n)\nonumber\\
&\leq nC_1+\sum_{i=1}^nI(X_{2i};Y_i|X_{1i}Y^{i-1})\nonumber\\
&\stackrel{(a)}{=} nC_1+\sum_{i=1}^nI(X_{2i};Y_i|X_{1i}Y^{i-1}U^{i-1})\nonumber\\
&\leq nC_1+\sum_{i=1}^nI(X_{2i};Y_i|X_{1i}U^{i-1})\nonumber\\
&= nC_1+nI(X_{2I};Y_I|X_{1I}U^{I-1}I)\nonumber\\
&= nC_1+nI(X_{2I};Y_I|X_{1I}Q).\label{20b}
\end{align}
Step $(a)$ follows because $X_{1i}X_{2i}Y_i-Y^{i-1}-U^{i-1}$ forms a Markov chain. 
Similarly, we have
\begin{align}
R&\leq  C_2+I(X_{1I};Y_I|X_{2I}Q)\label{19b}\\
R&\leq I(X_{1I}X_{2I};Y_I|Q).\label{21b}
\end{align}
We  further have % that $p_{X_{1I}X_{2I}Y_IU_IQ}(x_1,x_2,y,u,q)$ factors as follows:% (see Appendix~\ref{factors}):
\begin{align}
p_{X_{1I}X_{2I}Y_IU_IQ}(x_1,x_2,y,u,q)&=p_{X_{1I}X_{2I}}(x_{1},x_{2})p_{Q|X_{1I}X_{2I}}(q|x_1,x_2)p_{Y|X_{1}X_{2}}(y|x_{1},x_{2})p_{U|Y}(u|y).\label{pfactors}
\end{align}
Renaming $(X_{1I},X_{2I},Y_I,U_I,Q)$ as $(X_{1},X_{2},Y,U,Q)$ concludes the proof of Theorem \ref{theoremupp}.
\end{proof}
%, we conclude that $R$ is achievable only if there exists a pmf $p(x_{1},x_{2})$ for which for all auxiliary channels $p(u|y)$ there exists $p(q|x_1,x_2)$ such that 
%\begin{align}
%R\leq  \left\{\begin{array}{l}
%C_1+C_2,\\
%C_1+I(X_2;Y|X_1Q),\\
%C_2+I(X_1;Y|X_2Q),\\
%I(X_1X_2;Y|Q),\\
%C_1+C_2-I(X_1X_2;U|Q)+I(X_2;U|X_1Q)+I(X_1;U|X_2Q)\end{array}\right\}.
%\end{align}
%This concludes the proof of Theorem \ref{theoremupp}.

%We study Theorem \ref{theoremupp} in Sections \ref{exGauss} and \ref{exAdd} for Gaussian MACs and binary adder MACs. In both examples, this upper bound strictly improves the cut-set bound and is tight for wide ranges of bit-pipe capacities.
\begin{remark}
Note that  $Q$ is defined based on $U$. That is, $p(q|x_1,x_2)$ could be a function of $p(u|y)$ and we cannot necessarily change the order in which we minimize over $p(u|y)$ and maximize over $p(q|x_1,x_2)$.
\end{remark}

\section{The Gaussian MAC}
\label{exGauss}
The output of the Gaussian MAC is 
$$Y=X_1+X_2+Z$$ where $Z\sim\mathcal{N}(0,1)$ and the transmitters have average block power constraints $P_1,P_2$; i.e., we have
\begin{align}
\frac{1}{n}\sum_{i=1}^n\mathbb{E}(X_{1,i}^2)\leq P_1\\
\frac{1}{n}\sum_{i=1}^n\mathbb{E}(X_{2,i}^2)\leq P_2. 
\end{align}
When $C_1=C_2=C$ and $P_1=P_2=P$, we call the network symmetric.
%\subsection{The upper bound}

To find a lower bound on the maximum achievable rate, we use Theorem \ref{lower bound}. We choose $(U,X_1,X_2)$ to be jointly Gaussian with zero mean and covariance matrix $\mathbf{K}_{UX_1X_2}$. A special case is when $U$ is null and $(X_1,X_2)$ is jointly Gaussian with the correlation coefficient $\rho$. The rates that satisfy the following constraints for some $\rho$, $0\leq \rho\leq 1$, are thus achievable.\begin{align}
R\leq&C_1+C_2-\frac{1}{2}\log\frac{1}{1-\rho^2}\label{eqlow1}\\
R\leq& C_1+\frac{1}{2}\log\left(1+P_2\left(1-\rho^2\right)\right)\\
R\leq& C_2+\frac{1}{2}\log\left(1+P_1\left(1-\rho^2\right)\right)\\
R\leq& \frac{1}{2}\log\left(1+P_1+P_2+2\sqrt{P_1P_2}\rho\right)\label{eqlow4}
\end{align}
This choice of $(U,X_1,X_2)$ is not optimal in general. 
For example when $C_1$ and $C_2$ are large (i.e., $C_1,C_2>\frac{1}{2}\log(1+P_1+P_2+2\sqrt{P_1P_2})$),
the rate $$R=\frac{1}{2}\log(1+P_1+P_2+2\sqrt{P_1P_2})$$ is not achievable 
by \eqref{eqlow1}-\eqref{eqlow4} but is achievable by Theorem \ref{lower bound} if we choose 
$(U,X_1,X_2)$ to be jointly Gaussian and such that 
$\frac{U}{\sqrt{P_1}}=\frac{X_1}{\sqrt{P_1}}=\frac{X_2}{\sqrt{P_2}}\sim\mathcal{N}(0,1)$. Theorem 
\ref{lower bound} therefore gives a strictly larger lower bound compared to \cite[Theorem 1]{TraskovKramer07}, \cite[Theorem 2]{KangLiu11}. More interestingly, in certain regimes of $C_1,C_2$ the optimal $(U,X_1,X_2)$ is not  jointly Gaussian.

Fig. \ref{closeup} shows the lower bound as a function of $C$ for a symmetric network with $P=1$. The  dotted curve in Fig. \ref{closeup} shows the rates achieved  using the scheme of Section \ref{secach} with jointly Gaussian random variables $(U,X_1,X_2)$ (see \cite[Fig. 2]{KangLiu11} and also \cite[Fig. 4]{KangLiuChong15}).  It is interesting that the obtained lower bound is not concave in $C$. This does not contradict Proposition~\ref{thmconcave} because  Gaussian distributions are sub-optimal. 
The improved solid curve shows rates that are achievable using a mixture of two Gaussian distributions. These rates are slightly larger than the rates achieved by time-sharing between two Gaussian distributions with powers $P_1=P_2=1$. If one permits both time-sharing and power control, then one achieves similar rates as for mixture distributions.

Theorems \ref{thmupp} and  \ref{theoremupp} give upper bounds on the capacity. From Remark \ref{uptighter}, Theorem \ref{theoremupp} is stronger than Theorem \ref{thmupp}. Nevertheless, the bound in Theorem \ref{thmupp} is simpler to evaluate analytically because we can use the maximum entropy lemma to bound all terms. We study both bounds for the Gaussian MAC.

First, we find an upper bound using Theorem \ref{thmupp}. We choose $U=Y+Z^\prime$, where $Z^\prime$ is Gaussian noise with zero mean and variance $N$ (to be optimized later).  The constrains in \eqref{thmupp15} are written as follows using maximum entropy lemmas:\allowdisplaybreaks
\begin{align}
&\hspace{-1.25cm}R\leq C_1+C_2\label{gauss1}\\
&\hspace{-1.25cm}R\leq C_1+\frac{1}{2}\log \left(1+P_2(1-\rho^2)\right)\\
&\hspace{-1.25cm}R\leq C_2+\frac{1}{2}\log \left(1+P_1(1-\rho^2)\right)\\
&\hspace{-1.25cm}R\leq \frac{1}{2}\log \left(1+P_1+P_2+2\rho\sqrt{P_1P_2}\right)\\
&\hspace{-1.25cm}2R\stackrel{(a)}{\leq}    C_1 + C_2 + \frac{1}{2}\log \left( {1 + P_1 + P_2 + 2\rho\sqrt{P_1P_2}}\right)+\frac{1}{2}\log   \left( \frac{(1+N+P_1(1-\rho^2))(1+N+P_2(1-\rho^2))}{(1+N+P_1+P_2+2\rho\sqrt{P_1P_2})(1+N)} \right) .\label{gauss5}
\end{align}
To obtain inequality $(a)$ above, write the last constraint of \eqref{thmupp15} as 
\begin{align}
%\begin{array}{l}
2R\leq C_1+C_2+h(Y|U)-h(YU|X_1X_2)+h(U|X_1)+h(U|X_2)-h(U|X_1X_2).
%\end{array}
\end{align} 
The negative terms are easy to calculate because of the Gaussian nature of the channel and the choice of $U$. The positive terms are bounded from above using the conditional version of the maximum entropy lemma \cite{Thomas87}. It remains to solve a max-min problem (max over $\rho$ and min over $N$). So the rate $R$ is achievable only if there exists some $\rho\geq0$ for which for every $N\geq0$ inequalities \eqref{gauss1}-\eqref{gauss5} hold. 

\begin{figure}[t!]
\centering
\input{Gaussianplot}
\caption{Upper and lower bounds on $R$ as functions of $C$ for the Gaussian MAC with $P_1=P_2=1$.}
\label{closeup}
\end{figure}

We choose $N$ to be (see \cite[eqn. (21)]{KangLiu11})
\begin{align}N=\left(\sqrt{P_1P_2}\left(\frac{1}{\rho}-\rho\right)-1\right)^+.\label{choiceofN}
\end{align} 
Let us first motivate this choice. From \eqref{equivthm6}, the inequality in \eqref{gauss5} is  
\begin{align}
\label{newgauss5}
2R\leq C_1+C_2+I(X_1X_2;Y)-I(X_1;X_2)+I(X_1;X_2|U)
\end{align} 
evaluated for the joint Gaussian distribution $p(x_1,x_2)$ with covariance matrix 
\begin{align}
\left[\begin{array}{cc}P_1&\rho\sqrt{P_1P_2}\\\rho\sqrt{P_1P_2}&P_2\end{array}\right].
\end{align} 
The choice \eqref{choiceofN} makes $U$ satisfy the Markov chain $X_1-U-X_2$ for the regime where 
\begin{align}
\sqrt{P_1P_2}\left(\frac{1}{\rho}-\rho\right)-1\geq0\label{regimeofN}
\end{align} and thus minimizes the RHS of \eqref{newgauss5}.  Otherwise, we choose $U=Y$ which results in a redundant bound.
The resulting upper bound is summarized in Corollary \ref{uppGauss}.
\begin{corollary}
\label{uppGauss}
Rate $R$ is achievable only if there are $\rho\geq0$ such that 
\begin{align}
&\rho\leq \sqrt{1+\frac{1}{4P_1P_2}}-\sqrt{\frac{1}{4P_1P_2}},\label{70meaningRHS}\\
&R\leq C_1+C_2\label{cor00}\\
&R\leq C_2+\frac{1}{2}\log\left(1+P_1(1-\rho^2)\right)\label{cor01}\\
&R\leq C_1+\frac{1}{2}\log\left(1+P_2(1-\rho^2)\right)\label{cor0}\\
&R\leq \frac{1}{2}\log \left(1+P_1+P_2+2\rho\sqrt{P_1P_2}\right)\label{cor1}\\
&2R\leq C_1+C_2+\frac{1}{2}\log \left(1 + P_1 + P_2 + 2\rho\sqrt{P_1P_2}\right)-\frac{1}{2}\log\left(\frac{1}{1-\rho^2}\right),\label{cor2}\\
&\hspace{-.65cm}\text{or}\nonumber\\
&\sqrt{1+\frac{1}{4P_1P_2}}-\sqrt{\frac{1}{4P_1P_2}}\leq\rho\leq 1,\label{70meaningLHS}\\
&R\leq C_1+C_2\\
&R\leq C_2+\frac{1}{2}\log\left(1+P_1(1-\rho^2)\right)\\
&R\leq C_1+\frac{1}{2}\log\left(1+P_2(1-\rho^2)\right)\\
&R\leq \frac{1}{2}\log \left(1+P_1+P_2+2\rho\sqrt{P_1P_2}\right).
\end{align}
\end{corollary}

%\subsection{The lower bound}

The above upper bound is plotted in Fig. \ref{closeup} for different values of $C$ and for $P=1$.
For symmetric diamond networks, we specify a regime of $C$ for which the above upper bound meets the lower bound in Theorem \ref{lower bound} and thus characterizes the capacity. This is summarized in Theorem \ref{gaussmatch} and its proof is deferred to Appendix \ref{messyproof}.

% \begin{figure}[ht!]
% \begin{center}
%\includegraphics[width=.45\textwidth]{figureGaussianDiamond.pdf}
%\end{center}
%\caption{???}
%\label{figGaussianDiamond}
%\end{figure}
\begin{theorem}
\label{gaussmatch}
For a symmetric Gaussian diamond network with orthogonal broadcast links, the upper bound in Theorem \ref{thmupp} is tight if $C\leq \frac{1}{4}\log(1+2P)$, $C\geq \frac{1}{2}\log(1+4P)$, or 
\begin{align}
\label{thmC}
\frac{1}{4}\log\frac{1+2P(1+\rho^{(1)})}{1-{\left(\rho^{(1)}\right)}^2}\leq C\leq \frac{1}{4}\log\frac{1+2P(1+\rho^{(2)})}{1-{\left({\rho^{(2)}}\right)}^2}
\end{align}
where 
\begin{align}
&\rho^{(1)}=\frac{-(1+2P)+\sqrt{12P^2+(1+2P)^2}}{6P}\label{r1}\\
&\rho^{(2)}=\sqrt{1+\frac{1}{4P^2}}-\frac{1}{2P}.\label{r2}
\end{align}
\end{theorem}
\begin{remark}
The $\rho^{(1)}$ given in \eqref{r1} maximizes the RHS of \eqref{cor2}.  The $\rho^{(2)}$ given in \eqref{r2} is the solution of \eqref{regimeofN} with equality. Note that $\rho^{(2)}$ forms the RHS of \eqref{70meaningRHS} and the LHS of \eqref{70meaningLHS}. In other words, for $\rho\leq \rho^{(2)}$ one can find $U$ as a degraded version of $Y$ such that $X_1-U-X_2$ forms a Markov chain. This is not possible for $\rho>\rho^{(2)}$.
\end{remark}
\begin{remark}
\label{regimestightthm2}
For $C\leq \frac{1}{4}\log(1+2P)$, the capacity is equal to $2C$ and is achieved by \eqref{eqlow1}-\eqref{eqlow4} with $\rho=0$ (no cooperation among the relays).  In the regime \eqref{thmC}, the capacity is given by \eqref{eqlow1}-\eqref{eqlow4} with partial cooperation among the relays.  For $C\geq \frac{1}{2}\log(1+4P)$, the capacity is equal to $\frac{1}{2}\log(1+4P)$ and is achieved using Theorem \ref{lower bound} with $X_1=X_2=U\sim \mathcal{N}(0,P)$ (full cooperation among the relays). 
\end{remark}

\begin{remark}
The bound in Corollary \ref{uppGauss} and the bound in \cite[Theorem 1]{KangLiu11} are closely related. 
The  bound in \cite[Theorem 1]{KangLiu11} is tighter than Corollary \ref{uppGauss} in certain regimes of operation. We will see that Theorem \ref{theoremupp} strengthens Corollary \ref{uppGauss} and  is in general tighter than \cite[Theorem 1]{KangLiu11}. % \textcolor{red}{Do you think I should explicitly give the regimes where they are the same?} 
\end{remark}
Based on Theorem \ref{gaussmatch}, the upper and lower bounds match in Fig. \ref{closeup} (where $P=1$) for $C\leq 0.3962$, $0.4807\leq C \leq 0.6942$, and $C\geq1.1610$. Theorem \ref{theoremupp} tightens the above upper bound as we show next. We again choose $U=Y+Z^\prime$ where $Z^\prime$ is a Gaussian random variable with  zero mean and variance $N$ (to be optimized). In contrast to Theorem \ref{thmupp}, it is not  clear whether Gaussian distributions are optimal in Theorem \ref{theoremupp}. To compute the bound in Theorem \ref{theoremupp}, we proceed as follows.

The first four bounds of \eqref{21} may be loosened by dropping the time-sharing random variable $Q$ and using the maximum entropy lemma:

\begin{align}
&R\leq C_1+C_2\label{66}\\
&R\leq C_1+I(X_2;Y|X_1Q)\leq C_1+\frac{1}{2}\log\left(1+P_2\left(1-\rho^2\right)\right)\label{67}\\
&R\leq  C_2+I(X_1;Y|X_2Q)\leq C_2+\frac{1}{2}\log\left(1+P_1\left(1-\rho^2\right)\right)\label{68}\\
&R\leq  I(X_1X_2;Y|Q) \leq \frac{1}{2}\log\left(1+P_1+P_2+2\sqrt{P_1P_2}\rho\right).\label{69}
\end{align}
To bound the last constraint in \eqref{21}, we use both the entropy power inequality \cite[Theorem 17.7.3]{CoverThomas} and the maximum entropy lemma:
\begin{align}
R\leq& C_1+C_2-I(X_1X_2;U|Q)+I(X_1;U|X_2Q)+I(X_2;U|X_1Q)\nonumber\\
=&C_1+C_2-h(U|Q)-h(U|X_1X_2)+h(U|X_2Q)+h(U|X_1Q)\nonumber\\
\leq&C_1+C_2-h(U|Q)-h(U|X_1X_2)+h(U|X_2)+h(U|X_1)\nonumber\\
\stackrel{(a)}{\leq}& C_1+C_2-\frac{1}{2}\log\left(2\pi eN+2^{2h(Y|Q)}\right)-h(U|X_1X_2)+h(U|X_2)+h(U|X_1)\nonumber\\
\stackrel{(b)}{\leq}& C_1+C_2-\frac{1}{2}\log\left(2\pi eN+2^{2h(Y|Q)}\right)-\frac{1}{2}\log\left(2\pi e (1+N)\right)\nonumber\\&+\frac{1}{2}\log\left( 2\pi e\left(1+N+P_1\left(1-\rho^2\right)\right)\right)+\frac{1}{2}\log\left( 2\pi e\left(1+N+P_2\left(1-\rho^2\right)\right)\right)\label{87cont}
\end{align}
where $(a)$ holds by the entropy power inequality and $(b)$ holds by the maximum entropy lemma.
We now use $R\leq I(X_1X_2;Y|Q)$ to write
\begin{align}
h(Y|Q)&=\frac{1}{2}\log(2\pi e)+I(X_1X_2;Y|Q)\nonumber\\
&\geq\frac{1}{2}\log(2\pi e)+R.\label{88cont}
\end{align}
From \eqref{87cont} and \eqref{88cont} we obtain
\begin{align}
R &\leq C_1+C_2-\frac{1}{2}\log\left(N+2^{2R}\right)-\frac{1}{2}\log\left(1+N\right)+\frac{1}{2}\log\left( 1+N+P_1\left(1-\rho^2\right)\right)+\frac{1}{2}\log\left(1+N+P_2\left(1-\rho^2\right)\right).\label{21N}
\end{align}

\begin{remark}
The above argument is similar to the argument used in \cite{Ozarow80}, and it is also  related to \cite[Section X]{TandonUlukus11}.
\end{remark}

\begin{remark}
Expression \eqref{21N} may be re-written as
%\begin{align}
%2^{2R}\left(N+2^{2R}\right)\leq 2^{2(C_1+C_2)}\frac{\left( 1+N+P_1\left(1-\rho^2\right)\right)\left( 1+N+P_2\left(1-\rho^2\right)\right)}{(1+N)}
%\end{align}
%or equivalently as
\begin{align}
R&\leq \frac{1}{2}\log\frac{-N+\sqrt{N^2+2^{2(C_1+C_2+1)}\frac{\left(1+N+P_1\left(1-\rho^2\right)\right)\left(1+N+P_2\left(1-\rho^2\right)\right)}{1+N}}}{2}.\label{21Nmin}
\end{align}
\end{remark}
Recall that \eqref{21Nmin} holds for any value of $N\geq 0$. We choose $N$ as a function of $\rho$ to minimize the RHS of \eqref{21Nmin}.  It remains to maximize over $\rho$ and find  the maximum rate $R$ admissible by \eqref{66}-\eqref{69}, \eqref{21Nmin}. We solve this optimization problem numerically for the symmetric Gaussian network with $P=1$,  and plot the resulting upper bound in Fig. \ref{closeup}. Note that the upper bound of Theorem \ref{theoremupp} is strictly tighter than Theorem \ref{thmupp} for $0.3962< C< 0.4807$. Furthermore, from the numerical evaluation of the bound, the upper bound of Theorem \ref{theoremupp} is tight for $C\leq 0.6942$ and $C\geq 1.1610$. This is made precise for symmetric Gaussian networks in the following theorem which we prove in Appendix \ref{apgaussmatch2}.
\begin{theorem}
\label{gaussmatch2}
For a symmetric Gaussian diamond network, the upper bound in Theorem \ref{theoremupp} meets the lower bound in Theorem~\ref{lower bound} for all $C$ such that $C\geq \frac{1}{2}\log(1+4P)$, or 
\begin{align}
\label{thmC2}
C\leq \frac{1}{4}\log\frac{1+2P(1+\rho^{(2)})}{1-\left(\rho^{(2)}\right)^2}
\end{align}
where 
\begin{align}
&\rho^{(2)}=\sqrt{1+\frac{1}{4P^2}}-\frac{1}{2P}.\label{ro2}
\end{align}
\end{theorem}

\begin{proof}[Sketch of proof]
The regime $C\geq \frac{1}{2}\log(1+4P)$ is addressed in Remark \ref{regimestightthm2}. We briefly outline the proof for the regime in \eqref{thmC2}.  Consider the lower bound in \eqref{eqlow1}-\eqref{eqlow4} and let $R_{\max}^{(l)}$ be the maximum achievable rate. This lower bound meets the cut-set bound (and is thus tight) unless \eqref{eqlow1} and \eqref{eqlow4} are both active in which case we have 
\begin{align}
R^{(l)}_{\max}=C_1+C_2-\frac{1}{2}\log\frac{1}{1-\lambda^2}=\frac{1}{2}\log\left(1+P_1+P_2+2\lambda\sqrt{P_1P_2}\right)\label{regimematch}
\end{align}
where $\lambda$ is the optimal correlation coefficient in \eqref{eqlow1}-\eqref{eqlow4}. We show in Appendix~\ref{apgaussmatch2} that the upper bound given by \eqref{66}-\eqref{69}, \eqref{21N} meets the lower bound $R^{(l)}_{\max}$ when we have \eqref{regimematch} and $\lambda\leq \rho^{(2)}$. One can check for  symmetric networks that $\lambda\leq \rho^{(2)}$ if and only if \eqref{thmC2} is satisfied.
\end{proof}

{
More generally, we have the following result for asymmetric networks. This is addressed in Remark \ref{generalgaussmatch2} in Appendix \ref{apgaussmatch2}.
\begin{theorem}
\label{asymetricGauss}
The upper bound in Theorem \ref{theoremupp} meets the lower bound in Theorem~\ref{lower bound} if any of the following conditions hold:
\begin{align}
%&C_1\leq \frac{1}{2}\log\left(\frac{1+P_1+P_2}{1+P_2}\right)\\
%&C_2\leq \frac{1}{2}\log\left(\frac{1+P_1+P_2}{1+P_1}\right)\\
&C_1+C_2\leq\frac{1}{2}\log\left(\frac{1+P_1+P_2+2\rho^{(2)}\sqrt{P_1P_2}}{1-\left(\rho^{(2)}\right)^2}\right)\label{PP0}\\
&C_1\leq \frac{1}{2}\log\left(\frac{1+P_1+P_2+2\rho_0\sqrt{P_1P_2}}{1+P_2(1-\rho_0^2)}\right)\label{PP1}\\
&C_2\leq \frac{1}{2}\log\left(\frac{1+P_1+P_2+2\rho_0\sqrt{P_1P_2}}{1+P_1(1-\rho_0^2)}\right)\label{PP2}\\
&\min(C_1,C_2)\geq \frac{1}{2}\log\left(1+P_1+P_2+2\sqrt{P_1P_2}\right)\label{PP}
\end{align}
where $\rho^{(2)}$ is given by \eqref{ro2} and $\rho_0$ is given by 
$$\rho_0=\frac{-\sqrt{P_1P_2}+\sqrt{P_1P_2+2^{2(C_1+C_2)}\left(2^{2(C_1+C_2)}-1-P_1-P_2\right)}}{2^{2(C_1+C_2)}}.$$
\end{theorem}
\begin{remark}
$\rho_0$ is defined such that $C_1+C_2=\frac{1}{2}\log\left(\frac{1+P_1+P_2+2\rho_0\sqrt{P_1P_2}}{1-\rho_0^2}\right)$. Note that we have $\rho_0\leq \rho^{(2)}$ if and only if \eqref{PP0} is satisfied. In defining $\rho_0$, we have implicitly assumed that $C_1+C_2\geq\frac{1}{2}\log\left(1+P_1+P_2\right)$; this is without loss of generality because otherwise $C_1,C_2$ are in the regime defined by \eqref{PP0}.
\end{remark}
\begin{remark}
In the regime \eqref{PP}, the cut-set bound is achievable using Theorem \ref{lower bound} with $\frac{U}{\sqrt{P_1}}=\frac{X_1}{\sqrt{P_1}}=\frac{X_2}{\sqrt{P_2}}\sim\mathcal{N}(0,1)$ and the lower bound in \eqref{eqlow1}-\eqref{eqlow4}  is loose.
\end{remark}
\begin{remark}
Theorem \ref{asymetricGauss} reduces to Theorem \ref{gaussmatch2} when $P_1=P_2=P$ and $C_1=C_2=C$.%, see Remark \ref{} in Appendix \ref{}.
\end{remark}
\begin{remark}
Theorem \ref{theoremupp} is strictly tighter than \cite[Theorem 1]{KangLiu11} and \cite[Theorem 1]{KangLiuChong15}. The regime of interest is given by \eqref{regimematch} because otherwise both upper bounds reduce to the cut-set bound which is tight. First  suppose $\lambda> \rho^{(2)}$. In this case,  \cite[Theorem 1]{KangLiu11} reduces to the cut-set bound and is larger than or equal to the upper bound of Theorem \ref{theoremupp}.
Next suppose $\lambda\leq \rho^{(2)}$. Here, the upper bound given by \eqref{66}-\eqref{69}, \eqref{21N} can be shown to be equal to $R_{\max}^{(l)}$ and is thus tight but \cite[Theorem 1]{KangLiu11} may not be tight, see \cite[Theorem 3]{KangLiuChong15}. For example, when $P_1=P_2=0.25$ and $C_1=C_2=0.15$ Theorem~\ref{theoremupp} gives $C^{\diamond}\leq .2994$ (which is tight) whereas \cite[Theorem 1]{KangLiu11} gives $C^{\diamond}\leq 0.3$. The looseness of \cite[Theorem 1]{KangLiu11} in comparison to our result seems to be due to the relaxation of \cite[eqn. (28)]{KangLiu11} in the final theorem statement of \cite[Theorem 1]{KangLiu11}.
\end{remark}

%\begin{remark}
%Following the proof steps of \cite[Theorem 1]{KangLiu11}, one sees that Kang and Liu has, in effect, proved the following statement that is stronger than \cite[Theorem 1]{KangLiu11}: If $R$ is achievable and $R\geq \frac{1}{2}\log\left(1+P_1+P_2\right)$, then there exists $\rho$,
%\begin{align}
%\rho=\frac{1}{2\sqrt{P_1P_2}}\left(2^{2R}-1-P_1-P_2\right),\label{rhokang}
%\end{align}
%such that \eqref{eqlow1}-\eqref{eqlow4} are satisfied if $\rho\leq \rho^{(2)}$,  and \eqref{66}-\eqref{69} are satisfied otherwise. If $R$ is achievable and $R< \frac{1}{2}\log\left(1+P_1+P_2\right)$, then \eqref{eqlow1}-\eqref{eqlow4} are satisfied with $\rho=0$. For example, set $P_1=P_2=0.25$ and $C_1=C_2=0.15$. We have $\rho^{(2)}\approx0.2361$. $R=0.3$ is not achievable because we have $\rho\approx .0314\leq \rho^{(2)}$ and \eqref{eqlow1} is violated. However, this rate is admissible by \cite[Theorem 1]{KangLiu11}. This looseness comes in because \cite[eqn. (28)]{KangLiu11} is relaxed in the statement of \cite[Theorem 1]{KangLiu11}. 
%\end{remark}
}

\section{The binary Adder Channel}
\label{exAdd}

Consider the binary adder channel  defined by $\mathcal{X}_1=\{0,1\}$, $\mathcal{X}_2=\{0,1\}$, $\mathcal{Y}=\{0,1,2\}$, and $Y=X_1+X_2$. Suppose without loss of generality that $C_1\leq C_2$.  When $C_1=C_2=C$, we call the network symmetric. The best known upper bound for this channel is the cut-set bound and the best known lower bound is given by \cite[Theorem 1]{TraskovKramer07}. More precisely, using a doubly symmetric input distribution,  $R$ is achievable if it satisfies the following inequalities for some $p$, $0\leq p\leq 1$.
\begin{align}
R&\leq C_1+C_2 - 1+ h_2(p) \label{eqlowbin1}\\
R&\leq C_1+ h_2(p)\label{loose}\\
R&\leq h_2(p) + 1 - p\label{eqlowbin3}
\end{align}
%Note that \eqref{loose} is redundant because it can be obtained from \eqref{eqlowbin1} and \eqref{eqlowbin3}. 
 This lower bound is a special case of Theorem \ref{lower bound} with $U$ a constant. The bound is  plotted in Fig. \ref{figplotbinaryadder} as a function of $C$ for symmetric networks where $C_1=C_2=C$. 
 \begin{figure}[t!]
\centering
\input{binaryadderMAC}
\caption{Upper and lower bounds on $R$ as functions of $C$ for the binary adder MAC.}
\label{figplotbinaryadder}
\end{figure}

We evaluate  Theorems \ref{thmupp} and \ref{theoremupp} to derive new upper bounds on the achievable rate. The obtained upper bounds are plotted in Fig. \ref{figplotbinaryadder}. 
{It turns out that the upper bound of Theorem \ref{theoremupp} meets the lower bound for all ranges of $C$. Theorem \ref{theoremupp} is better than Theorem \ref{thmupp},  but  Theorem \ref{thmupp} is simpler to analyze and gives capacity for $C\leq0.75$ and $C\geq.7929$.}

Consider first Theorem \ref{thmupp}. Let $p(u|y)$ be a symmetric channel as shown in Fig. \ref{channelu} with parameter $\alpha$, $\alpha\leq \frac{1}{2}$, to be optimized. From Theorem~\ref{thmupp},  we must solve a max-min problem (max over $p(x_1,x_2)$, min over $\alpha$). 
%We first loosen the  bound by looking at the min-max problem (rather than the max-min problem). 
For a fixed $\alpha$, the upper bound is concave in $p(x_1,x_2)$ (see Remark \ref{concave}). The concavity together with the symmetry of the problem and the auxiliary channel in $p(x_1,x_2)$ imply the following lemma. We defer the proof to Appendix \ref{detailoptlemma}.

\begin{lemma}
\label{lemmadoubly}
An optimizing pmf $p(x_1,x_2)$ in \eqref{thmupp15} is that of a doubly symmetric binary source.
\end{lemma}

%\begin{remark}
%\label{whyminmax}
%The reason we loosen the max-min problem to a min-max problem is that we would like to constrain the optimization space of Theorem \ref{thmupp} to the set of doubly symmetric binary sources and for that we need $p(u|y)$ to be fixed. 
%\end{remark}

So suppose $p(x_1,x_2)$ is a doubly symmetric binary source with parameter $p$. The upper bound in Theorem~\ref{thmupp} with $p(u|y)$ in Fig. \ref{channelu} reduces to
\begin{align}
\max_{0\leq p\leq \frac{1}{2}}  \min_{0\leq \alpha\leq \frac{1}{2}} \min\left\{\begin{array}{l}
C_1+C_2,\\
C_1 + h_2(p),\\
h_2(p) + 1 - p,\\
\frac{C_1+C_2}{2} + h_2(p) - \frac{p}{2}+\frac{1}{2}I(X_1;X_2|U)
\end{array}\right\}
\label{RHSi}
\end{align}
where the last term of \eqref{RHSi} is written using \eqref{equivthm6}, and where the range of $p$ is $[0,\frac{1}{2}]$. Note that $\alpha$ is implicit in $I(X_1;X_2|U)$:% since
\begin{eqnarray}
\label{mutualinfoa}
I(X_1;X_2|U)=2h_2\left(\alpha\star\frac{p}{2}\right)-(1-p)h_2(\alpha)-h_2(p)-p.\label{explicit}
\end{eqnarray}
Here the $\star$ operator is defined by $\alpha\star\beta=\alpha(1-\beta)+\beta(1-\alpha)$, $\beta\leq 1$. To obtain the best bound, we choose $U$ such that $X_1-U-X_2$ forms a Markov chain.
%Such a choice depends on $p$ and is not permitted since we are solving a $\min-\max$ problem (see Remark \ref{whyminmax}). 
This requires 
\begin{align}
\label{choicea}
\alpha(1-\alpha)=\left(\frac{p}{2(1-p)}\right)^2
\end{align}
%where $p^\star$ is the $p$ that maximizes
%\begin{align}
%\min\left\{   
%2C,
%C+h_2(p),
%h_2(p)+1-p,
%C+h_2(p)-\frac{p}{2}
%  \right\}.\label{upp}
%\end{align}
which has a solution for $\alpha$ because $p\leq \frac{1}{2}$.
%, the choice \eqref{choicea} is valid. The choice in \eqref{choicea} makes $X_1-U-X_2$ form a Markov chain at $p=p^\star$. The motivation behind this choice is explained in detail in Appendix \ref{apmotivation} where the following corollary is proved.
\input{channelu}

\begin{corollary}
\label{corbinaryadder}
Rate $R$ is achievable only if there is some $p$, $0\leq p\leq\frac{1}{2}$, such that
\begin{align}
R&\leq C_1+C_2\label{2Cbound}\\
R&\leq C_1 + h_2(p)\\
R&\leq h_2(p) + 1 - p\\
R&\leq \frac{C_1+C_2}{2} + h_2(p) - \frac{p}{2}.
\end{align}
\end{corollary} 

We compare Corollary \ref{corbinaryadder} with the lower bound in \eqref{eqlowbin1}-\eqref{eqlowbin3} for symmetric networks and find the capacity  for some ranges of $C$. This is summarized in Theorem \ref{ranges} and the proof is deferred to Appendix \ref{apranges}.

\begin{theorem}
\label{ranges}
The upper bound in Theorem \ref{thmupp} meets the lower bound in Theorem \ref{lower bound} for the symmetric diamond network with a binary adder channel if $C\leq .75$ or
\begin{eqnarray}
\label{thmCbinary}
C\geq 1-\frac{p^{(1)}}{2}\approx 0.7929
\end{eqnarray}
where $p^{(1)}=\frac{1}{1+\sqrt{2}}\approx 0.4142$.
\end{theorem}

In the rest of this section, we show that Theorem \ref{theoremupp} gives the capacity of the diamond network with a binary adder MAC for all ranges of $C_1,C_2$. We first state a generalization of Mrs. Gerber's Lemma \cite{WynerZiv73} that we prove in Appendix \ref{concavity lemma}. For other generalizations, please see
\cite{Witsenhausen74, AhlswedeKorner77, ChayatShamai89, ShamaiWyner90, HarremosVignat03, SharmaDasMuthukrishnan11, JohnsonYu10, HaghighatshoarAbbeTelatar13, JogAnantharam14, Cheng14}. Our generalization
is different than previous ones in that it establishes the convexity of a {\em difference} of
entropies, rather than an individual entropy. In this sense, Lemma \ref{GMGL} seems similar to an
extension \cite{LiuViswanath07} of Shannon's entropy power inequality \cite{Shannon}.
\begin{lemma}[Generalization of Mrs. Gerber's Lemma]
\label{GMGL}
The function
 \begin{align}
 g(x,y)=h_2\left(\alpha\star\left(\frac{y}{2}+(1-y)h_2^{-1}\left(\frac{(x-h_2(y))^+}{1-y}\right)\right)\right)-h_2\left(\alpha\star \frac{y}{2}\right)\label{fxy}
 \end{align}
is jointly convex in $x$ and $y$,  $0\leq x\leq 1+h_2(y)-y$, $0\leq y\leq 1$.
We recover Mrs. Gerber's Lemma by choosing $y=0$.
\end{lemma}
{\begin{theorem}
\label{thm:capbinaryadder}
The upper bound of Theorem \ref{theoremupp} matches the lower bound of Theorem \ref{lower bound}; i.e., the capacity $C^\diamond$ of diamond networks with binary adder MACs and $C_1\leq C_2$ is
\begin{align}
C^\diamond=\max_{0\leq p\leq \frac{1}{2}}\min\left\{\begin{array}{l}C_1+C_2-1+h_2(p)\\C_1+h_2(p)\\h_2(p)+1-p.\end{array}\right.
\end{align}
\end{theorem}}
\begin{proof}
\input{channeluyty}
We again use the auxiliary channel $p(u|y)$ depicted in Fig. \ref{channelu}. This channel may  be viewed as the cascade of the channels  $p(\tilde{y}|y)$ and $p(u|\tilde{y})$ shown in Fig. \ref{channeluyty}, where $p(u|\tilde{y})$ is a BSC with cross over probability $\alpha$.
Define $p_i,q_i$, $i\in\mathcal{Q}$, and $q$ by
\begin{align}
&p_i=p_{Y|Q}(0|i)\\
&q_i=p_{Y|Q}(1|i)\\
&q=p_{Y}(1).\label{qpy1}
\end{align} 
The first four terms of \eqref{21} may be loosened by dropping the time sharing random variable $Q$. We use the symmetry and concavity of those terms in $p(x_1,x_2)$ to write 
\begin{align}
R&\leq C_1+C_2\label{cap1}\\
R&\leq C_1+h_2(q)\\
R&\leq h_2(q)+1-q.\label{cap3}
\end{align}
%
%
%
%
%\begin{align}
%H(U|X_1)&=p_{X_1}(0)H(U|X_1=0)+p_{X_1}(1)H(U|X_1=1)\\
%&=p_{X_1}(0)h_2\left(p_{X_2|X_1}(0|0)(1-\alpha)+p_{X_2|X_1}(1|0)\frac{1}{2}\right)+p_{X_1}(1)h_2\left(p_{X_2|X_1}(1|1)(1-\alpha)+p_{X_2|X_1}(0|1)\frac{1}{2}\right)\\
%&\stackrel{(a)}{\leq} h_2\left((1-q)(1-\alpha)+\frac{1}{2}q\right)\label{fullcap1}
%\end{align}
%where \eqref{fullcap1} follows by Jensen's inequality. Similarly, we have
%\begin{align}
%H(U|X_1)&\leq  h_2\left((1-q)(1-\alpha)+\frac{1}{2}q\right)\label{fullcap2}\\
%H(U|X_1X_2)&=(1-q) h_2(\alpha)+q.\label{fullcap3}\\
%H(Y|\tilde{Y})&\leq h_2(q)
%\end{align}

It remains to bound the last term of \eqref{21}:
\begin{align}
R\leq& C_1+C_2-I(X_1X_2;U|Q)+I(X_2;U|X_1Q)+I(X_1;U|X_2Q)\nonumber\\
=& C_1+C_2-H(U|Q)-H(U|X_1X_2)+H(U|X_1Q)+H(U|X_2Q).\label{rewrote}
\end{align}
We optimize the RHS of \eqref{rewrote} under the constraint
\begin{align}
R\leq I(X_1X_2;Y|Q)=H(Y|Q)\label{cons}
\end{align}
that is imposed by the fourth term of \eqref{21}.
We have $H(U|X_1X2)=(1-q)h_2(\alpha)+q$ and can upper bound both $H(U|X_1Q=i)$ and $H(U|X_1Q=i)$ by $$h_2\left(\alpha\star\frac{q_i}{2}\right)$$ by the concavity of $h_2(.)$ and  symmetry of $p_{U|Y}$ (see Appendix \ref{HUcX1Q}). But how to bound \eqref{rewrote} from above is not obvious because $H(U|Q)$  appears with a negative sign. We start with
 \begin{align}
&H(U|Q=i)=h_2\left(\alpha\star\left(\frac{q_i}{2}+p_i\right)\right)\label{tosimplofy}\\
&H(Y|Q=i)=h_2(q_i)+(1-q_i)h_2\left(\frac{p_i}{1-q_i}\right).\label{sym2}
 \end{align}
Note that both \eqref{tosimplofy} and \eqref{sym2} are symmetric in $p_i$ with respect to $\frac{1-q_i}{2}$. We may therefore choose $p_i\leq \frac{1-q_i}{2}$, find $p_i$ from \eqref{sym2} and insert it into $\eqref{tosimplofy}$ to obtain
\begin{align}
H(U|Q=i)=h_2\left(\alpha\star\left(\frac{q_i}{2}+(1-q_i)h_2^{-1}\left(\frac{H(Y|Q=i)-h_2(q_i)}{1-q_i}\right)\right)\right).
\end{align}

Combining the above bounds and inserting in \eqref{rewrote}, we have
\begin{align}
 R\leq&  C_1+C_2+\sum_{i=1}^{|\mathcal{Q}|}p_Q(i)\left(-h_2\!\left(\alpha\star\!\left(\frac{q_i}{2}+(1-q_i)h_2^{-1}\!\left(\frac{H(Y|Q\!=\!i)\!-\!h_2(q_i)}{1-q_i}\right)\!\right)\right)-(1-q_i) h_2(\alpha)-q_i+2h_2\left(\alpha\star\frac{q_i}{2}\right)\right)\nonumber\\
\stackrel{(a)}{\leq}&C_1+C_2-h_2\left(\alpha\star\left(\frac{q}{2}+(1-q)h_2^{-1}\left(\frac{\left(H(Y|Q)-h_2(q)\right)^+}{1-q}\right)\right)\right)-(1-q) h_2(\alpha)-q+2h_2\left(\alpha\star\frac{q}{2}\right)\nonumber\\
%=&C_1+C_2-h_2\left(\alpha\star\left(\frac{q}{2}+(1-q)h_2^{-1}\left(\frac{\left(I(X_1X_2;Y|Q)-h_2(q)\right)^+}{1-q}\right)\right)\right)-(1-q) h_2(\alpha)-q+2h_2\left(\alpha\star\frac{q}{2}\right)\nonumber\\
\stackrel{(b)}{\leq}&C_1+C_2-h_2\left(\alpha\star\left(\frac{q}{2}+(1-q)h_2^{-1}\left(\min\left(1,\frac{\left(R-h_2(q)\right)^+}{1-q}\right)\right)\right)\right)-(1-q) h_2(\alpha)-q+2h_2\left(\alpha\star\frac{q}{2}\right)\label{cap4}
 \end{align}
 where $(a)$ follows by concavity of the binary entropy function and Lemma \ref{GMGL}, and  $(b)$ follows from \eqref{cons} because $h_2(\alpha\star\left(\frac{q}{2}+(1-q)h_2^{-1}(x)\right))$ is non-decreasing in $x$ for $\alpha\leq \frac{1}{2}$.
%
% \begin{align}
% R&\leq C_1+C_2-I(U;X_1X_2|Q)+I(X_1;U|X_1Q)+I(X_1;U|X_2Q)\nonumber\\
% &=C_1+C_2-H(U|Q)-H(U|X_1X_2)+H(U|X_1Q)+H(U|X_2Q)\nonumber\\
% &\leq C_1+C_2-H(U|Q)-H(U|X_1X_2)+H(U|X_1)+H(U|X_2)\nonumber\\
%  &= C_1+C_2-H(U|Q)-(1-q) h_2(\alpha)-q+H(U|X_1)+H(U|X_2)\nonumber\\
% &\stackrel{(a)}{\leq} C_1+C_2-H(U|Q)-(1-q) h_2(\alpha)-q+2h_2\left((1-q)(1-\alpha)+\frac{1}{2}q\right)\nonumber\\
% &\stackrel{(b)}{\leq}  C_1+C_2-h_2\left(\alpha\star h_2^{-1}\left(H(\tilde{Y}|Q)\right)\right)-(1-q) h_2(\alpha)-q+2h_2\left((1-q)(1-\alpha)+\frac{1}{2}q\right)\label{lastmrs}
% \end{align}
%where $(a)$ follows by the symmetry and concavity of $H(U|X_1)$ and $H(U|X_2)$ in $p(x_1,x_2)$, and $(b)$ follows by Mrs. Gerber's lemma.  Here the $\star$ operator is defined by $\alpha\star\beta=\alpha(1-\beta)+\beta(1-\alpha)$. We next relate $H(\tilde{Y}|Q)$ to $R$:
% \begin{align}
% R&\leq I(X_1X_2;Y|Q)\nonumber\\
% &= I(X_1X_2;Y\tilde{Y}|Q)\nonumber\\
%&= H(\tilde{Y}|Q)+H(Y|\tilde{Y}Q)-H(\tilde{Y}|X_1X_2)\nonumber\\
%&\leq H(\tilde{Y}|Q)+h_2(q)-q.\label{relate}
% \end{align}
%Using \eqref{lastmrs} and \eqref{relate} we obtain
%\begin{align}
% R&\leq C_1+C_2-h_2\left(\alpha\star h_2^{-1}\left(\left(R-h_2(q)+q\right)^+\right)\right)-(1-q) h_2(\alpha)-q+2h_2\left((1-q)(1-\alpha)+\frac{1}{2}q\right).\label{cap4}
% \end{align}
%\begin{remark}
%The proof steps above are very similar to  \eqref{87cont}-\eqref{21N}. This is due to the duality between Mrs. Gerber's lemma for binary random variables and the entropy power inequality for continuous random variables \cite{ShamaiWyner90}.
%\end{remark}
{Choosing $\alpha$ appropriately, we show in Appendix \ref{last theorem} that the upper bound in  \eqref{cap1}-\eqref{cap3} and \eqref{cap4} matches the lower bound in \eqref{eqlowbin1}-\eqref{eqlowbin3}.} %We thus find the capacity of diamond networks with binary adder MACs.}

\end{proof}
\begin{remark}
One may use Mrs. Gerber's lemma \cite{WynerZiv73} to obtain 
\begin{align}
H(U|Q)\geq& h_2\left(\alpha\star h_2^{-1}\left(H(\tilde{Y}|Q)\right)\right)\nonumber\\
\geq &h_2\left(\alpha\star h_2^{-1}\left(\min\left(1,\left(R-h_2(q)+q\right)^+\right)\right)\right)
\end{align}
where  the second inequality follows by
 \begin{align}
 R&\leq I(X_1X_2;Y|Q)\nonumber\\
 &= I(X_1X_2;Y\tilde{Y}|Q)\nonumber\\
&= H(\tilde{Y}|Q)+H(Y|\tilde{Y}Q)-H(\tilde{Y}|X_1X_2)\nonumber\\
&\leq H(\tilde{Y}|Q)+h_2(q)-q\label{relate}
 \end{align}
and the monotonicity of $h_2(\alpha\star h_2^{-1}(x))$ in $x$. This approach gives an upper bound that is tight when $C_1+C_2\geq 1.5317$ (and when $C_1+C_2\leq 1.5$). The range of symmetric bit-pipe capacities $C$ for which Mrs. Gerber's lemma is tight is shown in Fig. \ref{figplotbinaryadder}. In fact, Mrs. Gerber's lemma is within less than $10^{-3}$ bits of capacity for all $C$ in Fig.~\ref{figplotbinaryadder}.
\end{remark}
\section{Conclusion}
We studied diamond networks with an orthogonal broadcast channel and found new upper and lower bounds on their capacities. 
The lower bound is based on Marton's coding technique and superposition coding. We showed through an example with a Gaussian MAC that the new lower bound strictly improves the previous bounds in \cite{TraskovKramer07,KangLiu11,KangLiuChong15}. The proof technique for the upper bound generalizes bounding techniques of Ozarow \cite{Ozarow80} and Kang and Liu \cite{KangLiu11} and applies to discrete memoryless MACs. We specialized the upper bound for networks with a Gaussian MAC and a binary adder MAC. We strengthened the  results of Kang and Liu \cite{KangLiu11,KangLiuChong15} for Gaussian MACs and found  the capacity for binary adder MACs.

\appendices
\section{Error probability analysis for Section \ref{secach}}
\label{aperroranalysis}
We give a proof for discrete alphabet MACs. A proof for the AWGN MAC follows in the usual way be quantizing alphabets and taking limits \cite[Page 50]{ElGamalKim}. We calculate the average error probability $P_e$, averaged over the codebook and the message set,  and show that $P_e$ approaches zero, as $n$ gets large,  if we have
\begin{align}
&R_1^\prime+R_2^\prime>I(X_1;X_2|U)\label{c1c2}\\
& R_{12}+R_1+R_1^\prime< C_1\label{C1}\\
& R_{12}+R_2+R_2^\prime< C_2\label{C2}\\
& R_{12}+R_1+R_1^\prime+R_2+R_2^\prime< I(X_1X_2;Y)+I(X_1;X_2|U)\label{R}\\
& R_1+R_1^\prime+R_2+R_2^\prime<I(X_1X_2;Y|U)+I(X_1;X_2|U)\label{R1+R2}\\
&R_2+R_2^\prime< I(X_2;Y|X_1,U)+I(X_1;X_2|U)\label{R2}\\
&R_1+R_1^\prime< I(X_1;Y|X_2,U)+I(X_1;X_2|U).\label{R1}
\end{align}
Conditions  \eqref{c1c2}-\eqref{R1}, together with $R_1^\prime,R_2^\prime,R_1,R_2,R_{12}\geq 0$ and the rate-splitting condition in \eqref{ratesplit}, characterize an achievable rate.
By the symmetry of the codebook construction and the encoding/decoding scheme, we have\begin{align}
P_e&=\Pr\left((M_{12},M_1,M_2)\neq(\hat{M}_{12},\hat{M}_1,\hat{M}_2)\right)\nonumber\\
&=\Pr\left((M_{12},M_1,M_2)\neq(\hat{M}_{12},\hat{M}_1,\hat{M}_2)|(M_{12},M_1,M_2)=(1,1,1)\right).
\end{align}
Conditioned on $(M_{12},M_1,M_2)=(1,1,1)$, an error occurs only if one of the following events occurs:
\begin{itemize}
\item  $\mathcal{E}_1$: There is no index pair $(m^\prime_1,m^\prime_2)$ such that $(U^n(1),X^n_1(1,1,\hat{m}^\prime_1),X^n_2(1,1,\hat{m}^\prime_2),Y^n)\in\mathcal{T}^n_\epsilon$. 
\item $\mathcal{E}_2$: There are $\tilde{m}_{12},\tilde{m}_1,\tilde{m}_2,\tilde{m}^\prime_1, \tilde{m}^\prime_2$ such that $(\tilde{m}_{12},\tilde{m}_1,\tilde{m}_2)\neq (1,1,1)$ and  $$(U^n(\tilde{m}_{12}),X^n_1(\tilde{m}_{12},\tilde{m}_1,\tilde{m}^\prime_1),X^n_2(\tilde{m}_{12},\tilde{m}_2,\tilde{m}^\prime_2),Y^n)\in\mathcal{T}^n_\epsilon.$$
\end{itemize}
We  have
\begin{align}
P_e&\leq \Pr(\mathcal{E}_1|(M_{12},M_1,M_2)=(1,1,1))+\Pr(\mathcal{E}_2|(M_{12},M_1,M_2)=(1,1,1))\nonumber\\
&= \Pr(\mathcal{E}_1)+\Pr(\mathcal{E}_2).
\end{align}
Using the  Mutual Covering lemma \cite[Lemma 8.1]{ElGamalKim}, $\Pr(\mathcal{E}_1)$ can be made  small for large  $n$ if  \eqref{c1c2} is satisfied.
To analyze $\Pr(\mathcal{E}_2)$, consider the following partition of $\mathcal{E}_2$:
\begin{itemize}
\item  $\tilde{m}_{12}\neq1$
\item  $\tilde{m}_{12}=1$, $\tilde{m}_{1}\neq1$, $\tilde{m}_{2}\neq1$
\item  $\tilde{m}_{12}=1$, $\tilde{m}_{1}=1$, $\tilde{m}_{2}\neq1$
\item  $\tilde{m}_{12}=1$, $\tilde{m}_{1}\neq1$, $\tilde{m}_{2}=1$.
\end{itemize}

The first case has a small error probability, for large $n$, if \eqref{R} is satisfied. Similarly, the second, third, fourth cases have  small error probabilities, for large $n$, if \eqref{R1+R2}, \eqref{R2}, \eqref{R1} are satisfied, respectively. 
We illustrate the analysis for the second case here:
\begin{align}
&\Pr\left(\bigcup_{\tilde{m}_{1}\neq1,\tilde{m}_{2}\neq1,\tilde{m}^\prime_{1},\tilde{m}^\prime_{2}} (U^n(1),X^n_1(1,\tilde{m}_1,\tilde{m}^\prime_1),X^n_2(1,\tilde{m}_2,\tilde{m}^\prime_2),Y^n)\in\mathcal{T}^n_\epsilon\right)\nonumber\\
&\leq\quad \sum_{\tilde{m}_{1}\neq1,\tilde{m}_{2}\neq1,\tilde{m}^\prime_{1},\tilde{m}^\prime_{2}}\Pr\left((U^n(1),X^n_1(1,\tilde{m}_1,\tilde{m}^\prime_1),X^n_2(1,\tilde{m}_2,\tilde{m}^\prime_2),Y^n)\in\mathcal{T}^n_\epsilon\right)\nonumber\\
&\leq\quad 2^{n(R_1+R_1^\prime+R_2+R_2^\prime)}\sum_{(u^n,x_1^n,x_2^n,y^n)\in\mathcal{T}_\epsilon} p(u^n)p(x_1^n|u^n)p(x_2^n|u^n)p(y^n|u^n)\nonumber\\
&\leq\quad 2^{n(R_1+R_1^\prime+R_2+R_2^\prime)}2^{n(H(UX_1X_2Y)+\epsilon H(UX_1X_2Y))}\nonumber\\&\quad\times2^{-n(H(U)-\epsilon H(U))}2^{-n(H(X_1|U)-\epsilon H(X_1|U))}2^{-n(H(X_2|U)-\epsilon H(X_2|U))}2^{-n(H(Y|U)-\epsilon H(Y|U))}\nonumber\\
&=\quad 2^{n(R_1+R_1^\prime+R_2+R_2^\prime)}2^{-n(I(X_1X_2;Y|U)+I(X_1;X_2|U)-\delta(\epsilon))}\label{pipi}
\end{align}
where $\delta(\epsilon)\to 0$ as $\epsilon\to 0$. For rates that satisfy \eqref{R1+R2}, the RHS in \eqref{pipi} approaches zero as $n$ grows large.
\section{Concavity in $p(x_1,x_2)$}
\label{app1}

Consider \eqref{thmupp15} and
\begin{align}
I(X_1;Y|X_2)=H(Y|X_2)-H(Y|X_1X_2).
\end{align}  
$H(Y|X_2)$ is a concave function of $p(x_2,y)$ which is  a linear function of $p(x_1,x_2)$.
$H(Y|X_1X_2)$ is a linear function of $p(x_1,x_2)$.  So $I(X_1;Y|X_2)$ is concave in $p(x_1,x_2)$. 
A similar result holds for $I(X_2;Y|X_1)$ and $I(X_1,X_2;Y)$ when $p(y|x_1,x_2)$ is fixed.

Finally,  consider the last RHS term in \eqref{thmupp15} when $p(y|x_1,x_2)$ and $p(u|x_1,x_2,y)$ are fixed. We have 
\begin{align}
&\hspace{-1.25cm}I(X_1X_2;Y|U)+I(X_1;U|X_2)+I(X_2;U|X_1)\label{app1eq}\\&\hspace{-1.25cm}=H(Y|U) - H(YU|X_1X_2) + H(U|X_2) + I(X_2;U|X_1).\nonumber
\end{align} 
$H(Y|U)$ is concave in $p(u,y)$, $H(YU|X_1X_2)$ is linear in $p(x_1,x_2)$, $H(U|X_2)$ is concave in $p(u,x_2)$ and $I(X_2;U|X_2)$ is concave in $p(x_1,x_2)$. Since $p(u,y)$ and $p(u,x_2)$ are both linear in $p(x_1,x_2)$, \eqref{app1eq} is concave in $p(x_1,x_2)$.
%: $p(u,y)=\sum_{x_1x_2}p(x_1,x_2)p(u,y|x_1x_2)$,  $p(u,x_2)=\sum_{x_1}p(x_1,x_2)p(u|x_1,x_2)$,  and $p(u,x_1)=\sum_{x_2}p(x_1,x_2)p(u|x_1,x_2)$ . 
%Putting these together, .

%\section{Factoring of $p_{X_{1I}X_{2I}Y_IU_IQ}$ in \eqref{pfactors}}
%\label{factors}
%The factoring of $p_{X_{1I}X_{2I}Y_IU_IQ}$ in \eqref{pfactors} may be seen as follows. Recall that $Q=(U^{I-1},I)$. We have 
%\begin{align}
%p_{X_{1I}X_{2I}Y_IU_IQ}(x_1,x_2,y,u,q)&=p_Q(q)p_{X_{1I}X_{2I}|Q}(x_{1},x_{2}|q)p_{Y_I|X_{1I}X_{2I}Q}(y|x_{1},x_{2},q)p_{U_I|Y_IX_{1I}X_{2I}Q}(u|y,x_1,x_2,q)\\
%%&=p_Q(q)p_{X_{1I}X_{2I}|Q}(x_{1},x_{2}|q)p_{Y_i|X_{1i}X_{2i}Q}(y|x_{1},x_{2},u^{i-1},i)p_{U_i|Y_iX_{1i}X_{2i}Q}(u|y,x_1,x_2,u^{i-1},i)\\
%&=p_Q(q)p_{X_{1i}X_{2i}|Q}(x_{1},x_{2}|q)p_{Y|X_{1}X_{2}}(y|x_{1},x_{2})p_{U|Y}(u|y)\\
%&=p_Q(q)p_{X_{1I}X_{2I}|Q}(x_{1},x_{2}|q)p_{Y_I|X_{1I}X_{2I}}(y|x_{1},x_{2})p_{U_I|Y_I}(u|y)\\
%&=p_{X_{1I}X_{2I}}(x_{1},x_{2})p_{Q|X_{1I}X_{2I}}(q|x_1x_2)p_{Y_I|X_{1I}X_{2I}}(y|x_{1},x_{2})p_{U_I|Y_I}(u|y)
%\end{align}
%Also, we have 
%\begin{align}
%p_{X_{1I}X_{2I}}(x_{1},x_{2})&=\frac{1}{n}\sum_{i=1}^np_{X_{1i}X_{2i}}(x_{1},x_{2})
%\end{align}
%and
%\begin{align}
%p_{Q|X_{1I}X_{2I}}(q|x_{1},x_{2})&=\frac{p_{QX_{1I}X_{2I}}(q,x_{1},x_{2})}{p_{X_{1I}X_{2I}}(x_{1},x_{2})}\\
%&=\frac{p_{U^{I-1}IX_{1I}X_{2I}}(u^{i-1},i,x_{1},x_{2})}{p_{X_{1I}X_{2I}}(x_{1},x_{2})}.
%\end{align}

\section{Cardinality bound for Theorem \ref{theoremupp}}
\label{cardin1}
We follow the line of argument in \cite[Appendix~B]{cribbing85}. Denote by $\mathcal{P}$ the set of all probability vectors $p(x_1,x_2)$ and let $P$ be an element of $\mathcal{P}$. Suppose that $R$ is such that for a certain distribution $p_0(x_1,x_2,u,y,q)=p_0(x_1,x_2)p^\star(y|x_1,x_2)p^\star(u|y)p_0(q|x_1x_2)$ the following inequalities hold:
\begin{align}
&R\leq C_1+I_0(X_2;Y|X_1Q),\\
&R\leq C_2+I_0(X_1;Y|X_2Q),\\
&R\leq I_0(X_1X_2;Y|Q),\\
&R\leq C_1+C_2-I_0(X_1X_2;U|Q)+I_0(X_2;U|X_1Q)+I_0(X_1;U|X_2Q)
\end{align}
In the above inequalities, the index $0$ on the mutual information terms emphasizes that the mutual information is evaluated for $p_0(x_1,x_2,u,y,q)$. 
We now interpret $p_0(x_1,x_2|q)$ for every $q\in\mathcal{Q}$ as an element $P^q_0$ of $\mathcal{P}$ with a corresponding probability  $p_0(q)$.
Consider the following continuous functions that map an element of $\mathcal{P}$ into an element of $\mathbb{R}$.
\begin{align}
&f_{x_1,x_2}(P)=\Pr_{\quad P}\{X_1=x_1,X_2=x_2\},\quad \forall (x_1,x_2)\in\mathcal{X}_1\times \mathcal{X}_2 \text{ except one}\\
&f_I(P)=I_P(X_2;Y|X_1),\label{fix1}\\
&f_{II}(P)= I_P(X_1;Y|X_2),\\
&f_{III}(P)= I_P(X_1X_2;Y),\\
&f_{IV}(P)= -I_P(X_1X_2;U)+I_P(X_2;U|X_1)+I_P(X_1;U|X_2)\label{fix4}
\end{align}

We are interested in the following terms.
\begin{align}
&p_0(x_1,x_2)=\sum_{q\in\mathcal{Q}}f_{x_1,x_2}(P^q_0)p_0(q),\quad \forall (x_1,x_2)\in\mathcal{X}_1\times \mathcal{X}_2 \text{ except one} \label{convexcomterm1}\\
&I_0(X_2;Y|X_1Q)=\sum_{q\in\mathcal{Q}}f_{I}(P^q_0)p_0(q),\\
&I_0(X_1;Y|X_2Q)=\sum_{q\in\mathcal{Q}}f_{II}(P^q_0)p_0(q),\\
&I_0(X_1X_2;Y|Q)= \sum_{q\in\mathcal{Q}}f_{III}(P^q_0)p_0(q),\\
&-I_0(X_1X_2;U|Q)+I_0(X_2;U|X_1Q)+I_0(X_1;U|X_2Q)=\sum_{q\in\mathcal{Q}}f_{IV}(P^q_0)p_0(q)\label{convexcomtermlast}
\end{align}
The Fenchel-Eggleston-Carath\'{e}odory theorem \cite[Appendix A]{ElGamalKim}  ensures that there are $|\mathcal{X}_1||\mathcal{X}_2|+3$ vectors $P_k\in\mathcal{P}$, $k=1,\ldots,|\mathcal{X}_1||\mathcal{X}_2|+3$, whose convex combination gives \eqref{convexcomterm1}-\eqref{convexcomtermlast}.
In other words, we may restrict attention to $|\mathcal{Q}|\leq |\mathcal{X}_1||\mathcal{X}_2|+3$.
%The LHS mutual information terms above are evaluated for $p(q^\prime,x_1,x_2,u,y)=p(q^\prime)p(x_1,x_2|q^\prime)p^\star(y|x_1,x_2)p^\star(u|y)$.

\begin{remark}
We need to keep $p(x_1,x_2)$ fixed because $p^\star(u|y)$ can be a function of $p(x_1,x_2)$ and we don't want to change $p^\star(u|y)$.
\end{remark}

\begin{remark}
\label{cardin2}
If  $p^\star(u|y)$ is fixed (and not a function of $p(x_1,x_2)$), then it suffices to have $|\mathcal{Q}|\leq 4$ because we must fix only \eqref{fix1}-\eqref{fix4}.
\end{remark}

\section{Theorem \ref{theoremupp} is tighter than Theorem \ref{thmupp} }
\label{tighter}
Let $R$ be less than or equal to the upper bound of Theorem \ref{theoremupp}. Therefore, there is a $p(x_1,x_2)$ for which for all $p(u|x_1,x_2,y)=p(u|y)$ there is a $p(q|x_1,x_2,y,u)=p(q|x_1,x_2)$ such that the constraints in \eqref{21} hold.
Combining the  last two bounds in  \eqref{21}, we obtain
\begin{align}
2R&\leq C_1+C_2+I(X_1X_2;Y|Q)-I(X_1X_2;U|Q)+I(X_1;U|X_2Q)+I(X_2;U|X_1Q)\nonumber\\
%&= C_1+C_2+I(X_1X_2;YU|Q)-I(X_1X_2;U|Q)+I(X_1;U|X_2Q)+I(X_2;U|X_1Q)\nonumber\\
&= C_1+C_2+I(X_1X_2;Y|UQ)+I(X_1;U|X_2Q)+I(X_2;U|X_1Q)\nonumber\\
&\leq C_1+C_2+I(X_1X_2;Y|U)+I(X_1;U|X_2)+I(X_2;U|X_1).\label{concavenewbound}
\end{align}
Furthermore, we have
\begin{align}
&I(X_2;Y|X_1Q)\leq I(X_2;Y|X_1)\\
&I(X_1;Y|X_2Q)\leq I(X_1;Y|X_2)\\
&I(X_1X_2;Y|Q)\leq I(X_1X_2;Y).
\end{align}
%Therefore, $R$ is also less than or equal to the upper bound of Theorem \ref{thmupp}.

\section{Proof of Theorem \ref{gaussmatch}}
\label{messyproof}
Consider first the regime  $$C\leq\frac{1}{4}\log(1+2P).$$ The lower bound of Theorem \ref{lower bound} meets the upper bound of Corollary \ref{uppGauss} for $U=\phi$, $X_1\sim\mathcal{N}(0,P)$, $X_2\sim\mathcal{N}(0,P)$ , and $X_1$ independent of $X_2$. Consider next $$C\geq\frac{1}{2}\log(1+4P).$$The lower bound of Theorem \ref{lower bound} meets the upper bound of Corollary \ref{uppGauss} for $X_1=X_2=U\sim\mathcal{N}(0,P)$. 

The more interesting regime of $C$ is given in \eqref{thmC}. Consider the upper bound in Corollary \ref{uppGauss} and define the  functions
%Set $C_1=C_2=C$ and $P_1=P_2=P$. Define 
\begin{align}
\begin{array}{lll}
&f_1(C)=2C\\
&f_2(C,\rho)=C+\frac{1}{2}\log\left(1+P(1-\rho^2)\right)\\
&f_3(\rho)=\frac{1}{2}\log\left(1+2P(1+\rho)\right)\\
&f_4(C,\rho)=\frac{1}{2}\left(2C+\frac{1}{2}\log \left(1+2P(1+\rho)\right)-\frac{1}{2}\log\left(\frac{1}{1-\rho^2}\right)\right)\\
&f^\prime_4(C,\rho)=\left\{\begin{array}{ll}f_4(C,\rho)&\rho\leq \rho^{(2)}\\f_2(C,\rho)&\rho>\rho^{(2)}.\end{array}\right.\end{array}\label{f4}
\end{align}
The functions $f_1(C)$, $f_2(C,\rho)$, $f_3(\rho)$, $f^\prime_4(C,\rho)$ are plotted in Fig.~\ref{poss} for different values of $C$  (where $C$ increases from Fig.~\ref{poss3} to Fig.~\ref{poss5}). 
\begin{remark}
 For  symmetric diamond networks, one can check that $f_4(C,\rho)\leq f_2(C,\rho)$ for all $0\leq \rho\leq 1$. Recall from Corollary \ref{uppGauss} that $f_4(C,\rho)$ is not limiting for $\rho>\rho^{(2)}$. This is  reflected in the definition of $f^\prime_4(C,\rho)$. So we can write the upper bound of Corollary \ref{uppGauss} as 
 \begin{align}
 \label{regimeremark}
 \max_\rho \min\{f_1(C),f_3(\rho),f^\prime_4(C,\rho)\}.
 \end{align}
 \end{remark}
\begin{remark}
The function $f_4(C,\rho)$ is concave in $\rho$ and  it attains its maximum at $\rho^{(1)}$ given in \eqref{r1}. $f_2(C,\rho)$ is concave and decreasing in $\rho$. One can check by substitution and differentiation that $f^\prime_4(C,\rho)$ is continuous and differentiable with respect to $\rho$ at $\rho=\rho^{(2)}$. Furthermore, $f^\prime_4(C,\rho)$ is concave and attains its maximum at $\rho^{(1)}$. The derivative of $f^\prime_4(C,\rho)$ is non-positive at $\rho^{(2)}$ and thus we have $\rho^{(1)}\leq \rho^{(2)}$.
\end{remark}
 
 \begin{remark}
 \label{zayeremark}
 In the regime $\frac{1}{4}\log(1+2P)< C <\frac{1}{2}\log(1+4P)$ the functions $f_3(\rho)$ and $f_4^\prime(C,\rho)$ have exactly one  point $\rho$ where $f_3(\rho)=f_4^\prime(C,\rho)$. To see this we study the zeros of the function $g(C,\rho)=f_3(\rho)-f^\prime_4(C,\rho)$. Since $C>\frac{1}{4}\log(1+2P)$, we have $g(C,0)<0$. Since $C<\frac{1}{2}\log(1+4P)$, we have $g(C,1)>0$.  Furthermore, we have
 \begin{align}
 \frac{\partial g}{\partial \rho}(C,\rho)= \frac{\partial f_3}{\partial \rho}(\rho)- \frac{\partial f^\prime_4}{\partial \rho}(C,\rho)> 0,\quad \forall \rho\in[0,1]
 \end{align}
and thus $g(C,\rho)$ is increasing in $\rho$. So $g(C,\rho)$ has exactly one zero.
 \end{remark}
\begin{remark}
\label{2Cfprime}
In the regime 
\begin{align}
2C\leq f^\prime_4(C,\rho^{(1)})\label{2Clarger}
\end{align}  we have
\begin{align}
f_3(\rho^{(1)})&=\frac{1}{2}\log\left(1+2P(1+\rho^{(1)})\right)\nonumber\\
&\geq \frac{1}{2}\log\left(1+2P(1+\rho^{(1)})\right)-\frac{1}{2}\log\left(\frac{1}{1-{\left(\rho^{(1)}\right)}^2}\right)\nonumber\\
&= 2f^\prime_4(C,\rho^{(1)})-2C \nonumber\\
&{\geq} f^\prime_4(C,\rho^{(1)}).\label{vflast}
\end{align}
Inequality \eqref{vflast} follows by \eqref{2Clarger}. The implication is that $f_3(\rho)$ is not ``limiting" in \eqref{regimeremark} for the the regime given by \eqref{2Clarger}. %This is due to \eqref{vflast} and the fact that $f_3(\rho)$ is increasing in $\rho$.
\end{remark}

\begin{figure}
        \centering
%        \begin{subfigure}[b]{0.49\textwidth}
 %               \includegraphics[width=\textwidth]{poss4.eps}
  %              \caption{}
   %             \label{poss4}
    %    \end{subfigure}%
    %    ~ %add desired spacing between images, e. g. ~, \quad, \qquad etc.
          %(or a blank line to force the subfigure onto a new line)
        \begin{subfigure}[b]{0.49\textwidth}
\input{Gausscasea}
                \caption{}
                \label{poss3}
        \end{subfigure}
        ~ %add desired spacing between images, e. g. ~, \quad, \qquad etc.
          %(or a blank line to force the subfigure onto a new line)
        \begin{subfigure}[b]{0.49\textwidth}
\input{Gausscaseb}
                     \caption{}
\label{poss2}
        \end{subfigure}
                \begin{subfigure}[b]{0.49\textwidth}
              \input{Gausscasec}

                     \caption{}
\label{poss5}
        \end{subfigure}
                \caption{The functions in \eqref{f4} for different values of $C$ where $C$ increases from Fig. \ref{poss3} to Fig. \ref{poss5}. Note the range of the $R-$axis in Fig. \ref{poss5}.}
        \label{poss}
\end{figure}

Fix the value of $C$. Define $\rho^\star$ to be the optimal correlation coefficient in Corollary \ref{uppGauss} and let $R_{\max}$ be the maximum value it attains. We have one of the following cases.
\begin{itemize}
\item[($a$)]  $\rho^\star$ is such that $R_{\max}=f_1(C)$.
\item[($b$)] $\rho^\star$ is the unique $\rho$ that maximizes $f^\prime_4(C,\rho)$ and $R_{\max}=f^\prime_4(C,\rho^{(1)})$.
\item[($c$)] $\rho^\star$ is such that $R_{\max}=f_3(\rho^\star)$.
\end{itemize}

%We remark that in  cases ($a$) and ($e$), we have $2C\leq f_4(C,\rho^{(1)})$;
%Otherwise, we would have $R_{\max}<2C$.
%Moreover, the optimizing $\rho^\star$ is not necessarily unique (see, e.g., Fig. \ref{poss4} and Fig. \ref{poss3}). We can choose, without loss of generality, $\rho^\star=\rho^{(1)}$ because $f_3(C,\rho^{(1)})\geq f_4(C,\rho^{(1)})$ (by Remark \ref{}) and  $f_2(C,\rho^{(1)})\geq f_4(C,\rho^{(1)})$ (by Remark \ref{}). So in both cases ($a$) and ($e$) we have, without loss of generality,  $\rho^\star=\rho^{(1)}$. Since $\rho^{(1)}\leq \rho^{(2)}$, case ($e$) becomes an irrelevant case.

When $C=0$, we are in case ($a$).  As $C$ increases, $f_3(\rho)$ remains unchanged but $f_1(C)$ and $f^\prime_4(C,\rho)$ increase. We remain in case $(a)$ as long as $2C\leq f_4^\prime(C,\rho^{(1)})$. This is illustrated in Fig. \ref{poss3}. We then transit to case $(b)$ where $R_{\text{max}}= f_4^\prime(C,\rho^{(1)})$, see Fig. \ref{poss2}. To see this, keep increasing $C$ until $2C=f^\prime_4(C,\rho^{(1)})$. At this point, we are about to leave case ($a$), and we still have $f_3(\rho^{(1)})\geq f^\prime_4(C,\rho^{(1)})$ (see Remark~\ref{2Cfprime}). Moreover, $f^\prime_4(C,\rho)$ is decreasing in $\rho$ for $\rho>\rho^{(1)}$, and $f_3(\rho)$ is increasing in $\rho$. So the crossing point of the two curves should be at or before $\rho^{(1)}$.  As $C$ further increases,  $\rho^\star$ remains equal to $\rho^{(1)}$ until $f_3(\rho^{(1)})=f^\prime_4(C,\rho^{(1)})$. From that point on, we have 
\begin{align}
C\geq \frac{1}{4}\log\frac{1+2P\left(1+{\rho^{(1)}}\right)}{1-{\rho^{(1)}}^2}\label{conditionCC1}
\end{align}
and  we move into case ($c$). In this case, $\rho^\star$ is such that $R_{\max}=f_3(\rho^\star)=f^\prime_4(C,\rho^\star)$ (see Fig. \ref{poss5}). We note that as $C$ increases, so does $\rho^\star$. We thus have $\rho^\star\leq\rho^{(2)}$ if
\begin{align}
C\leq \frac{1}{4}\log\frac{1+2P\left(1+{\rho^{(2)}}\right)}{1-{\rho^{(2)}}^2}.\label{conditionC2}
\end{align}
In this regime, besides the bounds \eqref{cor00}-\eqref{cor1}, we have
\begin{align}
R_{\max}&=f_3(\rho^\star)\nonumber\\
&= 2f^\prime_4(C,\rho^\star)-f_3(\rho^\star)\nonumber\\
&= \left\{\begin{array}{ll}2C-\frac{1}{2}\log\left(\frac{1}{1-{\rho^\star}^2}\right)& \rho^\star\leq \rho^{(2)}\\2f_2(C,\rho^\star)-f_3(\rho^\star)& \rho^\star>\rho^{(2)}\end{array}\right..
\end{align}
Therefore, when $C$ satisfies \eqref{conditionCC1} and \eqref{conditionC2} the upper bound meets the lower bound of \eqref{eqlow1}-\eqref{eqlow4}.

\section{Proof of Theorem \ref{gaussmatch2}}
\label{apgaussmatch2}
Consider \eqref{21N} and suppose $C^{\diamond}$ is the capacity of the network. For symmetric networks, we thus have
\begin{align}
C^{\diamond}\leq & 2C-\frac{1}{2}\log\left(N+2^{2C^{\diamond}}\right)-\frac{1}{2}\log\left(1+N\right)+\log\left( 1+N+P\left(1-\rho^2\right)\right)\nonumber\\
\stackrel{(a)}{\leq}& 2C-\frac{1}{2}\log\left(N+2^{2R^{(l)}_{\max}}\right)-\frac{1}{2}\log\left(1+N\right)+\log\left( 1+N+P\left(1-\rho^2\right)\right)\label{above25}
\end{align}
where $R^{(l)}_{\max}$ is the maximum  admissible rate in the lower bound of \eqref{eqlow1}-\eqref{eqlow4}. Note that $(a)$ holds because  $R^{(l)}_{\max}\leq C^{\diamond}$.
The upper bound of Theorem \ref{theoremupp} is thus loosened to:
\begin{align}
R^{(u)}_{\max}=\max_{0\leq\rho\leq 1} \min_N\min\left\{f_1(C),f_2(C,\rho),f_3(\rho),f_5(C,\rho,N)\right\}\label{upphectic}
\end{align}
where $f_1$, $f_2$, $f_3$ are defined in \eqref{f4} and $f_5(C,\rho,N)$ is the RHS of \eqref{above25}. Furthermore, define 
\begin{align}
f_0(C,\rho)=2C-\frac{1}{2}\log\left(\frac{1}{1-\rho^2}\right)
\end{align}
so that we have
\begin{align}
R^{(l)}_{\max}=\max_{0\leq\rho\leq 1} \min\left\{f_0(C,\rho),f_2(C,\rho),f_3(\rho)\right\}.\label{lowhectic}
\end{align}

We shall prove that \eqref{lowhectic} is equal to \eqref{upphectic} for the range of $C$ given in \eqref{thmC2}. We start with \eqref{lowhectic}. Fix $C$ and let $\lambda$ be the optimizing correlation coefficient.  The functions $f_0(C,\rho)$ and $f_2(C,\rho)$ are decreasing in $\rho$ and $f_3(\rho)$ is increasing in $\rho$. So depending on how $f_3(0)$ compares with $\min(f_0(C,0),f_2(C,0))$ we have the following cases for $\lambda$:
\begin{itemize}
\item If $2C\leq \frac{1}{2}\log(1+2P)$, then we have $\lambda=0$ and $R=2C$ is achievable using independent Gaussian random variables $X_1$, $X_2$ with zero mean and variance $P$. 
\item If $2C> \frac{1}{2}\log(1+2P)$, then $\lambda$ is such that either 
$R^{(l)}_{\max}=f_3(\lambda)=f_2(C,\lambda)$ (where the cut-set bound is achievable) or $R^{(l)}_{\max}=f_3(\lambda)=f_0(C,\lambda)$. We show that in the latter case $R^{(l)}_{\max}$ and  $R^{(u)}_{\max}$ are equal if \eqref{thmC2} is satisfied.
\end{itemize} 

Suppose $\lambda$ is such that 
\begin{align}
R^{(l)}_{\max}=f_3(\lambda)=f_0(C,\lambda).\label{becauseeq}
\end{align} 
Here $\lambda$ is  defined by the crossing point of $f_3(\rho)$ and $f_0(C,\rho)$ and is such that
\begin{align}
C=\frac{1}{4}\log\left(\frac{1+2P(1+\lambda)}{1-\lambda^2}\right).\label{inc}
\end{align}
The RHS of \eqref{inc} is an increasing function of $\lambda$ and thus we have $\lambda\leq \rho^{(2)}$ if and only if \eqref{thmC2} is satisfied. 

Next consider \eqref{upphectic} in the regime of $C$ given by \eqref{thmC2}. 
First note that for a fixed $N$, $f_5(C,\rho,N)$ is decreasing in $\rho$.  Let 
\begin{align}N_\lambda=P\left(\frac{1}{\lambda}-\lambda\right)-1.\label{step1}\end{align} 
Since $\lambda\leq \rho^{(2)}$, we have $N_\lambda\geq0$ 
 and the upper bound may be written as follows:
\begin{align}
R\leq& f_5(C,\rho,N_\lambda)\nonumber\\
=& 2C-\frac{1}{2}\log\left(N_\lambda+2^{2R^{(l)}_{\max}}\right)-\frac{1}{2}\log\left(1+N_\lambda\right)+\log\left( 1+N_\lambda+P\left(1-\rho^2\right)\right)\nonumber\\
 \stackrel{(a)}{=}&2C-\frac{1}{2}\log\left(1+N_\lambda+2P(1+\lambda)\right)-\frac{1}{2}\log\left(1+N_\lambda\right)+\log\left( 1+N_\lambda+P\left(1-\rho^2\right)\right)\label{rhsakhar}
\end{align}
where $(a)$ holds by \eqref{becauseeq}. 
The RHS of \eqref{rhsakhar}, evaluated for $\rho=\lambda$, is given by $$f_5(C,\lambda,N_\lambda)=2C-\frac{1}{2}\log\left(\frac{1}{1-\lambda^2}\right)=f_0(C,\lambda).$$ This follows by the argument in \eqref{choiceofN}-\eqref{regimeofN}. We conclude that $f_5(C,\rho,N_\lambda)$ is equal to $f_0(C,\lambda)=f_3(\lambda)$ at $\rho=\lambda$. Since $f_5(C,\rho,N_\lambda)$ is decreasing in $\rho$,  $\lambda$ is the optimal $\rho^\star$ in \eqref{upphectic} too. This is illustrated in Fig.~\ref{possnew}. So in the regime characterized by \eqref{thmC2} the upper bound is equal to $R^{(l)}_{\max}$ and is thus achievable. 
\begin{figure}[t!]
        \begin{center}
               \input{lambdafig}
                \caption{Several functions of $\rho$. At $\rho=\lambda$, the curves $f_0(C,\rho)$, $f_3(\rho)$, and $f_5(C,\rho,N_\lambda)$  all have the same value.}
                \label{possnew}
                \end{center}
\end{figure}%

\begin{remark}
\label{generalgaussmatch2}
A similar result can be established for asymmetric networks. Let $f_0(C_1,C_2,\rho)$ and $f_3(\rho)$ be the RHSs of \eqref{eqlow1} and \eqref{eqlow4}, respectively. Define $\lambda$ to be the optimizing correlation coefficient in \eqref{eqlow1}-\eqref{eqlow4} and  $R_{\max}^{(l)}$ as  the maximum achievable rate. One can check that $R_{\max}^{(l)}$ is equal to the cut-set bound if $C_1,C_2$ satisfy $C_1+C_2\leq  \frac{1}{2}\log\left(1+P_1+P_2\right)$, or if \eqref{PP1} or \eqref{PP2} are satisfied. The cut-set bound may not be achievable by \eqref{eqlow1}-\eqref{eqlow4}  when we have
\begin{align}
&R_{\max}^{(l)}=f_3(\lambda)=f_0(C_1,C_2,\lambda).\label{regime}
\end{align} 
For $C_1,C_2$ where \eqref{PP0} is satisfied, we have
\begin{align}
&\lambda\leq\rho^{(2)}.\label{lambdaro}
\end{align}
Therefore, following the steps in \eqref{step1}-\eqref{rhsakhar}, we have a matching upper bound based on Theorem \ref{theoremupp} (in the form of \eqref{upphectic} but written for general $C_1,C_2$). 
Finally, for $C_1,C_2$ that satisfy \eqref{PP} the cut-set bound is achievable using Theorem \ref{lower bound} with $\frac{U}{\sqrt{P_1}}=\frac{X_1}{\sqrt{P_1}}=\frac{X_2}{\sqrt{P_2}}\sim\mathcal{N}(0,1)$.
\end{remark}

%\begin{remark}
%\label{reduces}
%Theorem \ref{theoremupp} is strictly tighter than \cite[Theorem 1]{KangLiu11}. We compare \eqref{upphectic} with the upper bound of \cite[Theorem 1]{KangLiu11}.  The regime of interest is given by \eqref{regime} because otherwise both upper bounds reduce to the cut-set bound which is tight. First  suppose $\lambda> \rho^{(2)}$. In this case,  \cite[Theorem 1]{KangLiu11} reduces to the cut-set bound and is larger than or equal to \eqref{upphectic}.
%Next suppose $\lambda\leq \rho^{(2)}$. Here, \eqref{upphectic} can be shown to be equal to $R_{\max}^{(l)}$ and is thus tight but \cite[Theorem 1]{KangLiu11} may not be tight, see \cite[Theorem 3]{KangLiuChong15}. For example, when $P=.25$ and $C=.15$ Theorem \ref{theoremupp} gives $C^{\diamond}\leq .2994$ (which is tight) whereas \cite[Theorem 1]{KangLiu11} gives $C^{\diamond}\leq .3$.
%\end{remark}

\section{Proof of Lemma \ref{lemmadoubly}}
\label{detailoptlemma}
Consider the following optimization problem:
\begin{align}
\max_{p(x_1,x_2)}\min_{p(u|y)}\label{apterm1}\min\left\{\begin{array}{l}
2C\\
C+I(X_2;Y|X_1)\\
C+I(X_1;Y|X_2)\\
I(X_1X_2;Y)\\
\frac{1}{2}\left(2C+I(X_1X_2;Y|U)+I(X_1;U|X_2)+I(X_2;U|X_1)\right)\end{array}\right\}.
\end{align}
%Let $p^{(1)}(x_1,x_2)$ be an optimizing pmf in \eqref{apterm1}. We construct an optimizing pmf $p^\star(x_1,x_2)$ that is given by a doubly symmetric source. 

We first note that the objective function   in \eqref{apterm1} is symmetric in $p(x_1,x_2)$ (when $p(u|y)$ is given by the channel in Fig.~\ref{channelu}). More precisely, for every pmf $p(x_1,x_2)$, the pmf $\bar{p}(x_1,x_2)$ with $\bar{p}(0,0)=p(1,1)$, $\bar{p}(0,1)=p(1,0)$, $\bar{p}(1,0)=p(0,1)$, $\bar{p}(1,1)=p(0,0)$ gives the same objective function. Let $p^{(1)}(x_1,x_2)$ be the pmf that attains the optimal value of \eqref{apterm1}. Take the pmf $p^{(1)}(x_1,x_2)$ and form the doubly symmetric pmf
\begin{align*}
p^\star(x_1,x_2)=\frac{p^{(1)}(x_1,x_2)+\bar{p}^{(1)}(x_1,x_2)}{2},\quad  x_1=0,1,\ x_2=0,1.\\
%p^\star(u|y)=p^{(1)}(u|y).
\end{align*} 
%Note that $p^\star(x_1,x_2)$ is the pmf of a doubly symmetric binary source (DSBS).

For a fixed $p(u|y)$, all terms of the min expression in \eqref{apterm1} are concave functions of $p(x_1,x_2)$ (see Remark \ref{concave}). Therefore, at any point $(p^\star(x_1,x_2),p(u|y))$ they take on values larger than or equal to their respective values at $(p^{(1)}(x_1,x_2),p(u|y))$ (or $(\bar{p}^{(1)}(x_1,x_2),p(u|y))$). This proves that there exists at least one optimizing doubly symmetric pmf $p(x_1,x_2)$ in \eqref{apterm1}.

\section{Proof of Theorem \ref{ranges}}
\label{apranges}
The proof is similar to the proof of Theorem \ref{gaussmatch}. The lower and upper bounds match for $C\leq 0.75$ and $C\geq 1$. In the former regime, the cut-set bound is achievable using no cooperation among the relays with $p=\frac{1}{2}$ in \eqref{eqlowbin1}-\eqref{eqlowbin3}. In the latter case, the cut-set bound is achievable using full cooperation among the relays with $p=\frac{1}{3}$ in \eqref{eqlowbin1}-\eqref{eqlowbin3}. We prove the theorem for $C$'s satisfying~\eqref{thmCbinary} and we assume $C\leq 1$.

%First, note that inequality \eqref{loose} is redundant for $p\leq \frac{1}{2}$. 
Define
\begin{align}
\label{binaryadderfunc}
\begin{array}{lll}
&g_1(C)=2C\\
&g_3(p)=h_2(p)+1-p\\
&g_4(C,p)=C+h_2(p)-\frac{p}{2}
\end{array}
\end{align}
so that the upper bound of Corollay \ref{corbinaryadder} is given by
\begin{align}
R_{\max}=\max_{0\leq p\leq \frac{1}{2}}\min\{g_1(C),g_3(p),g_4(C,p)\}.\label{maximumcorollary}
\end{align}
The functions $g_1(C)$, $g_3(p)$, and $g_4(C,p)$ are drawn in Fig. \ref{binaryposs} as functions of $p\leq \frac{1}{2}$ for different values of $C$, where $C$ increases from Fig. \ref{binaryposs1} to Fig.~\ref{binaryposs4}. We consider two probabilities: 
\begin{align}
&p^{(1)}=\frac{1}{1+\sqrt{2}}\\ 
&p^{(3)}=2(1-C). %A few remarks follows.
\end{align}
\begin{remark}
$g_4(C,p)$ is concave in $p$ and it attains its maximum at $p=p^{(1)}$. $g_3(p)$ is also concave and it attains its maximum at $p=\frac{1}{3}$. We have $\frac{1}{3}\leq p^{(1)}$.  $g_3(p)$ and $g_4(C,p)$ cross at $p=p^{(3)}$.
\end{remark}
\begin{remark}
\label{2Ctight}
For the regime 
\begin{align}
\label{regimeCgiven}
2C\leq g_4(C,p^{(1)})
\end{align}
we have
\begin{align}
C&\leq h_2(p^{(1)})-\frac{p^{(1)}}{2}< 1-\frac{p^{(1)}}{2}.\label{chertent}
\end{align}
%with equality in \eqref{chertent} if and only if $p=\frac{1}{2}$.
As a result, we have $$g_4(C,p^{(1)})< g_3(p^{(1)}).$$ % and the inequality is strict if $p\neq \frac{1}{2}$.
 The implication is that $g_3(p)$ is not ``limiting" for the regime given by \eqref{regimeCgiven}.
\end{remark}

%\begin{remark}
%\label{p3}
%$g_3(p)$ and $g_4(C,p)$ cross at $p^{(3)}=2(1-C)$ for $\frac{1}{2}\leq C\leq 1$. The intersecting point decreases as $C$ increases. For the regime of  interest, i.e. $C>.75$, we have $p^{(3)}\leq\frac{1}{2}$. We have the following cases:
%\begin{itemize}
%\item If $.75\leq C\leq \frac{g_4(C,p^{(1)})}{2}$, we have $p^{(1)}\leq p^{(3)}\leq \frac{1}{2}$ and $2C\leq g_4(C,p^{(1)})\leq g_3(p^{(1)})$. So the maximum value of  \eqref{maximumcorollary} is given by $g_1(C,p)=2C$, see Fig. \ref{binaryposs1}.
%\item If $\frac{g_4(C,p^{(1)})}{2}\leq C\leq 1-\frac{p^{(1)}}{2}$, we have $p^{(1)}\leq p^{(3)}\leq \frac{1}{2}$ and $g_4(C,p^{(1)})\leq g_3(p^{(1)})$. Since $2C\geq g_4(C,p^{(1)})$, $p^{(1)}$ is the unique $p$ that maximizes \eqref{maximumcorollary}, see 
% Fig. \ref{binaryposs2}
%\item If $1-\frac{p^{(1)}}{2}\leq C\leq \frac{5}{6}$, we have  $\frac{1}{3}\leq p^{(3)}\leq p^{(1)}$. At $p=p^{(3)}$, $g_3(p)$ is decreasing and $g_4(C,p)$ is increasing. So $p^{(3)}$ maximizes \eqref{maximumcorollary}, see Fig. \ref{binaryposs3}.
%\item If $\frac{5}{6}\leq C\leq 1$, we have  $0\leq p^{(3)}\leq \frac{1}{3}$. Furthermore, one can check that $g_3(\frac{1}{3})\leq g_4(C,\frac{1}{3})$. So $p=\frac{1}{3}$ maximizes \eqref{maximumcorollary}, see Fig. \ref{binaryposs4}.
%%\item If $C>1$, $g_3(p)$ and $g_4(C,p)$ do not intersect and $g_3(p)$ lies below $g_4(C,p)$ for all $p$, $0\leq p\leq \frac{1}{2}$. In this regime  $p=\frac{1}{3}$ also maximizes \eqref{upp}.
%\end{itemize}
%\end{remark}

\begin{figure}
        \centering
        \begin{subfigure}[b]{0.48\textwidth}
          \input{binarycasea1}
                \caption{}
                \label{binaryposs1}
        \end{subfigure}%
        ~ %add desired spacing between images, e. g. ~, \quad, \qquad etc.
          %(or a blank line to force the subfigure onto a new line)
        \begin{subfigure}[b]{0.48\textwidth}
          \input{binarycaseb}
                \caption{}
                \label{binaryposs2}
        \end{subfigure}
        ~      %add desired spacing between images, e. g. ~, \quad, \qquad etc.
          %(or a blank line to force the subfigure onto a new line)
        \begin{subfigure}[b]{0.48\textwidth}
          \input{binarycasec}
                     \caption{}
\label{binaryposs3}
        \end{subfigure}
    ~            \begin{subfigure}[b]{0.48\textwidth}
          \input{binarycased}
                     \caption{}
\label{binaryposs4}
        \end{subfigure}
                \caption{The functions in \eqref{binaryadderfunc} for different values of $C$.}
        \label{binaryposs}
\end{figure}

Fix the value of $C$ and denote the maximizing $p$ in  \eqref{maximumcorollary} by $p^\star$. We have one of the following cases:
\begin{itemize}
\item[($a$)] $p^\star$ is such that $R_{\max}=g_1(C)$. 
\item[($b$)] $p^\star=p^{(1)}$ and $R_{\max}=g_4(C,p^{(1)})$.
\item[($c$)] $p^\star=p^{(3)}$ and $R_{\max}=g_3(p^\star)=g_4(C,p^\star)$.
\item[($d$)] $p^\star=\frac{1}{3}$ and $R_{\max}=g_3(\frac{1}{3})$.
\end{itemize}
When $C=0$, we are in case ($a$). We remain in this case as long as $2C\leq g_4(C,\rho^{(1)})$.   When $2C= g_4(C,p^{(1)})$ we  have  $g_4(C,p^{(1)})< g_3(p^{(1)})$, see Remark \ref{2Ctight}. So as $C$  increases we have $p^\star=p^{(1)}$ and we move into  case ($b$), see Fig. \ref{binaryposs2}. In this regime, $g_3(p)$ and $g_4(C,p)$ cross at the point $p^{(3)}$ which is larger than or equal to $p^{(1)}$. As $C$ further increases, the crossing point $p^{(3)}$ of $g_3(p)$ and $g_4(C,p)$ decreases towards $p^{(1)}$, and as soon as $g_3(p^{(1)})=g_4(C,p^{(1)})$ we move into case ($c$) where $p^\star=p^{(3)}$, see Fig. \ref{binaryposs3}. In this case, we have
\begin{align}
R_{\max}=g_3(p^\star)=g_4(C,p^\star).
\label{rmax2}
\end{align}
Using \eqref{rmax2}, we obtain 
\begin{align}R_{\max}&=2g_4(C,p^\star)-g_3(p^\star)\nonumber\\&= 2C+h_2(p^\star)-1.\end{align} 
So the upper bound reduces to
\begin{align}
%R&\leq 2C\\
R&\leq C + h_2(p^\star)\\
R&\leq h_2(p^\star) + 1 - p^\star\\
R&\leq 2C + h_2(p^\star) - 1
\end{align}
which matches the lower bound in \eqref{eqlowbin1}-\eqref{eqlowbin3}. Finally, when $p^{(3)}\leq \frac{1}{3}$ we have $p^\star=\frac{1}{3}$ and full cooperation is achieved. We thus meet the cut-set bound. This concludes the proof.
\section{An upper bound on $H(U|X_1,Q=i)$ and $H(U|X_2,Q=i)$}
\label{HUcX1Q}
$H(U|X_1,Q=i)$ may be expanded as follows:
\begin{align}
&H(U|X_1,Q=i)\nonumber\\
  &=p_{X_1|Q}(0|i)h_2\left(\frac{p_{X_1X_2|Q}(0,0|i)(1-\alpha)+\frac{p_{X_1X_2}(0,1)}{2}}{p_{X_1|Q}(0|i)}\right)+p_{X_1|Q}(1|i)h_2\left(\frac{p_{X_1X_2|Q}(1,1|i)(1-\alpha)+\frac{p_{X_1X_2}(1,0)}{2}}{p_{X_1|Q}(1|i)}\right)\nonumber\\
  &\leq h_2\left((1-q_i)(1-\alpha)+\frac{q_i}{2}\right)\nonumber\\&=h_2(\alpha\star\frac{q_i}{2})\label{sym1}
 \end{align}
 where the inequality is by the concavity of $h_2(.)$ in its argument. $H(U|X_2Q)$ may be bounded similarly.

\section{Proof of Lemma \ref{GMGL}}
\label{concavity lemma}
First, note that for $x\leq h_2(y)$, we have $g(x,y)=0$. For $x\geq h_2(y)$, $g(x,y)$ is non-negative. Since $g(x,y)$ is continuous at $x=h_2(y)$, it suffices to prove convexity of $g(x,y)$ in the regime $x\geq h_2(y)$. Recall that $y\leq 1$ and $\alpha\leq \frac{1}{2}$.
We prove that the  Hessian matrix 
 \begin{align}
\mathbf{H}=\left[\begin{array}{cc}\frac{\partial^2{g(x,y)}}{\partial{x^2}}&\frac{\partial^2{g(x,y)}}{\partial{x}\partial{y}}\\\frac{\partial^2{g(x,y)}}{\partial{y}\partial{x}}&\frac{\partial^2{g(x,y)}}{\partial{y^2}}\end{array}\right].
 \end{align}
  is positive semi-definite. Using Sylvester's criterion \cite[Theorem 7.2.5]{HornJohnson}, $\mathbf{H}$ is positive semi-definite if and only if its leading principal minors are non-negative; i.e., if
\begin{align}
&\frac{\partial^2{g(x,y)}}{\partial{x^2}}\geq 0\\
&\frac{\partial^2{g(x,y)}}{\partial{x^2}}\frac{\partial^2{g(x,y)}}{\partial{y^2}}-\frac{\partial^2{g(x,y)}}{\partial{x}\partial{y}}\frac{\partial^2{g(x,y)}}{\partial{y}\partial{x}}\geq 0.
\end{align}
 
We use the following notation:
\begin{align}
&z=h_2^{-1}\left(\frac{x-h_2(y)}{1-y}\right),\quad y\leq 1\\
&s=\alpha\star\left(\frac{y}{2}+(1-y)z\right)\\
&r(x)=-\frac{h_2^\prime(x)}{h_2^{\prime\prime}(x)}=x(1-x)\ln \frac{1-x}{x}\quad  \ 0\leq x\leq \frac{1}{2}\\
&h_2^\prime(x)=\frac{1}{\ln 2}\ln\frac{1-x}{x}\quad  \ 0\leq x\leq \frac{1}{2}\label{h2prime}\\
&h_2^{\prime\prime}(x)=-\frac{1}{\ln 2}\frac{1}{x(1-x)}\quad  \ 0\leq x\leq \frac{1}{2}.
\end{align}

Taking the partial derivatives, we have 
 \begin{align}
\frac{\partial^2{g(x,y)}}{\partial{x^2}}&=\frac{(1-2\alpha)}{(1-y)\left(h_2^\prime(z)\right)^3}\left(-h_2^{\prime\prime}(z) h_2^{\prime}(s)+(1-2\alpha)(1-y)h_2^{\prime\prime}(s)h^\prime_2(z)\right)\nonumber\\
\label{convexx}&=\frac{(1-2\alpha)h_2^{\prime\prime}(z)h_2^{\prime\prime}(s)}{(1-y)\left(h_2^\prime(z)\right)^3}\left(r(s)-(1-2\alpha)(1-y)r(z)\right)\\
\frac{\partial^2{g(x,y)}}{\partial{x^2}}&\frac{\partial^2{g(x,y)}}{\partial{y^2}}-\frac{\partial^2{g(x,y)}}{\partial{x}\partial{y}}\frac{\partial^2{g(x,y)}}{\partial{y}\partial{x}}\nonumber\\&= \frac{(1-2\alpha)^2}{\left(h_2^\prime(z)\right)^3(1-y)}\left(\begin{array}{l}-\frac{(1-2\alpha)}{4}\left(-h_2^{\prime\prime}(z)h_2^\prime(s)+(1-2\alpha)(1-y)h_2^{\prime\prime}(s)h_2^\prime(z)\right)h_2^{\prime\prime}(\alpha\star\frac{y}{2})\\+\frac{h_2^{\prime\prime}(z)}{h_2^\prime(z)}h_2^{\prime\prime}(y)\left(h_2^\prime(s)\right)^2\\+(1-2\alpha)\left(-(\frac{1}{2}-z)^2h_2^{\prime\prime}(z)-(1-y)h_2^{\prime\prime}(y)\right)h_2^\prime(s)h_2^{\prime\prime}(s)\end{array}\right).\label{convexjoint}
\end{align}
We shall prove that \eqref{convexx} and \eqref{convexjoint} are non-negative for all non-negative parameters $\alpha\leq \frac{1}{2}$, $y\leq 1$, $z\leq \frac{1}{2}$.
 % Recall that $s=\alpha\star(\frac{y}{2}+(1-y)z)$. 
 We start with \eqref{convexx}. We claim that
\begin{align}
r(s)-(1-2\alpha)(1-y)r(z)
\label{rsrz}
\end{align} is non-negative because it is equal to $0$ at $z=\frac{1}{2}$ and is a non-increasing function of $z$, $0\leq z\leq \frac{1}{2}$. To see this, consider
\begin{align}
\frac{\partial}{\partial z}\left(r(s)-(1-2\alpha)(1-y)r(z)\right)=(1-2\alpha)(1-y)\left(r^\prime(s)-r^\prime(z)\right)\label{firstder}
\end{align}
where 
\begin{align}
r^\prime(x)=(1-2x)\ln\frac{1-x}{x}-1.
\end{align}
Since $r^\prime(x)$ is non-increasing in $x$ and $s\geq z$, \eqref{firstder} is non-positive.
This line of argument is very similar to \cite[Theorem 2]{Witsenhausen74} and \cite[Theorem 2]{Cheng14}.

To prove the non-negativity of \eqref{convexjoint} we proceed as follows. We have $h_2^{\prime\prime}(\alpha\star\frac{y}{2})\leq h_2^{\prime\prime}(s)$ because $h_2^{\prime\prime}(x)$ is increasing in $x$, $0\leq x\leq \frac{1}{2}$, and $\alpha\star \frac{y}{2}\leq s$. We thus have 
\begin{align}
\frac{\partial^2{g(x,y)}}{\partial{x^2}}&\frac{\partial^2{g(x,y)}}{\partial{y^2}}-\frac{\partial^2{g(x,y)}}{\partial{x}\partial{y}}\frac{\partial^2{g(x,y)}}{\partial{y}\partial{x}}\nonumber\\
&{\geq} \frac{(1-2\alpha)^2}{\left(h_2^\prime(z)\right)^3(1-y)}\left(\begin{array}{l}-\frac{(1-2\alpha)}{4}\left(-h_2^{\prime\prime}(z)h_2^\prime(s)+(1-2\alpha)(1-y)h_2^{\prime\prime}(s)h_2^\prime(z)\right)h_2^{\prime\prime}(s)\\+\frac{h_2^{\prime\prime}(z)}{h_2^\prime(z)}h_2^{\prime\prime}(y)\left(h_2^\prime(s)\right)^2\\+(1-2\alpha)\left(-(\frac{1}{2}-z)^2h_2^{\prime\prime}(z)-(1-y)h_2^{\prime\prime}(y)\right)h_2^\prime(s)h_2^{\prime\prime}(s)\end{array}\right)\nonumber\\
&=  \frac{(1-2\alpha)^4\left(h_2^{\prime\prime}(s)\right)^2h_2^{\prime\prime}(z)h_2^{\prime\prime}(y)}{\left(h_2^\prime(z)\right)^4(1-y)}\left(-\frac{1}{4}y(1-y)^2\ln \left(\frac{1-z}{z}\right)r(z)+\frac{(r(s))^2}{(1-2\alpha)^2}-(1-y)^2\frac{r(s)r(z)}{1-2\alpha}\right).\label{convexjoint2}
\end{align}
%where $(a)$ follows from \eqref{convexjoint} because $-h_2^{\prime\prime}(x)$ is decreasing in $x$, $0\leq x\leq \frac{1}{2}$, and since $\alpha\star\frac{q}{2}\leq s$.
It suffices to prove the non-negativity of 
\begin{align}
t(\alpha,y,z)=-\frac{1}{4}y(1-y)^2\ln \left(\frac{1-z}{z}\right)r(z)+\frac{(r(s))^2}{(1-2\alpha)^2}-(1-y)^2\frac{r(s)r(z)}{1-2\alpha}\label{behave}
\end{align}
for every non-negative parameter $\alpha\leq \frac{1}{2}$, $y\leq 1$, $z\leq \frac{1}{2}$. First we show that $t(\alpha,y,z)$ is non-decreasing in $\alpha$, and conclude that $t(\alpha,y,z)$ is non-negative if and only if $t(0,y,z)$ is non-negative. 
By taking the derivative of $t(\alpha,y,z)$ with respect to $\alpha$ we obtain
\begin{align}
\frac{\partial{t(\alpha,y,z)}}{\partial \alpha}=\frac{\left(2r(s)-(1-2\alpha)(1-y)^2r(z)\right)\left(2r(s)+(1-2\alpha)(1-y)(1-2z)r^\prime(s)\right)}{(1-2\alpha)^3}.
\end{align}
Both terms in the numerator are non-negative. The first term is non-negative because it is larger than \eqref{rsrz} and the second term is non-negative because it is a non-increasing function of $z$, $0\leq z\leq \frac{1}{2}$, and is equal to $0$ at $z=\frac{1}{2}$.
%$$t_1(\alpha,y,z)=2r(s)+(1-2\alpha)(1-y)(1-2z)r^\prime(s).$$  At $z=\frac{1}{2}$, we have $t_1(\alpha,y,\frac{1}{2})=0$. Moreover, $t_1(\alpha,y,z)$ is non-increasing in $z$, $0\leq z\leq \frac{1}{2}$.  One can verify this by showing that the derivative of $t_1(\alpha,y,z)$ with respect to $z$ is non-positive.

Finally, we plot $$f(y,z)=\frac{t(0,y,z)}{(1-y)^2(1-2z)^4}$$ in Fig. \ref{nonneg} and demonstrate that the function is non-negative over $0\leq y\leq 1$ and $0\leq z\leq \frac{1}{2}$. The terms in the denominator of $f(y,z)$ capture the behaviour of $t(0,y,z)$ around $z=\frac{1}{2}$ and $y=1$. $f(y,z)$ is zero at $y=0$ with a strictly positive slope at $y=0$ for all $z$, $0< z< \frac{1}{2}$.

 \begin{figure}\centering
 \includegraphics[width=.85\textwidth]{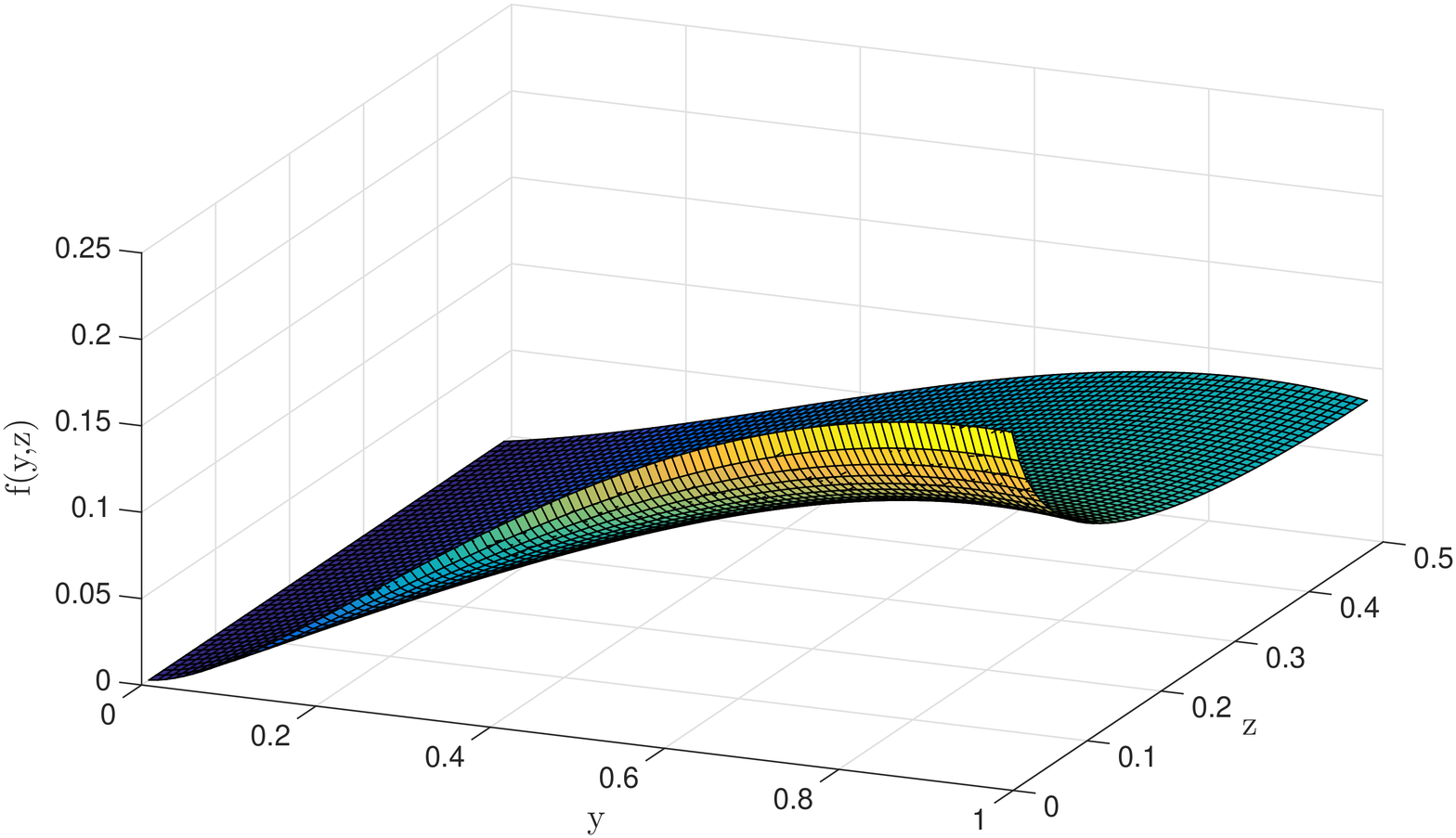}
 \caption{The function $f(y,z)$ is non-negative for $y,z$, $0\leq y\leq 1$, $0\leq z\leq \frac{1}{2}$.}
 \label{nonneg}
 \end{figure}
 
\section{Proof of Theorem \ref{thm:capbinaryadder}}
\label{last theorem}
Consider \eqref{eqlowbin1}-\eqref{eqlowbin3} and let $R_{\max}^{(l)}$ be  the maximum rate admissible. We have
\begin{align}
R_{\max}^{(l)}&=\max_{0\leq p\leq 1}\min\{g_0(C_1,C_2,p),g_2(C_1,p),g_3(p)\}\label{lb}
%R_{\max}^{u}&=\max_{0\leq q\leq \frac{1}{2}}\min_{0\leq \alpha\leq \frac{1}{2}}\min\{g_5(C,q,\alpha),g_1(C,q),g_3(p)\}.\label{ub}
\end{align} 
where %$g_3(p)$, $g_0(C_1,C_2,p)$ and $g_2(C_1,p)$ are defined by
\begin{align}
&g_0(C_1,C_2,p)=C_1+C_2-1+h_2(p)\\
&g_2(C_1,p)=C_1+h_2(p)\\
&g_3(p)=h_2(p)+1-p.
%&g_5(C,q,\alpha)= 2C-h_2\left(\alpha\star h_2^{-1}\left(R-h_2(q)+q\right)\right)-(1-q) h_2(\alpha)-q+2h_2\left((1-q)(1-\alpha)+\frac{1}{2}q\right).
\end{align}
Consider next \eqref{cap1}-\eqref{cap3}, \eqref{cap4} with $q=P_Y(1)$ as in \eqref{qpy1}. Let $R_{\max}^{(u)}$ be the maximum admissible rate. Since the function $h_2(\alpha\star\left(\frac{q}{2}+(1-q)h_2^{-1}(x)\right))$ is non-decreasing in $x$, we have
\begin{align}
R_{\max}^{(u)}\leq C_1+C_2-h_2\left(\alpha\star\left(\frac{q}{2}+(1-q)h_2^{-1}\left(\min\left(1,\frac{\left(R_{\max}^{(l)}-h_2(q)\right)^+}{1-q}\right)\right)\right)\right)-(1-q) h_2(\alpha)-q+2h_2\left(\alpha\star\frac{q}{2}\right).\label{g5}
\end{align}
So the capacity is upper bounded by
\begin{align}
%R_{\max}^{l}&=\max_{0\leq q\leq \frac{1}{2}}\min\{g_0(C,q),g_1(C,q),g_3(p)\}\label{lb}
R_{\max}^{(u)}&\leq\max_{0\leq q\leq 1}\min_{0\leq \alpha\leq \frac{1}{2}}\min\{g_1(C_1,C_2),g_2(C_1,q),g_3(q),g_5(C_1,C_2,q,\alpha)\}\label{ub}
\end{align} 
where $g_1(C_1,C_2)=C_1+C_2$ and $g_5(C_1,C_2,q,\alpha)$ is the RHS of \eqref{g5}.
Since \eqref{ub}  depends only on $q=p_{X_1X_2}(0,1)+p_{X_1X_2}(1,0)$, we may assume, without  loss of generality, that $p(x_1,x_2)$ is a doubly symmetric binary pmf with parameter $q$. Using \eqref{mutualinfoa} the bound in \eqref{g5} may be re-written as
\begin{align}
g_5(C_1,C_2,q,\alpha)&=g_0(C_1,C_2,q)+I(X_1;X_2|U)+1-h_2\left(\alpha\star\left(\frac{q}{2}+(1-q)h_2^{-1}\left(\min\left(1,\frac{\left(R_{\max}^{(l)}-h_2(q)\right)^+}{1-q}\right)\right)\right)\right)
\end{align}
where $I(X_1;X_2|U)$ is a function of $\alpha$ and $q$. We denote this conditional mutual information by $I_{\alpha,q}(X_1;X_2|U)$.

Consider \eqref{lb} and let $\eta$ be the optimizing $p$. We have $\eta\leq\frac{1}{2}$. 
Note that when $C_2\geq 1$, $g_0(C_1,C_2,p)$ is redundant in \eqref{lb} and the cut-set bound is achievable. Otherwise, $g_2(C_1,p)$ is redundant and we have one of the following cases:
\begin{itemize}
\item[($a$)] $\eta=\frac{1}{3}$ and $R_{\max}^{(l)}= g_3(\frac{1}{3})$. In this case, $R_{\max}^{(l)}=\log_2(3)$ and the cut-set bound is achievable.
\item[($b$)] $\eta=\frac{1}{2}$ and  $R_{\max}^{(l)}= g_0(C_1,C_2,\frac{1}{2})$. In this case, $R_{\max}^{(l)}=C_1+C_2$ and the cut-set bound is achievable.
\item[($c$)] $\eta$ is such that $R_{\max}^{(l)}= g_0(C_1,C_2,\eta)=g_3(\eta)$. 
We show that $R_{\max}^{(l)}$ and $R_{\max}^{(u)}$ match in this case. Here, we have
%\begin{align}
%\eta=2-C_1-C_2.\label{etaC}
%\end{align}
%and
\begin{align}
&\eta=2-C_1-C_2\label{etaC}\\
&R_{\max}^{(l)}= g_0(C_1,C_2,\eta)=g_3(\eta).\label{khasteshodam}
\end{align}
 \end{itemize}
\noindent Consider \eqref{ub} in regime ($c$) and let $\alpha_{\eta}$ be the solution of
\begin{align}
\alpha_{\eta}(1-\alpha_{\eta})=\left(\frac{{\eta}}{2(1-{\eta})}\right)^2\label{alphaeta}
\end{align}
that is less than or equal to $\frac{1}{2}$ where $\eta$ is given by \eqref{etaC}.
%\begin{align}
%R_{\max}^{(l)}= g_0(C_1,C_2,\eta)=g_3(\eta).\label{st}
%\end{align}
%This is a valid choice because $\eta\leq \frac{1}{2}$.
With this choice of $\alpha$, we study $g_5(C_1,C_2,q,\alpha_{\eta})$.
At $q=\eta$, $X_1-U-X_2$ forms a Markov chain and using \eqref{khasteshodam} we  have 
\begin{align}
g_5(C_1,C_2,\eta,\alpha_{\eta})=g_0(C_1,C_2,\eta).\label{g5crossestoo}
\end{align} 
In regime ($c$) given by \eqref{khasteshodam}, $g_3(q)$ and $g_0(C_1,C_2,q)$ cross at $q=\eta$. At this point, $g_3(q)$ is decreasing and $g_0(C_1,C_2,q)$ is increasing in $q$.  Also, $g_5(C_1,C_2,q,\alpha_\eta)$  crosses the two curves at $q=\eta$ as shown in \eqref{g5crossestoo}. Therefore, if $g_5(C_1,C_2,q,\alpha_\eta)$ is non-decreasing in $q$, $q\leq\eta$, then $\eta$ maximizes  \eqref{ub} and  $R^{(u)}_{\max}=R^{(l)}_{\max}$, see Fig. \ref{functionscross}.

\begin{figure}[t!]
\centering
\input{functionscross}
\caption{Several functions of $q$. At $q=\eta$, the curves $g_0(C_1,C_2,q)$, $g_3(q)$, and  $g_5(C_1,C_2,q,\alpha_\eta)$ all have the same value.}
\label{functionscross}
\end{figure}

It remains to show that $g_5(C_1,C_2,q,\alpha_{\eta})$ is non-decreasing in $q$, $q\leq \eta$. $g_5(C_1,C_2,q,\alpha_{\eta})$ is continuous and piece-wise concave. We thus look at the following two regimes and prove that $g_5(C_1,C_2,q,\alpha_{\eta})$ is non-decreasing in both regimes: $q\leq \tilde{\eta}$ and $\tilde{\eta}\leq q\leq {\eta}$ where $\tilde{\eta}$ and $\eta$ are the two solutions of  $R_{\max}^{(l)}= 1+h_2(q)-q$. 
\begin{itemize}
\item $q\leq \tilde{\eta}$: Here, we have $R_{\max}^{(l)}\geq 1+h_2(q)-q$ and 
\begin{align}
g_5(C_1,C_2,q,\alpha_\eta)&=g_0(C_1,C_2,q)+I_{\alpha_\eta,q}(X_1;X_2|U)\nonumber\\
&=C_1+C_2-(1-q) h_2(\alpha_\eta)-q+2h_2\left(\alpha_\eta\star\frac{q}{2}\right)-1.
\label{balaI}
\end{align}
 The RHS of \eqref{balaI} is  concave in $q$. By showing that this function is non-decreasing at $q=\eta$ we  prove that it is non-decreasing in $q$, $q\leq \eta$. We first show that $I_{\alpha_\eta,q}(X_1;X_2|U)$ has a zero derivative at $q=\eta$: \begin{align}
\frac{\partial I_{\alpha_\eta,q}(X_1;X_2|U)}{\partial q}%&=h_2(\alpha_{\eta})-1-\log\left(\frac{1-q}{q}\right)-(1-2\alpha_\eta)\log\left( \frac{(1-q)(\alpha_{\eta})+\frac{1}{2}q}{(1-q)(1-\alpha_{\eta})+\frac{1}{2}q}\right)\\
&=h_2(\alpha_{\eta})-1-\log\left(\frac{1-q}{q}\right)-(1-2\alpha_\eta)\log\left(\frac{\alpha_{\eta}+\frac{q}{2(1-q)}}{(1-\alpha_{\eta})+\frac{q}{2(1-q)}}\right)\nonumber\\
&=h_2(\alpha_{\eta})-1-\log\left(\frac{1-q}{q}\right)-(1-2\alpha_\eta)\log\left(\frac{\alpha_{\eta}+\frac{q}{2(1-q)}}{(1-\alpha_{\eta})+\frac{q}{2(1-q)}}\times\frac{\frac{q}{2(1-q)}-\alpha_\eta}{\frac{q}{2(1-q)}-\alpha_\eta}\right)\nonumber\\
&=h_2(\alpha_{\eta})-1-\log\left(\frac{1-q}{q}\right)-(1-2\alpha_\eta)\log\left(\frac{(\frac{q}{2(1-q)})^2-\alpha^2_{\eta}}{(\frac{q}{2(1-q)})^2-\alpha_\eta(1-\alpha_{\eta})+\frac{q}{2(1-q)}(1-2\alpha_\eta)}\right).
\end{align}
We use \eqref{alphaeta} to write
\begin{align}
\left.\frac{\partial I_{\alpha_\eta,q}(X_1;X_2|U)}{\partial q}\right|_{q=\eta}&=
h_2(\alpha_{\eta})-\log\left(\frac{2(1-\eta)}{\eta}\right)-(1-2\alpha_\eta)\log\left(\frac{\alpha_\eta(1-2\alpha_\eta)}{\frac{\eta}{2(1-\eta)}(1-2\alpha_\eta)}\right)\nonumber\\
&=
h_2(\alpha_{\eta})+\frac{1}{2}\log\left(\alpha_\eta(1-\alpha_\eta)\right)-\left(1-2\alpha_\eta\right)\log\left(\alpha_\eta\right)+\left(\frac{1}{2}-\alpha_\eta\right)\log(\alpha_\eta(1-\alpha_\eta))\nonumber\\
&=0.\label{moshtaghzero}
\end{align}
At $q=\eta$, $I_{\alpha_\eta,q}(X_1;X_2|U)$ has a zero derivative  and $g_0(C_1,C_2,q)$ is non-decreasing. So the RHS of \eqref{balaI} is non-decreasing at $q=\eta$, and since it is concave, it is also non-decreasing at all $q$, $q\leq \eta$, and in particular at all $q$, $q\leq \tilde{\eta}$.
\item $\tilde{\eta}\leq q\leq {\eta}$: In this regime, we have
\begin{align}
g_5(C_1,C_2,q,\alpha_\eta)&=g_0(C_1,C_2,q)+I_{\alpha_\eta,q}(X_1;X_2|U)+1-h_2\left(\alpha_\eta\star\left(\frac{q}{2}+(1-q)h_2^{-1}\left(\frac{\left(R_{\max}^{(l)}-h_2(q)\right)^+}{1-q}\right)\right)\right)\label{khozaval2}\\
&=C_1+C_2-(1-q) h_2(\alpha)-q+2h_2\!\left(\alpha\star\frac{q}{2}\right)-h_2\!\left(\!\alpha\star\!\left(\frac{q}{2}+(1-q)h_2^{-1}\!\left(\!\frac{\left(R_{\max}^{(l)}-h_2(q)\right)^+}{1-q}\!\right)\right)\!\right).
\label{khozaval}
\end{align}
The RHS of \eqref{khozaval} is concave in $q$ (see Appendix \ref{concavity lemma}). Furthermore,  \eqref{khozaval} is non-decreasing  at $q=\eta$. To see this, one can take its derivative with respect to $q$ and evaluate it at $q=\eta-\epsilon$, $\epsilon\to 0$:
\begin{align}
\lim_{\epsilon\to 0}\left.\frac{\partial g_5(C_1,C_2,q,\alpha_\eta)}{\partial q}\right|_{q=\eta-\epsilon}&\stackrel{(a)}{=}\log\left(\frac{1-\eta}{\eta}\right)+\left.\frac{\partial I_{\alpha_\eta,q}(X_1;X_2|U)}{\partial q}\right|_{q=\eta}-(1-2\alpha_\eta)^2\left(1-h_2^{\prime}(\eta)\right)(1-\eta)\nonumber\\
&\stackrel{(b)}{=}\log\left(\frac{1-\eta}{\eta}\right)-\left(1-\left(\frac{\eta}{1-\eta}\right)^2\right)\left(1-\log\left(\frac{1-\eta}{\eta}\right)\right)(1-\eta)\nonumber\\
&=\log\left(\frac{1-\eta}{\eta}\right)-\left(1-\frac{\eta}{1-\eta}\right)\left(1-\log\left(\frac{1-\eta}{\eta}\right)\right)
\label{moshtagh}
\end{align}
where $(a)$ follows by differentiating \eqref{khozaval2} with respect to $q$ and evaluating it at $q=\eta$ and $(b)$ follows by \eqref{moshtaghzero} and \eqref{alphaeta}. The function $\log\left(x\right)-(1-\frac{1}{x})\left(1-\log(x)\right)$ is equal to $0$ at $x=1$ and is non-decreasing for $x\geq 1$; therefore, \eqref{moshtagh} is non-negative and $g_5(C_1,C_2,q,\alpha_\eta)$ is non-decreasing in $q$, $\tilde{\eta}\leq q\leq \eta$.
\end{itemize}

\ifCLASSOPTIONcaptionsoff
  \newpage
\fi

% trigger a \newpage just before the given reference
% number - used to balance the columns on the last page
% adjust value as needed - may need to be readjusted if
% the document is modified later
%\IEEEtriggeratref{8}
% The "triggered" command can be changed if desired:
%\IEEEtriggercmd{\enlargethispage{-5in}}

% references section

% can use a bibliography generated by BibTeX as a .bbl file
% BibTeX documentation can be easily obtained at:
% http://www.ctan.org/tex-archive/biblio/bibtex/contrib/doc/
% The IEEEtran BibTeX style support page is at:
% http://www.michaelshell.org/tex/ieeetran/bibtex/
%\bibliographystyle{IEEEtran}
% argument is your BibTeX string definitions and bibliography database(s)
%\bibliography{IEEEabrv,../bib/paper}
%
% <OR> manually copy in the resultant .bbl file
% set second argument of \begin to the number of references
% (used to reserve space for the reference number labels box)

% biography section
\bibliographystyle{IEEEtran}
\bibliography{bibliographyMAC}
% 
% If you have an EPS/PDF photo (graphicx package needed) extra braces are
% needed around the contents of the optional argument to biography to prevent
% the LaTeX parser from getting confused when it sees the complicated
% \includegraphics command within an optional argument. (You could create
% your own custom macro containing the \includegraphics command to make things
% simpler here.)
%\begin{IEEEbiography}[{\includegraphics[width=1in,height=1.25in,clip,keepaspectratio]{mshell}}]{Michael Shell}
% or if you just want to reserve a space for a photo:

%\begin{IEEEbiographynophoto}{Michael Shell}
%Biography text here.
%\end{IEEEbiographynophoto}

% if you will not have a photo at all:
%\begin{IEEEbiographynophoto}{John Doe}
%Biography text here.
%\end{IEEEbiographynophoto}

% insert where needed to balance the two columns on the last page with
% biographies
%\newpage

% You can push biographies down or up by placing
% a \vfill before or after them. The appropriate
% use of \vfill depends on what kind of text is
% on the last page and whether or not the columns
% are being equalized.

%\vfill

% Can be used to pull up biographies so that the bottom of the last one
% is flush with the other column.
%\enlargethispage{-5in}

% that's all folks
\end{document}

%% file: two-user-mac
\begin{figure}[t!]
\centering
\begin{tikzpicture}[scale=1]
\tikzstyle{every node}=[draw,shape=circle];

\path (-1.2,1.5) node[shape=rectangle] (s0) {$\text{Source}$};
\path (1,1.5) node[shape=rectangle] (s) {$\text{Encoder}$};

\path (2.5,3) node[shape=rectangle] (r2) {$\text{Relay\hspace{-.05cm} 1}$};
\path (2.5,0) node[shape=rectangle] (r1) {$\text{Relay\hspace{-.05cm} 2}$};
\path (5,1.5) node[shape=rectangle,minimum height=3.5cm] (mac) {$\hspace{-.1cm}\begin{array}{c}\text{MAC}\\p(\hspace{-.05cm}y|x_1\hspace{-.05cm},\hspace{-.05cm}x_2\hspace{-.05cm})\end{array}\hspace{-.05cm}$};

\path (7.5,1.5) node[shape=rectangle] (d) {$\text{Decoder}$};
\path (9.51,1.5) node[shape=rectangle] (ds) {$\text{Sink}$};

\path (4.25,0) node[draw=none] (inp1) {};
\path (4.25,3) node[draw=none] (inp2) {};
\path (5.75,1.5) node[draw=none] (out) {};

%\draw[->]
  %  (v0) -- (v1) ;
\draw[->] (s0) --node[draw=none,yshift=.2cm]{$W$} (s);  
\draw[->] (d) --node[draw=none,yshift=.2cm]{$\hat{W}$} (ds);  
\draw[->] (s) --node[left,draw=none]{$V^n_2$} (r1);
\draw[->] (s) --node[left,draw=none]{$V^n_1$} (r2);
\draw[->] (r2) --node[draw=none,yshift=.2cm]{$X^n_1$} (inp2);
\draw[->] (r1) --node[draw=none,yshift=.2cm]{$X^n_2$} (inp1);
\draw[->] (out) --node[draw=none,yshift=.2cm]{$Y^n$} (d);

\end{tikzpicture}
\caption{Problem setup.}
\label{two-user-mac}
\end{figure}

%% file: seclowerbound.tex
\section{Lower bound}
\label{secach}
Our coding scheme is based on \cite{TraskovKramer07}, but we further send a common message to both relaying nodes. This is  done by rate splitting, superposition coding and Marton's coding and is summarized in the following theorem. %\footnote{A similar lower bound  appears  independently  in \cite{KangLiuChong15}}.
\begin{theorem}
\label{lower bound}
The rate $R$ is achievable if it satisfies the following condition for some pmf $p(u,x_1,x_2,y)=p(u,x_1,x_2) p(y|x_1,x_2)$, and $U\in\mathcal{U}$ with $|\mathcal{U}|\leq \min\{|\mathcal{X}_1||\mathcal{X}_2|+3,|\mathcal{Y}|+4\}$.
\begin{align}
\label{inner}
R \leq  \min  \left\{    \begin{array}{l}C_1+C_2-I(X_1;X_2|U),\\C_2+I(X_1;Y|X_2U),\\C_1+I(X_2;Y|X_1U),\\\frac{1}{2}(C_1 + C_2 + I(X_1X_2;Y|U) - I(X_1;X_2|U) ),\\I(X_1X_2;Y)\end{array}    \right\}
\end{align}
\end{theorem}
\begin{remark}
If $U$ is a constant then the fourth bound in \eqref{inner} is redundant as it is half the sum of the first and fifth bounds. This shows that Theorem \ref{lower bound} without a $U$ reduces to \cite[Theorem 1]{TraskovKramer07}. $U$ turns out to be useful for Gaussian MACS, as shown in Fig.~\ref{closeup}.  Theorem \ref{lower bound} appeared in \cite[Theorem 2]{SaeediKramer14}  and also in \cite[Theorem 2]{KangLiuChong15}. 
\end{remark}
\begin{remark}
One could add a time-sharing random variable Q to \eqref{inner}. However, by combining $Q$ with $U$, one can check that Theorem \ref{lower bound} is at least as large as this region.
\end{remark}

\begin{proof}[Sketch of proof]
\ %The inner bound in Theorem \ref{lower bound} is found using the following achievable scheme.
\paragraph{Codebook construction}
Fix the joint pmf $p(u,x_1,x_2)$ and $R_{12},R_1,R_2,R_1^\prime,R_2^\prime \geq 0$. Let
\begin{align}
\label{ratesplit}
R=R_{12}+R_1+R_2.
\end{align}
Generate $2^{nR_{12}}$ sequences $u^n(m_{12})$  independently, each in an i.i.d manner according to $\prod_lP_{U}(u_l)$.
For each sequence $u^n(m_{12})$, generate (i) $2^{n(R_1+R_1^\prime)}$ sequences $x^n_1(m_{12},m_1,m^\prime_1)$, $m_1=1,\ldots,2^{nR_1}$, $m_1^\prime=1,\ldots,2^{nR_1^\prime}$, conditionally independently, each in an i.i.d manner according to $\prod_lP_{X_1|U}(x_{1,l}|u_l(m_{12}))$ and (ii) $2^{n(R_2+R_2^\prime)}$  sequences $x^n_2(m_{12},m_2,m^\prime_2)$, $m_2=1,\ldots,2^{nR_2}$, $m_2^\prime=1,\ldots,2^{nR_2^\prime}$,  conditionally independently, each in an i.i.d manner according to $\prod_lP_{X_2|U}(x_{2,l}|u_l(m_{12}))$.
For each bin index $(m_{12},m_1,m_2)$, pick a sequence pair $(x^n_1(m_{12},m_1,m^\prime_1),x^n_2(m_{12},m_2,m^\prime_2))$ that is jointly typical. 

\paragraph{Encoding}
To communicate message $W=(m_{12},m_1,m_2)$, communicate $(m_{12},m_1,m^\prime_1)$ to relay $1$  and $(m_{12},m_2,m^\prime_2)$ to relay $2$; here $(x^n_1(m_{12},m_1,m_1^\prime),x^n_2(m_{12},m_2,m_2^\prime))$ is the  jointly typical pair picked in the bin indexed by $(m_{12},m_1,m_2)$.
\paragraph{Decoding}
Upon receiving $y^n$, the receiver looks for indices $\hat{m}_{12},\hat{m}_1,\hat{m}_2$ for which the following tuple is jointly typical for some $\hat{m}^\prime_1$, $\hat{m}^\prime_2$: $$(u^n(\hat{m}_{12}),x^n_1(\hat{m}_{12},\hat{m}_1,\hat{m}^\prime_1),x^n_2(\hat{m}_{12},\hat{m}_2,\hat{m}^\prime_2),y^n)\in\mathcal{T}^n_\epsilon.$$
\paragraph{Error Analysis}
The error analysis is standard and is deferred to Appendix \ref{aperroranalysis}.  {Eliminating} $R_{12},R_1,R_2,R_1^\prime,R_2^\prime$ by Fourier-Motzkin elimination, we arrive at Theorem \ref{lower bound}. Cardinality bounds follow by using the standard method via the Fenchel-Eggleston-Carath\'{e}odory theorem \cite[Appendix~C]{ElGamalKim}, \cite[Appendix~B]{cribbing85}.
\end{proof}
\begin{proposition}
\label{thmconcave}
The lower bound of Theorem \ref{lower bound} is concave in $C_1,C_2$.
\end{proposition}
\begin{proof}
We prove the statement for $C_1=C_2=C$. The same argument holds in general.
We express the lower bound of Theorem \ref{lower bound} in terms of the following maximization problem. 
\begin{align}
f_\ell(C)=\max_{p(u,x_1,x_2)}f^p_\ell(C,p(u,x_1,x_2))
\end{align} 
In this formulation, $f^p_\ell(C,p(u,x_1,x_2))$ is  the minimum term on the right hand side (RHS) of \eqref{inner}.

The proof  is by contradiction. Suppose that the lower bound is not concave in $C$; i.e., there exist  values $C^{(1)},C^{(2)}$, and $\alpha$, $0\leq\alpha\leq 1$, such that $C^\star=\alpha C^{(1)}+(1-\alpha)C^{(2)}$ and $f_\ell(C^\star)<\alpha f_\ell(C^{(1)})+(1-\alpha)f_\ell(C^{(2)})$. 
Let $p^{(1)}(u,x_1,x_2)$ (resp. $p^{(2)}(u,x_1,x_2)$) be  the pmf that maximizes $f^p_\ell(C^{(1)},p(u,x_1,x_2))$ (resp. $f^p_\ell(C^{(2)},p(u,x_1,x_2))$). Let $p_Q(1)=\alpha$, $p_Q(2)=1-\alpha$, and define $p_{UX_1X_2|Q}(u,x_1,x_2|1)=p^{(1)}(u,x_1,x_2)$ and $p_{UX_1X_2|Q}(u,x_1,x_2|2)=p^{(2)}(u,x_1,x_2)$. Then we have
\begin{align}
&f_\ell(C^\star)\nonumber\\
&<\alpha f_\ell(C^{(1)})+(1-\alpha)f_\ell(C^{(2)})\nonumber\\
&\leq\min  \left\{   \begin{array}{l}2C^\star-I(X_1;X_2|UQ),\\C^\star+I(X_1;Y|X_2UQ),\\C^\star+I(X_2;Y|X_1UQ),\\\frac{1}{2}(2C^\star + I(X_1X_2;Y|UQ) - I(X_1;X_2|UQ)),\\I(UX_1X_2;Y|Q)\end{array}     \right\}\nonumber\\
&\leq\min  \left\{   \begin{array}{l}2C^\star-I(X_1;X_2|UQ),\\C^\star+I(X_1;Y|X_2UQ),\\C^\star+I(X_2;Y|X_1UQ),\\\frac{1}{2}(2C^\star + I(X_1X_2;Y|UQ) - I(X_1;X_2|UQ)),\\I(UX_1X_2;Y)\end{array}     \right\}\label{referenceRHS}\\
&\stackrel{(a)}{\leq} f_\ell(C^\star)\label{referenceLHS}.
\end{align}
Step $(a)$ follows by renaming $(U,Q)$ a $U$ and comparing \eqref{referenceRHS} with the lower bound  of Theorem~\ref{lower bound}. %So we have a contradiction, and the proposition is proved.
\end{proof}

%% file: dg.tex
\begin{figure}[t!]
\centering
\begin{tikzpicture}[scale=1.25,arrowmark/.style 2 args={decoration={markings,mark=at position #1 with {\arrow{#2}}}}]
\tikzstyle{every node}=[draw,shape=circle,minimum size=.2pt]; 
node[fill,circle,inner sep=0pt,minimum size=1pt]
\path (3.5,4) node[fill=none,scale=0.5] (w) {};
\path (3.5,4.25) node[draw=none] (ww) {$W$};
\path (0,3) node[fill,scale=0.5] (x11) {};
\path (-.5,3) node[draw=none] (x11x11) {$X_{11}$};
\path (1,3) node[fill,scale=0.5] (x21) {};
\path (1.5,3) node[draw=none] (x21x21) {$X_{21}$};
\path (3,3) node[fill,scale=0.5] (x12) {};
\path (2.5,3) node[draw=none] (x12x12) {$X_{12}$};
\path (4,3) node[fill,scale=0.5] (x22) {};
\path (4.5,3) node[draw=none] (x22x22) {$X_{22}$};
\path (6,3) node[fill,scale=0.5] (x13) {};
\path (5.5,3) node[draw=none] (x13x13) {$X_{13}$};
\path (7,3) node[fill,scale=0.5] (x23) {};
\path (7.5,3) node[draw=none] (x23x23) {$X_{23}$};
\path (-.5,2) node[fill,scale=0.5]  (z1) {};
\path (2.5,2) node[fill,scale=0.5]  (z2) {};
\path (5.5,2) node[fill,scale=0.5]  (z3) {};
\path (-1,2) node[draw=none] (z1z1) {$Z_1$};
\path (2,2) node[draw=none] (z2z2) {$Z_2$};
\path (5,2) node[draw=none] (z3z3) {$Z_3$};
\path (0.5,2) node[fill,scale=0.5] (y1) {};
\path (3.5,2) node[fill,scale=0.5] (y2) {};
\path (6.5,2) node[fill,scale=0.5] (y3) {};
\path (1,2) node[draw=none] (y1y1) {$Y_1$};
\path (4,2) node[draw=none] (y2y2) {$Y_2$};
\path (7,2) node[draw=none] (y3y3) {$Y_3$};
\path (8,2) node[draw=none] (y4) {$\ldots$};
\path (-.5,1) node[fill,scale=0.5]  (zz1) {};
\path (2.5,1) node[fill,scale=0.5]  (zz2) {};
\path (5.5,1) node[fill,scale=0.5]  (zz3) {};
\path (-1,1) node[draw=none] (zz1zz1) {$Z^\prime_1$};
\path (2,1) node[draw=none] (zz2zz2) {$Z^\prime_2$};
\path (5,1) node[draw=none] (zz3zz3) {$Z^\prime_3$};
\path (.5,1) node[fill,scale=0.5] (u1) {};
\path (3.5,1) node[fill,scale=0.5] (u2) {};
\path (6.5,1) node[fill,scale=0.5] (u3) {};
\path (1,1) node[draw=none] (u1u1) {$U_1$};
\path (4,1) node[draw=none] (u2u2) {$U_2$};
\path (7,1) node[draw=none] (u3u3) {$U_3$};
\draw[postaction={decorate},
    arrowmark={.5}{stealth},
    ]  (w) --(x11) ;
\draw[postaction={decorate},
    arrowmark={.5}{stealth},
    ]  (w) --(x12) ;
\draw[postaction={decorate},
    arrowmark={.5}{stealth},
    ]  (w) --(x21) ;
\draw[postaction={decorate},
    arrowmark={.5}{stealth},
    ]  (w) --(x22) ;
\draw[postaction={decorate},
    arrowmark={.5}{stealth},
    ]  (w) --(x13) ;
\draw[postaction={decorate},
    arrowmark={.5}{stealth},
    ]  (w) --(x23) ;
\draw[postaction={decorate},
    arrowmark={.5}{stealth},
    ]  (x11) --(y1) ;
\draw[postaction={decorate},
    arrowmark={.5}{stealth},
    ]  (x21) --(y1) ;
\draw[postaction={decorate},
    arrowmark={.5}{stealth},
    ]  (z1) --(y1) ;
\draw[postaction={decorate},
    arrowmark={.5}{stealth},
    ]  (x12) --(y2) ;
\draw[postaction={decorate},
    arrowmark={.5}{stealth},
    ]  (x22) --(y2) ;
\draw[postaction={decorate},
    arrowmark={.5}{stealth},
    ]  (z2) --(y2) ;
\draw[postaction={decorate},
    arrowmark={.5}{stealth},
    ]  (x13) --(y3) ;
\draw[postaction={decorate},
    arrowmark={.5}{stealth},
    ]  (x23) --(y3) ;
\draw[postaction={decorate},
    arrowmark={.5}{stealth},
    ]  (z3) --(y3) ;
%\draw[->] (x11) --(u1) ;
%\draw[->] (x21) --(u1) ;
\draw[postaction={decorate},
    arrowmark={.5}{stealth},
    ]  (zz1) --(u1) ;
\draw[postaction={decorate},
    arrowmark={.5}{stealth},
    ]  (y1) --(u1) ;
%\draw[->] (x12) --(u2) ;
%\draw[->] (x22) --(u2) ;
\draw[postaction={decorate},
    arrowmark={.5}{stealth},
    ]  (zz2) --(u2) ;
\draw[postaction={decorate},
    arrowmark={.5}{stealth},
    ] (y2) --(u2) ;
\draw[postaction={decorate},
    arrowmark={.5}{stealth},
    ] (zz3) --(u3) ;
\draw[postaction={decorate},
    arrowmark={.5}{stealth},
    ] (y3) --(u3) ;
%\draw[->] (v5) -- (d1);
\end{tikzpicture}
\caption{The associated FDG for Theorem \ref{theoremupp} and $n=3$. The random variables $Z_1,Z_2,Z_3$ and $Z^\prime_1,Z^\prime_2,Z^\prime_3$ are appropriate noise random variables with the distributions $p_Z(.)$ and $p_{Z^\prime}(.)$, respectively.}
\label{dg}
\end{figure}

%% file: Gaussianplot.tex
\begin{tikzpicture} [every pin/.style={fill=white}]
  \begin{axis}[scale=1.05,
width=4.82222222222222in,
%width=0.75\columnwidth,
height=3.80333333333333in,
scale only axis,
xmin=0.35,
xmax=1.2,
xmajorgrids,
xlabel={Link Capacity $C$},
ymin=.7,
ymax=1.2,
ymajorgrids,
ylabel={Rate $R$},
axis x line*=bottom,
axis y line*=left,
legend pos=north west,
legend style={draw=none,fill=none,legend cell align=left, font=\small}
]
%\addplot +[mark=none] table [x index=0, y index=1] {\datatable};

     \addplot[color=black,dashed,line width=1.0pt]
 table[row sep=crcr]{
   0.350000000000000   0.700000000000000\\
   0.360000000000000   0.720000000000000\\
   0.370000000000000   0.740000000000000\\
   0.380000000000000   0.760000000000000\\
   0.390000000000000   0.780000000000000\\
   0.400000000000000   0.800000000000000\\
   0.410000000000000   0.820000000000000\\
   0.420000000000000   0.840000000000000\\
   0.430000000000000   0.860000000000000\\
   0.440000000000000   0.880000000000000\\
   0.450000000000000   0.900000000000000\\
   0.460000000000000   0.920000000000000\\
   0.470000000000000   0.931969225211986\\
   0.480000000000000   0.938086966675790\\
   0.490000000000000   0.944137074889725\\
   0.500000000000000   0.949974493094836\\
   0.510000000000000   0.955646983185310\\
   0.520000000000000   0.961100049222028\\
   0.530000000000000   0.966670375778095\\
   0.540000000000000   0.972054399553977\\
   0.550000000000000   0.977259576353727\\
   0.560000000000000   0.982271565792679\\
   0.570000000000000   0.987264656241941\\
   0.580000000000000   0.992102459590818\\
   0.590000000000000   0.996941274102052\\
   0.600000000000000   1.001603183651737\\
   0.610000000000000   1.006105541975437\\
   0.620000000000000   1.010739863705226\\
   0.630000000000000   1.014991433107857\\
   0.640000000000000   1.019354125942146\\
   0.650000000000000   1.023590995831836\\
   0.660000000000000   1.027656669828330\\
   0.670000000000000   1.031751471153079\\
   0.680000000000000   1.035789741772150\\
   0.690000000000000   1.039646383636427\\
   0.700000000000000   1.043391883071033\\
   0.710000000000000   1.047118034922883\\
   0.720000000000000   1.050825038047217\\
   0.730000000000000   1.054513088249184\\
   0.740000000000000   1.057889266357215\\
   0.750000000000000   1.061501977027410\\
   0.760000000000000   1.064806369065235\\
   0.770000000000000   1.068095693143572\\
   0.780000000000000   1.071370086058041\\
   0.790000000000000   1.074629682749811\\
   0.800000000000000   1.077874616338972\\
   0.810000000000000   1.080782628198516\\
   0.820000000000000   1.084000062258910\\
   0.830000000000000   1.086883533868353\\
   0.840000000000000   1.089755525135755\\
   0.850000000000000   1.092616127115080\\
   0.860000000000000   1.095465429781307\\
   0.870000000000000   1.098303522047409\\
   0.880000000000000   1.100953650974273\\
   0.890000000000000   1.103634086556014\\
   0.900000000000000   1.106188079569774\\
   0.910000000000000   1.108925548405566\\
   0.920000000000000   1.111402280522631\\
   0.930000000000000   1.113870537919564\\
   0.940000000000000   1.116330378395137\\
   0.950000000000000   1.118781859158838\\
   0.960000000000000   1.120989438429522\\
   0.970000000000000   1.123356050265931\\
   0.980000000000000   1.125706889644732\\
   0.990000000000000   1.127901418301359\\
   1.000000000000000   1.130012827980728\\
   1.010000000000000   1.132194515163256\\
   1.020000000000000   1.134424733830899\\
   1.030000000000000   1.136509364915759\\
   1.040000000000000   1.138444504643433\\
   1.050000000000000   1.140478156915528\\
   1.060000000000000   1.142553125530627\\
   1.070000000000000   1.144541216312822\\
   1.080000000000000   1.146390874613923\\
   1.090000000000000   1.148282854259021\\
   1.100000000000000   1.150208313918437\\
   1.110000000000000   1.151963418240347\\
   1.120000000000000   1.153714262596124\\
   1.130000000000000   1.155460867615172\\
   1.140000000000000   1.157203253777410\\
   1.150000000000000   1.158941441414707\\
   1.160000000000000   1.160675450712308\\
   1.170000000000000   1.160964047443681\\
   1.180000000000000   1.160964047443681\\
   1.190000000000000   1.160964047443681\\
   1.200000000000000   1.160964047443681\\};
       \addlegendentry{Cut-Set bound}
       
      \addplot[color=cyan,dashed,line width=1.0pt]
 table[row sep=crcr]{
 0.350000000000000   0.700000000000000\\
   0.375000000000000   0.750000000000000\\
 0.400000000000000   0.799815311393616\\
   0.425000000000000   0.841735946427614\\
   0.450000000000000   0.875429028580649\\
   0.475000000000000   0.903445297804259\\
   0.500000000000000   0.927596203979577\\
   0.525000000000000   0.948426536448353\\
   0.550000000000000   0.966786319130512\\
   0.575000000000000   0.983215233653149\\
   0.600000000000000   0.997832704868055\\
   0.625000000000000   1.011067055345495\\
   0.650000000000000   1.022721485380584\\
   0.675000000000000   1.033475121962314\\
   0.700000000000000   1.043391883071033\\
   0.725000000000000   1.052319530899368\\
   0.750000000000000   1.060507700480683\\
   0.775000000000000   1.068095693143572\\
   0.800000000000000   1.074954833767462\\
   0.825000000000000   1.081105018157154\\
   0.850000000000000   1.086883533868353\\
   0.875000000000000   1.092298842077880\\
   0.900000000000000   1.097234145717459\\
   0.925000000000000   1.101757205565379\\
   0.950000000000000   1.105817626663575\\
   0.975000000000000   1.109855319774390\\
   1.000000000000000   1.113562467224963\\
   1.025000000000000   1.116637269495467\\
   1.050000000000000   1.119796622506395\\
   1.075000000000000   1.122747831409960\\
   1.100000000000000   1.125480786766609\\
   1.125000000000000   1.127901418301359\\
   1.150000000000000   1.130012827980728\\
    1.175000000000000   1.132118075521773\\
   1.200000000000000   1.134217196788192   \\
        };
    \addlegendentry{Lower bound of \cite{TraskovKramer07}, \cite{KangLiu11}}
    
      \addplot[color=red,solid,line width=1.0pt]
 table[row sep=crcr]{
 0.350000000000000   0.700000000000000\\
   0.375000000000000   0.750000000000000\\
 0.400000000000000   0.799815311393616\\
   0.425000000000000   0.841735946427614\\
   0.450000000000000   0.875429028580649\\
   0.475000000000000   0.903445297804259\\
   0.500000000000000   0.927596203979577\\
   0.525000000000000   0.948426536448353\\
   0.550000000000000   0.966786319130512\\
   0.575000000000000   0.983215233653149\\
   0.600000000000000   0.997832704868055\\
   0.625000000000000   1.011067055345495\\
   0.650000000000000   1.022721485380584\\
   0.675000000000000   1.033475121962314\\
   0.700000000000000   1.043537149795216\\
   0.725000000000000   1.052453931835881\\
   0.750000000000000   1.060633934151871\\
   0.775000000000000   1.068100301789480\\
   0.800000000000000   1.074973252710512\\
%      0.825000000000000   1.081105018157154\\
   0.850000000000000   1.087463474003065\\
  % 0.875000000000000   1.095\\
   0.900000000000000   1.099781435870653\\
%   0.92500000000000   1.104\\
   0.950000000000000   1.111643014454234\\
      0.97500000000000   1.117848953630963\\   
   1.000000000000000   1.123899497178737\\
      1.02500000000000   1.129438013192025\\
       1.050000000000000  1.135503091330361\\
  % 1.07500000000000   1.14\\
   1.100000000000000   1.146656586545616\\
  % 1.12500000000000   1.151\\
  % 1.150000000000000   1.158\\
  1.160964047443681 1.160964047443681\\
   1.175000000000000   1.160964047443681\\
   1.200000000000000   1.160964047443681\\
          };
    \addlegendentry{Lower bound of Theorem \ref{lower bound} (Mixture of two Gaussian dist.)}

      \addplot[color=red,dotted,line width=1.0pt]
 table[row sep=crcr]{
 0.350000000000000   0.700000000000000\\
   0.375000000000000   0.750000000000000\\
 0.400000000000000   0.799815311393616\\
   0.425000000000000   0.841735946427614\\
   0.450000000000000   0.875429028580649\\
   0.475000000000000   0.903445297804259\\
   0.500000000000000   0.927596203979577\\
   0.525000000000000   0.948426536448353\\
   0.550000000000000   0.966786319130512\\
   0.575000000000000   0.983215233653149\\
   0.600000000000000   0.997832704868055\\
   0.625000000000000   1.011067055345495\\
   0.650000000000000   1.022721485380584\\
   0.675000000000000   1.033475121962314\\
   0.700000000000000   1.043537149795216\\
   0.710000000000000   1.047206582320764\\
   0.720000000000000   1.050746355248836\\
   0.730000000000000   1.054161508422926\\
   0.740000000000000   1.057456960471423\\
   0.750000000000000   1.060633934151871\\
   0.760000000000000   1.063704406221050\\
   0.770000000000000   1.066668873916290\\
   0.780000000000000   1.069531729662670\\
   0.790000000000000   1.072299478186974\\
   0.800000000000000   1.074973252710512\\
   0.810000000000000   1.077559099151186\\
   0.820000000000000   1.080058095777557\\
   0.830000000000000   1.082475651268153\\
   0.840000000000000   1.084814336160255\\
   0.850000000000000   1.087076235237165\\
   0.875000000000000   1.092417369712015\\
   0.900000000000000   1.097430166693560\\
   0.925000000000000   1.102569759005018\\
   0.950000000000000   1.107907401454996\\
   0.975000000000000   1.113444340058519\\
   1.000000000000000   1.119179995392207\\
   1.025000000000000   1.125117343255547\\
   1.050000000000000   1.131255653521982\\
   1.075000000000000   1.137601647857240\\
   1.100000000000000   1.144147366044726\\
   1.125000000000000    1.150897218319399\\
   1.150000000000000   1.157845904974804\\
   1.160964047443681   1.160964047443681\\
   1.175000000000000   1.160964047443681\\
   1.200000000000000   1.160964047443681\\
          };
    \addlegendentry{Lower bound of Theorem \ref{lower bound} (Joint Gaussian dist.)}

 \addplot[color=blue,solid,line width=1.0pt]
 table[row sep=crcr]{
 0.350000000000000   0.700000000000000\\
   0.375000000000000   0.750000000000000\\
   0.400000000000000   0.800000000000000\\
   0.425000000000000   0.850000000000000\\
   0.450000000000000   0.878656173951344\\
   0.475000000000000   0.903656173951344\\
   0.500000000000000   0.927596203979577\\
   0.525000000000000   0.948444628504222\\
   0.550000000000000   0.966846326354730\\
   0.575000000000000   0.983215233653149\\
   0.600000000000000   0.997832704868055\\
   0.625000000000000   1.011081087581264\\
   0.650000000000000   1.022847936219154\\
   0.675000000000000   1.033631641631333\\
   0.700000000000000   1.043403426152151\\
   0.725000000000000   1.052607979642753\\
   0.750000000000000   1.061513233918666\\
   0.775000000000000   1.069745876666619\\
   0.800000000000000   1.077885620608923\\
   0.825000000000000   1.085293920228713\\
   0.850000000000000   1.092626908784928\\
   0.875000000000000   1.099571999368729\\
   0.900000000000000   1.106157939549759\\
   0.925000000000000   1.112647951471246\\
   0.950000000000000   1.118792256751611\\
   0.975000000000000   1.124581266122787\\
   1.000000000000000   1.130023064943771\\
   1.025000000000000   1.135424118376827\\
   1.050000000000000   1.140488246433633\\
   1.075000000000000   1.145517070132512\\
   1.100000000000000   1.150218268255137\\
   1.125000000000000   1.154597987820095\\
   1.150000000000000   1.158951275964778\\
   1.175000000000000   1.160685261650175\\
   1.200000000000000   1.160685261650175\\
        };
    \addlegendentry{Upper bound of Theorem \ref{thmupp}}
 \addplot[color=green!50!black,dashed,line width=1.0pt,mark=+]
 table[row sep=crcr]{
 0.350000000000000   0.700000000000000\\
   0.375000000000000   0.750000000000000\\
 0.400000000000000   0.799815311393616\\
   0.425000000000000   0.841735946427614\\
   0.450000000000000   0.875429028580649\\
   0.475000000000000   0.903445297804259\\
   0.500000000000000   0.927596203979577\\
   0.525000000000000   0.948426536448353\\
   0.550000000000000   0.966786319130512\\
   0.575000000000000   0.983215233653149\\
   0.600000000000000   0.997832704868055\\
   0.625000000000000   1.011067055345495\\
   0.650000000000000   1.022721485380584\\
   0.675000000000000   1.033475121962314\\
    0.700000000000000   1.043391883071033\\
   0.725000000000000   1.052628280897330\\
   0.750000000000000   1.061501977027410\\
   0.775000000000000   1.069734747519398\\
   0.800000000000000   1.077874616338972\\
   0.825000000000000   1.085283028396034\\
   0.850000000000000   1.092616127115080\\
   0.875000000000000   1.099561321006800\\
   0.900000000000000   1.106188079569774\\
   0.925000000000000   1.112637464934785\\
   0.950000000000000   1.118781859158838\\
   0.975000000000000   1.124570951640174\\
   1.000000000000000   1.130012827980728\\
   1.025000000000000   1.135413957776832\\
   1.050000000000000   1.140478156915528\\
   1.075000000000000   1.145507050708651\\
   1.100000000000000   1.150208313918437\\
   1.125000000000000   1.154588093739257\\
 1.150000000000000   1.158941441414707\\
   1.175000000000000   1.160964047443681\\
   1.200000000000000   1.160964047443681\\};
    \addlegendentry{Upper bound of Theorem \ref{theoremupp}}
\coordinate (pt) at (axis cs:0.45,.875);

\end{axis}

\node[pin={[pin distance=3.7cm]0:{%
    \begin{tikzpicture}[baseline,trim axis left,trim axis right]
    \begin{axis}[
%        tiny,
%      xlabel={C},
%      ylabel={R},
%      x unit={eV},
      xmin=0.3962, xmax=0.4807,
      ymin=0.79,ymax=0.91,
      enlargelimits,
    ]
 \addplot[color=black,dashed,line width=1.0pt]
 table[row sep=crcr]{
   0.350000000000000   0.700000000000000\\
   0.360000000000000   0.720000000000000\\
   0.370000000000000   0.740000000000000\\
   0.380000000000000   0.760000000000000\\
   0.390000000000000   0.780000000000000\\
   0.400000000000000   0.800000000000000\\
   0.410000000000000   0.820000000000000\\
   0.420000000000000   0.840000000000000\\
   0.430000000000000   0.860000000000000\\
   0.440000000000000   0.880000000000000\\
   0.450000000000000   0.900000000000000\\
   0.460000000000000   0.920000000000000\\
   0.470000000000000   0.931969225211986\\
   0.480000000000000   0.938086966675790\\
   0.490000000000000   0.944137074889725\\
   0.500000000000000   0.949974493094836\\
   0.510000000000000   0.955646983185310\\
   0.520000000000000   0.961100049222028\\
   0.530000000000000   0.966670375778095\\
   0.540000000000000   0.972054399553977\\
   0.550000000000000   0.977259576353727\\
   0.560000000000000   0.982271565792679\\
   0.570000000000000   0.987264656241941\\
   0.580000000000000   0.992102459590818\\
   0.590000000000000   0.996941274102052\\
   0.600000000000000   1.001603183651737\\
   0.610000000000000   1.006105541975437\\
   0.620000000000000   1.010739863705226\\
   0.630000000000000   1.014991433107857\\
   0.640000000000000   1.019354125942146\\
   0.650000000000000   1.023590995831836\\
   0.660000000000000   1.027656669828330\\
   0.670000000000000   1.031751471153079\\
   0.680000000000000   1.035789741772150\\
   0.690000000000000   1.039646383636427\\
   0.700000000000000   1.043391883071033\\
   0.710000000000000   1.047118034922883\\
   0.720000000000000   1.050825038047217\\
   0.730000000000000   1.054513088249184\\
   0.740000000000000   1.057889266357215\\
   0.750000000000000   1.061501977027410\\
   0.760000000000000   1.064806369065235\\
   0.770000000000000   1.068095693143572\\
   0.780000000000000   1.071370086058041\\
   0.790000000000000   1.074629682749811\\
   0.800000000000000   1.077874616338972\\
   0.810000000000000   1.080782628198516\\
   0.820000000000000   1.084000062258910\\
   0.830000000000000   1.086883533868353\\
   0.840000000000000   1.089755525135755\\
   0.850000000000000   1.092616127115080\\
   0.860000000000000   1.095465429781307\\
   0.870000000000000   1.098303522047409\\
   0.880000000000000   1.100953650974273\\
   0.890000000000000   1.103634086556014\\
   0.900000000000000   1.106188079569774\\
   0.910000000000000   1.108925548405566\\
   0.920000000000000   1.111402280522631\\
   0.930000000000000   1.113870537919564\\
   0.940000000000000   1.116330378395137\\
   0.950000000000000   1.118781859158838\\
   0.960000000000000   1.120989438429522\\
   0.970000000000000   1.123356050265931\\
   0.980000000000000   1.125706889644732\\
   0.990000000000000   1.127901418301359\\
   1.000000000000000   1.130012827980728\\
   1.010000000000000   1.132194515163256\\
   1.020000000000000   1.134424733830899\\
   1.030000000000000   1.136509364915759\\
   1.040000000000000   1.138444504643433\\
   1.050000000000000   1.140478156915528\\
   1.060000000000000   1.142553125530627\\
   1.070000000000000   1.144541216312822\\
   1.080000000000000   1.146390874613923\\
   1.090000000000000   1.148282854259021\\
   1.100000000000000   1.150208313918437\\
   1.110000000000000   1.151963418240347\\
   1.120000000000000   1.153714262596124\\
   1.130000000000000   1.155460867615172\\
   1.140000000000000   1.157203253777410\\
   1.150000000000000   1.158941441414707\\
   1.160000000000000   1.160675450712308\\
   1.170000000000000   1.160964047443681\\
   1.180000000000000   1.160964047443681\\
   1.190000000000000   1.160964047443681\\
   1.200000000000000   1.160964047443681\\};
%       \addlegendentry{Cut-Set bound}
       
      \addplot[color=cyan,dashed,line width=1.0pt]
 table[row sep=crcr]{
  0.350000000000000   0.700000000000000\\
   0.360000000000000   0.720000000000000\\
   0.370000000000000   0.740000000000000\\
   0.380000000000000   0.760000000000000\\
   0.390000000000000   0.780000000000000\\
   0.400000000000000   0.799828904932796\\
   0.410000000000000   0.817903128132570\\
   0.420000000000000   0.834214527887310\\
   0.430000000000000   0.849093237748129\\
   0.440000000000000   0.862782103214840\\
   0.450000000000000   0.875457603749385\\
   0.460000000000000   0.887251222554830\\
   0.470000000000000   0.898298877024153\\
   0.480000000000000   0.908647900100313\\
   0.490000000000000   0.918398365037833\\
   0.500000000000000   0.927596203979577\\
   0.510000000000000   0.936299755046546\\
   0.520000000000000   0.944542049299075\\
   0.530000000000000   0.952381294499923\\
   0.540000000000000   0.959841656698732\\
   0.550000000000000   0.966937370857678\\
   0.575000000000000   0.983215233653149\\
   0.600000000000000   0.997832704868055\\
   0.625000000000000   1.011067055345495\\
   0.650000000000000   1.022721485380584\\
   0.675000000000000   1.033475121962314\\
   0.700000000000000   1.043391883071033\\
   0.725000000000000   1.052319530899368\\
   0.750000000000000   1.060507700480683\\
   0.775000000000000   1.068095693143572\\
   0.800000000000000   1.074954833767462\\
   0.825000000000000   1.081105018157154\\
   0.850000000000000   1.086883533868353\\
   0.875000000000000   1.092298842077880\\
   0.900000000000000   1.097234145717459\\
   0.925000000000000   1.101757205565379\\
   0.950000000000000   1.105817626663575\\
   0.975000000000000   1.109855319774390\\
   1.000000000000000   1.113562467224963\\
   1.025000000000000   1.116637269495467\\
   1.050000000000000   1.119796622506395\\
   1.075000000000000   1.122747831409960\\
   1.100000000000000   1.125480786766609\\
   1.125000000000000   1.127901418301359\\
   1.150000000000000   1.130012827980728\\
    1.175000000000000   1.132118075521773\\
   1.200000000000000   1.134217196788192   \\
        };
%    \addlegendentry{Lower bound of \cite{KangLiu11}}

      \addplot[color=red,dotted,line width=1.0pt]
 table[row sep=crcr]{
  0.350000000000000   0.700000000000000\\
   0.360000000000000   0.720000000000000\\
   0.370000000000000   0.740000000000000\\
   0.380000000000000   0.760000000000000\\
   0.390000000000000   0.780000000000000\\
   0.400000000000000   0.799828904932796\\
   0.410000000000000   0.817903128132570\\
   0.420000000000000   0.834214527887310\\
   0.430000000000000   0.849093237748129\\
   0.440000000000000   0.862782103214840\\
   0.450000000000000   0.875457603749385\\
   0.460000000000000   0.887251222554830\\
   0.470000000000000   0.898298877024153\\
   0.480000000000000   0.908647900100313\\
   0.490000000000000   0.918398365037833\\
   0.500000000000000   0.927596203979577\\
   0.510000000000000   0.936299755046546\\
   0.520000000000000   0.944542049299075\\
   0.530000000000000   0.952381294499923\\
   0.540000000000000   0.959841656698732\\
   0.550000000000000   0.966937370857678\\
   0.575000000000000   0.983215233653149\\
   0.600000000000000   0.997832704868055\\
   0.625000000000000   1.011067055345495\\
   0.650000000000000   1.022721485380584\\
   0.675000000000000   1.033475121962314\\
     0.700000000000000  1.043391883071033\\
   0.725000000000000   1.052391268058945\\
   0.750000000000000   1.060507700480683\\
    0.775000000000000   1.068095693143572\\
     0.800000000000000   1.074954833767462\\
      0.825000000000000   1.081105018157154\\
  0.850000000000000   1.087076235237165\\
   0.875000000000000   1.092417369712015\\
   0.900000000000000   1.097430166693560\\
   0.925000000000000   1.102569759005018\\
   0.950000000000000   1.107907401454996\\
   0.975000000000000   1.113444340058519\\
   1.000000000000000   1.119179995392207\\
   1.025000000000000   1.125117343255547\\
   1.050000000000000   1.131255653521982\\
   1.075000000000000   1.137601647857240\\
   1.100000000000000   1.144147366044726\\
   1.125000000000000    1.150897218319399\\
   1.150000000000000   1.157845904974804\\
   1.160964047443681   1.160964047443681\\
   1.175000000000000   1.160964047443681\\
   1.200000000000000   1.160964047443681\\
          };
   % \addlegendentry{Lower bound of Theorem \ref{lower bound} (Joint Gaussian dist.)}

      \addplot[color=red,solid,line width=1.0pt]
 table[row sep=crcr]{
 0.350000000000000   0.700000000000000\\
   0.360000000000000   0.720000000000000\\
   0.370000000000000   0.740000000000000\\
   0.380000000000000   0.760000000000000\\
   0.390000000000000   0.780000000000000\\
   0.400000000000000   0.799828904932796\\
   0.410000000000000   0.817903128132570\\
   0.420000000000000   0.834214527887310\\
   0.430000000000000   0.849093237748129\\
   0.440000000000000   0.862782103214840\\
   0.450000000000000   0.875457603749385\\
   0.460000000000000   0.887251222554830\\
   0.470000000000000   0.898298877024153\\
   0.480000000000000   0.908647900100313\\
   0.490000000000000   0.918398365037833\\
   0.500000000000000   0.927596203979577\\
   0.510000000000000   0.936299755046546\\
   0.520000000000000   0.944542049299075\\
   0.530000000000000   0.952381294499923\\
   0.540000000000000   0.959841656698732\\
   0.550000000000000   0.966937370857678\\
      0.575000000000000   0.983215233653149\\
   0.600000000000000   0.997832704868055\\
   0.625000000000000   1.011067055345495\\
   0.650000000000000   1.022721485380584\\
   0.675000000000000   1.033475121962314\\
     0.700000000000000  1.043391883071033\\
   0.725000000000000   1.052391268058945\\
   0.750000000000000   1.060507700480683\\
    0.775000000000000   1.068095693143572\\
     0.800000000000000   1.074954833767462\\
      0.825000000000000   1.081105018157154\\
   0.850000000000000   1.087\\
   0.88000000000000   1.095\\
   0.900000000000000   1.1\\
   0.92000000000000   1.104\\
   0.950000000000000   1.111\\
      0.97000000000000   1.116\\   
   1.000000000000000   1.123\\
      1.02000000000000   1.128\\
       1.050000000000000  1.135\\
   1.07000000000000   1.14\\
   1.100000000000000   1.147\\
   1.12000000000000   1.151\\
   1.150000000000000   1.158\\
   1.175000000000000   1.160964047443681\\
   1.200000000000000   1.160964047443681\\
          };
 %   \addlegendentry{Lower bound of Theorem \ref{lower bound} (Mixture of two Gaussian dist.)}

 \addplot[color=blue,solid,line width=1.0pt]
 table[row sep=crcr]{
   0.350000000000000   0.700000000000000\\
   0.360000000000000   0.720000000000000\\
   0.370000000000000   0.740000000000000\\
   0.380000000000000   0.760000000000000\\
   0.390000000000000   0.780000000000000\\
   0.400000000000000   0.800000000000000\\
   0.410000000000000   0.820000000000000\\
   0.420000000000000   0.840000000000000\\
   0.430000000000000   0.858656202328345\\
   0.440000000000000   0.868656202328345\\
   0.450000000000000   0.878656202328345\\
   0.460000000000000   0.888656202328345\\
   0.470000000000000   0.898656202328345\\
   0.480000000000000   0.908656202328345\\
   0.490000000000000   0.918402396567819\\
   0.500000000000000   0.927602839068990\\
   0.510000000000000   0.936306542614240\\
   0.520000000000000   0.944557438120016\\
   0.530000000000000   0.952392232669429\\
   0.540000000000000   0.959845187631181\\
   0.550000000000000   0.966943543357405\\
   0.575000000000000   0.983215233653149\\
   0.600000000000000   0.997832704868055\\
   0.625000000000000   1.011081087581264\\
   0.650000000000000   1.022847936219154\\
   0.675000000000000   1.033631641631333\\
   0.700000000000000   1.043403426152151\\
   0.725000000000000   1.052607979642753\\
   0.750000000000000   1.061513233918666\\
   0.775000000000000   1.069745876666619\\
   0.800000000000000   1.077885620608923\\
   0.825000000000000   1.085293920228713\\
   0.850000000000000   1.092626908784928\\
   0.875000000000000   1.099571999368729\\
   0.900000000000000   1.106157939549759\\
   0.925000000000000   1.112647951471246\\
   0.950000000000000   1.118792256751611\\
   0.975000000000000   1.124581266122787\\
   1.000000000000000   1.130023064943771\\
   1.025000000000000   1.135424118376827\\
   1.050000000000000   1.140488246433633\\
   1.075000000000000   1.145517070132512\\
   1.100000000000000   1.150218268255137\\
   1.125000000000000   1.154597987820095\\
   1.150000000000000   1.158951275964778\\
   1.175000000000000   1.160685261650175\\
   1.200000000000000   1.160685261650175\\
        };
%    \addlegendentry{Upper bound of Theorem \ref{thmupp}}
 \addplot[color=green!50!black,dashed,line width=1.0pt,mark=+]
 table[row sep=crcr]{
 0.350000000000000   0.700000000000000\\
   0.360000000000000   0.720000000000000\\
   0.370000000000000   0.740000000000000\\
   0.380000000000000   0.760000000000000\\
   0.390000000000000   0.780000000000000\\
   0.400000000000000   0.799828904932796\\
   0.410000000000000   0.817903128132570\\
   0.420000000000000   0.834214527887310\\
   0.430000000000000   0.849093237748129\\
   0.440000000000000   0.862782103214840\\
   0.450000000000000   0.875457603749385\\
   0.460000000000000   0.887251222554830\\
   0.470000000000000   0.898298877024153\\
   0.480000000000000   0.908647900100313\\
   0.490000000000000   0.918398365037833\\
   0.500000000000000   0.927596203979577\\
   0.510000000000000   0.936299755046546\\
   0.520000000000000   0.944542049299075\\
   0.530000000000000   0.952381294499923\\
   0.540000000000000   0.959841656698732\\
   0.550000000000000   0.966937370857678\\
   0.575000000000000   0.983215233653149\\
   0.600000000000000   0.997832704868055\\
   0.625000000000000   1.011067055345495\\
   0.650000000000000   1.022721485380584\\
   0.675000000000000   1.033475121962314\\
    0.700000000000000   1.043391883071033\\
   0.725000000000000   1.052628280897330\\
   0.750000000000000   1.061501977027410\\
   0.775000000000000   1.069734747519398\\
   0.800000000000000   1.077874616338972\\
   0.825000000000000   1.085283028396034\\
   0.850000000000000   1.092616127115080\\
   0.875000000000000   1.099561321006800\\
   0.900000000000000   1.106188079569774\\
   0.925000000000000   1.112637464934785\\
   0.950000000000000   1.118781859158838\\
   0.975000000000000   1.124570951640174\\
   1.000000000000000   1.130012827980728\\
   1.025000000000000   1.135413957776832\\
   1.050000000000000   1.140478156915528\\
   1.075000000000000   1.145507050708651\\
   1.100000000000000   1.150208313918437\\
   1.125000000000000   1.154588093739257\\
 1.150000000000000   1.158941441414707\\
   1.175000000000000   1.160964047443681\\
   1.200000000000000   1.160964047443681\\};
  %  \addlegendentry{Upper bound of Theorem \ref{theoremupp}}    
    \end{axis}

    \end{tikzpicture}%
}}] at (pt) {};

\end{tikzpicture}

%% file: binaryadderMAC.tex
\begin{tikzpicture}[scale=1]
  \begin{axis}[scale=1,
width=4.82222222222222in,
%width=0.75\columnwidth,
height=3.80333333333333in,
scale only axis,
xmin=0.7225,
xmax=0.875,
xmajorgrids,
xlabel={Link Capacity $C$},
ymin=1.445,
ymax=1.59,
ymajorgrids,
ylabel={Rate $R$},
axis x line*=bottom,
axis y line*=left,
legend pos=south east,
legend style={draw=none,fill=none,legend cell align=left, font=\small}
]
      \addplot[color=black,dashed,line width=1.0pt,]
 table[row sep=crcr]{
   0.70000000000000   1.40000000000000\\
  0.730000000000000   1.460000000000000\\
   0.732500000000000   1.465000000000000\\
   0.735000000000000   1.470000000000000\\
   0.737500000000000   1.475000000000000\\
   0.740000000000000   1.480000000000000\\
   0.742500000000000   1.485000000000000\\
   0.745000000000000   1.490000000000000\\
   0.747500000000000   1.495000000000000\\
   0.750000000000000   1.500000000000000\\
   0.752500000000000   1.505000000000000\\
   0.755000000000000   1.510000000000000\\
   0.757500000000000   1.515000000000000\\
   0.760000000000000   1.520000000000000\\
   0.762500000000000   1.525000000000000\\
   0.765000000000000   1.530000000000000\\
   0.767500000000000   1.535000000000000\\
   0.770000000000000   1.540000000000000\\
   0.772500000000000   1.545000000000000\\
   0.775000000000000   1.550000000000000\\
   0.777500000000000   1.555000000000000\\
   0.780000000000000   1.560000000000000\\
   0.782500000000000   1.565000000000000\\
   0.785000000000000   1.570000000000000\\
   0.787500000000000   1.575000000000000\\
   0.790000000000000   1.580000000000000\\
   0.792500000000000   1.584962139987238\\
   0.795000000000000   1.584962139987238\\
   0.797500000000000   1.584962139987238\\
   0.800000000000000   1.584962139987238\\
   0.802500000000000   1.584962139987238\\
   0.805000000000000   1.584962139987238\\
   0.807500000000000   1.584962139987238\\
   0.810000000000000   1.584962139987238\\
   0.812500000000000   1.584962139987238\\
   0.815000000000000   1.584962139987238\\
   0.817500000000000   1.584962139987238\\
   0.820000000000000   1.584962139987238\\
   0.822500000000000   1.584962139987238\\
   0.825000000000000   1.584962139987238\\
   0.827500000000000   1.584962139987238\\
   0.830000000000000   1.584962139987238\\
   0.832500000000000   1.584962139987238\\
   0.835000000000000   1.584962139987238\\
   0.837500000000000   1.584962139987238\\
   0.840000000000000   1.584962139987238\\
   0.842500000000000   1.584962139987238\\
   0.845000000000000   1.584962139987238\\
   0.847500000000000   1.584962139987238\\
   0.850000000000000   1.584962139987238\\
   0.852500000000000   1.584962139987238\\
   0.855000000000000   1.584962139987238\\
   0.857500000000000   1.584962139987238\\
   0.860000000000000   1.584962139987238\\
   0.862500000000000   1.584962139987238\\
   0.865000000000000   1.584962139987238\\
   0.867500000000000   1.584962139987238\\
   0.870000000000000   1.584962139987238\\
   0.872500000000000   1.584962139987238\\
   0.875000000000000   1.584962139987238\\
   0.877500000000000   1.584962139987238\\
   0.880000000000000   1.584962139987238\\
   0.882500000000000   1.584962139987238\\
   0.885000000000000   1.584962139987238\\
   0.887500000000000   1.584962139987238\\
   0.890000000000000   1.584962139987238\\
   0.892500000000000   1.584962139987238\\
   0.895000000000000   1.584962139987238\\
        };
    \addlegendentry{Cut-Set bound}
   
    \addplot[color=red,solid,line width=1.0pt,]
        plot coordinates {
           (0.70000000000000,   1.40000000000000)
        (.73,1.460000000000000)
        (.735,1.470000000000000)
        (.74,1.480000000000000)
        (.745,1.490000000000000)
        (.75,1.500000000000000)
        (.755,1.509711441752810)
	(.76,1.518845535995202)
	(.765,1.527402000000000)
	(.77,1.535378000000000)
	(.775,1.542774000000000)
	(.78,1.549588000000000)
	(.785,1.555815037178920)
	(.79,1.561453895033654)
	(.795,1.566500468757824)
	(.8,1.570950594454669)
	(.805,1.574799548505087)
	(.81,1.578042022226300)
	(.815,1.580672092687066)
	(.82,1.582683189255492)
	(.825,1.584068055375491)
	(.83,1.584818704973030)
	(.835,1.584962139987238)
	(.84,1.584962139987238)
	(.845,1.584962139987238)
	(.85,1.584962139987238)
	(.855,1.584962139987238)
	(.86,1.584962139987238)
	(.865,1.584962139987238)
	(.87,1.584962139987238)
	(.875,1.584962139987238)
	(.88,1.584962139987238)
	(.885,1.584962139987238)
	(.89,1.584962139987238)
	(.895,1.584962139987238)
        };
    \addlegendentry{Lower bound of Theorem \ref{lower bound} and Capacity}
    
      \addplot[color=blue,solid,line width=1.0pt,]
 table[row sep=crcr]{
    0.70000000000000   1.40000000000000\\
   0.730000000000000   1.460000000000000\\
   0.732500000000000   1.465000000000000\\
   0.735000000000000   1.470000000000000\\
   0.737500000000000   1.475000000000000\\
   0.740000000000000   1.480000000000000\\
   0.742500000000000   1.485000000000000\\
   0.745000000000000   1.490000000000000\\
   0.747500000000000   1.495000000000000\\
   0.750000000000000   1.500000000000000\\
   0.752500000000000   1.505000000000000\\
   0.755000000000000   1.510000000000000\\
   0.757500000000000   1.515000000000000\\
   0.760000000000000   1.520000000000000\\
   0.762500000000000   1.525000000000000\\
   0.765000000000000   1.530000000000000\\
   0.767500000000000   1.535000000000000\\
   0.770000000000000   1.540000000000000\\
   0.772500000000000   1.544053167565929\\
   0.775000000000000   1.546553167565929\\
   0.777500000000000   1.549053167565929\\
   0.780000000000000   1.551553167565929\\
   0.782500000000000   1.554053167565929\\
   0.785000000000000   1.556553167565929\\
   0.787500000000000   1.559053167565929\\
   0.790000000000000   1.561553167565929\\
   0.792500000000000   1.564053167565929\\
   0.795000000000000   1.566500468757824\\
   0.797500000000000   1.568800369438225\\
   0.800000000000000   1.570950594454669\\
   0.802500000000000   1.572950535693640\\
   0.805000000000000   1.574799548505087\\
   0.807500000000000   1.576496950823555\\
   0.810000000000000   1.578042022226300\\
   0.812500000000000   1.579434002924965\\
   0.815000000000000   1.580672092687066\\
   0.817500000000000   1.581755449683193\\
   0.820000000000000   1.582683189255492\\
   0.822500000000000   1.583454382602564\\
   0.825000000000000   1.584068055375491\\
   0.827500000000000   1.584523186179226\\
   0.830000000000000   1.584818704973030\\
   0.832500000000000   1.584953491363080\\
   0.835000000000000   1.584962139987238\\
   0.837500000000000   1.584962139987238\\
   0.840000000000000   1.584962139987238\\
   0.842500000000000   1.584962139987238\\
   0.845000000000000   1.584962139987238\\
   0.847500000000000   1.584962139987238\\
   0.850000000000000   1.584962139987238\\
   0.852500000000000   1.584962139987238\\
   0.855000000000000   1.584962139987238\\
   0.857500000000000   1.584962139987238\\
   0.860000000000000   1.584962139987238\\
   0.862500000000000   1.584962139987238\\
   0.865000000000000   1.584962139987238\\
   0.867500000000000   1.584962139987238\\
   0.870000000000000   1.584962139987238\\
   0.872500000000000   1.584962139987238\\
   0.875000000000000   1.584962139987238\\
   0.877500000000000   1.584962139987238\\
   0.880000000000000   1.584962139987238\\
   0.882500000000000   1.584962139987238\\
   0.885000000000000   1.584962139987238\\
   0.887500000000000   1.584962139987238\\
   0.890000000000000   1.584962139987238\\
   0.892500000000000   1.584962139987238\\
   0.895000000000000   1.584962139987238\\
        };
    \addlegendentry{Upper bound of Theorem \ref{thmupp}}
     \addplot[color=green!50!black,dashed,mark=+,mark options={solid}] plot coordinates {
        (0.70000000000000,   1.40000000000000)
         (0.70500000000000,   1.41000000000000)
          (0.71000000000000,   1.42000000000000)
           (0.71500000000000,   1.43000000000000)
            (0.72000000000000,   1.44000000000000)
	(0.72500000000000,   1.45000000000000)
        (.73,1.460000000000000)
        (.735,1.470000000000000)
        (.74,1.480000000000000)
        (.745,1.490000000000000)
        (.75,1.500000000000000)
        (.755,1.509711441752810)
	(.76,1.518845535995202)
	(.765,1.527402000000000)
	(.77,1.535378000000000)
	(.775,1.542774000000000)
	(.78,1.549588000000000)
	(.785,1.555815037178920)
	(.79,1.561453895033654)
	(.795,1.566500468757824)
	(.8,1.570950594454669)
	(.805,1.574799548505087)
	(.81,1.578042022226300)
	(.815,1.580672092687066)
	(.82,1.582683189255492)
	(.825,1.584068055375491)
	(.83,1.584818704973030)
	(.835,1.584962139987238)
	(.84,1.584962139987238)
	(.845,1.584962139987238)
	(.85,1.584962139987238)
	(.855,1.584962139987238)
	(.86,1.584962139987238)
	(.865,1.584962139987238)
	(.87,1.584962139987238)
	(.875,1.584962139987238)
	(.88,1.584962139987238)
	(.885,1.584962139987238)
	(.89,1.584962139987238)
	( 0.895,  1.584962139987238)
    };
    \addlegendentry{Upper bound of Theorem \ref{theoremupp} and Capacity }
                \addplot[color=black,dotted,line width=1pt,]
 table[row sep=crcr]{
                       .7929   1.564\\
                      .7929 1.54\\
                       };
\addplot[color=black,dotted,line width=1pt,]
 table[row sep=crcr]{
                       0.8333   1.585\\
                       0.8333 1.56\\
                       };
\addplot[color=black,dotted,line width=1pt,]
 table[row sep=crcr]{
                       .7659   1.528\\
                       .7659 1.52\\
                       };
                       
\addplot[color=black,dotted,line width=1pt,]
 table[row sep=crcr]{
                       .75   1.5\\
                       .75 1.52\\
                       };
                       
\addplot[color=black,dotted,line width=1pt,]
 table[row sep=crcr]{
                       .75   1.5\\
                       .75 1.48\\
                       };
\draw[fill] (axis cs:{ .8333, 1.5575}) node[above right] {$\xrightarrow{\text{\small Cut-Set bound is tight}}$};
\draw[fill] (axis cs:{ .7929, 1.5375}) node[above right] {$\xrightarrow{\text{\small Theorem \ref{thmupp} is tight}}$};
\draw[fill] (axis cs:{ .7659, 1.5175}) node[above right] {$\xrightarrow{\text{\small Theorem \ref{theoremupp} with Mrs. Gerber's lemma is tight}}$};
\draw[fill] (axis cs:{ .751, 1.5175}) node[above left] {$\xleftarrow{\substack{\text{\small Cut-Set bound}\\ \text{\small is tight}}}$};
\draw[fill] (axis cs:{ .749, 1.4775}) node[above right] {$\xrightarrow{\text{\small Theorem \ref{theoremupp} with Lemma \ref{GMGL} (Generalized Mrs. Gerber's lemma) is tight}}$};

%\draw[fill] (axis cs:{ .7929, 1.542}) node[below right] {$\longrightarrow$};
    \end{axis}
    \end{tikzpicture}

%% file: channelu.tex
\begin{figure}[t!]
\centering
\begin{tikzpicture}[scale=.9]
\tikzstyle{every node}=[draw,shape=circle];

\path (-.5,4.5) node[draw=none] () {$Y$};
\path (3.5,4.5) node[draw=none] () {$U$};

\coordinate (y0) at (0,0);
 \fill (y0) circle (4pt);
 \coordinate (y1) at (0,1.875);
 \fill (y1) circle (4pt);
 \coordinate (y2) at (0,3.75);
 \fill (y2) circle (4pt);

\path (-.5,0) node[draw=none]() {$2$};
\path (-.5,1.875) node[draw=none] () {$1$};
%\path (-.5,2.5)  node[draw=none] () {$01$};
\path (-.5,3.75) node[draw=none] () {$0$};

 \coordinate (u0) at (3,1);
 \fill (u0) circle (4pt);
 \coordinate (u1) at (3,2.5);
 \fill (u1) circle (4pt);
 \path (3.5,1)  node[draw=none] () {$1$};
\path (3.5,2.5) node[draw=none] () {$0$};

%\draw[->]
  %  (v0) -- (v1) ;
\draw[->] (y0) --node[left,draw=none,xshift=.4cm,yshift=-.4cm]{$1-\alpha$} (u0);
\draw[->] (y0) --node[left,draw=none,xshift=0cm,yshift=-.4cm]{$\alpha$} (u1);
\draw[->] (y1) --node[left,draw=none,xshift=-.3cm,yshift=-0.15cm]{$\frac{1}{2}$} (u0);
\draw[->] (y1) --node[left,draw=none,xshift=-.3cm,yshift=0.15cm]{$\frac{1}{2}$} (u1);
%\draw[->] (x10) --node[left,draw=none,xshift=-.5cm,yshift=.2cm]{$\frac{1}{2}$} (u0);
%\draw[->] (x10) --node[left,draw=none,xshift=0.1cm,yshift=-.2cm]{$\frac{1}{2}$} (u1);
\draw[->] (y2) --node[left,draw=none,xshift=.4cm,yshift=.4cm]{$1-\alpha$} (u1);
\draw[->] (y2) --node[left,draw=none,xshift=0cm,yshift=.45cm]{$\alpha$} (u0);

\end{tikzpicture}
\caption{Auxiliary channel $p(u|y)$ for the binary adder MAC.}
\label{channelu}
\end{figure}

%% file: channeluyty.tex
\begin{figure}[t!]
\centering
\begin{tikzpicture}[scale=.8]
\tikzstyle{every node}=[draw,shape=circle];

\path (-.5,4.5) node[draw=none] () {$Y$};
\path (6.5,4.5) node[draw=none] () {$U$};
\path (3.5,4.5) node[draw=none] () {$\tilde{Y}$};

\coordinate (y0) at (0,0);
 \fill (y0) circle (4pt);
 \coordinate (y1) at (0,1.875);
 \fill (y1) circle (4pt);
 \coordinate (y2) at (0,3.75);
 \fill (y2) circle (4pt);

\path (-.5,0) node[draw=none]() {$2$};
\path (-.5,1.875) node[draw=none] () {$1$};
%\path (-.5,2.5)  node[draw=none] () {$01$};
\path (-.5,3.75) node[draw=none] () {$0$};

 \coordinate (yt0) at (3,1);
 \fill (yt0) circle (4pt);
 \coordinate (yt1) at (3,2.5);
 \fill (yt1) circle (4pt);
 \path (3.5,.5)  node[draw=none] () {$1$};
\path (3.5,3) node[draw=none] () {$0$};

 \coordinate (u0) at (6,1);
 \fill (u0) circle (4pt);
 \coordinate (u1) at (6,2.5);
 \fill (u1) circle (4pt);
 \path (6.5,.5)  node[draw=none] () {$1$};
\path (6.5,3) node[draw=none] () {$0$};

%\draw[->]
  %  (v0) -- (v1) ;
\draw[->] (y0) --node[left,draw=none,xshift=.4cm,yshift=-.4cm]{$1$} (yt0);
%\draw[->] (y0) --node[left,draw=none,xshift=0cm,yshift=-.4cm]{$$} (yt1);
\draw[->] (y1) --node[left,draw=none,xshift=-.3cm,yshift=-0.15cm]{$\frac{1}{2}$} (yt0);
\draw[->] (y1) --node[left,draw=none,xshift=-.3cm,yshift=0.15cm]{$\frac{1}{2}$} (yt1);
%\draw[->] (x10) --node[left,draw=none,xshift=-.5cm,yshift=.2cm]{$\frac{1}{2}$} (u0);
%\draw[->] (x10) --node[left,draw=none,xshift=0.1cm,yshift=-.2cm]{$\frac{1}{2}$} (u1);
\draw[->] (y2) --node[left,draw=none,xshift=.4cm,yshift=.4cm]{$1$} (yt1);
%\draw[->] (y2) --node[left,draw=none,xshift=0cm,yshift=.45cm]{$$} (yt0);
\draw[->] (yt0) --node[left,draw=none,below,xshift=0cm,yshift=.4cm]{$1-\alpha$} (u0);
\draw[->] (yt0) --node[left,draw=none,above,xshift=-.4cm,yshift=0.05cm]{$\alpha$} (u1);
\draw[->] (yt1) --node[left,draw=none,below,xshift=-.4cm,yshift=-0.05cm]{$\alpha$} (u0);
\draw[->] (yt1) --node[left,draw=none,above,xshift=0cm,yshift=-.4cm]{$1-\alpha$} (u1);
\end{tikzpicture}
\caption{The auxiliary channel $p(u|y)$ as the cascade of $p(\tilde{y}|y)$ and $p(u|\tilde{y})$ for the binary adder MAC.}
\label{channeluyty}
\end{figure}

%% file: Gausscasea.tex
\begin{tikzpicture}
\begin{axis}[scale=0.65,
width=4.5in,
height=3.5in,
scale only axis,
xmin=0,
xmax=1,
xmajorgrids,
xlabel={$\rho$},
ymin=1,
ymax=1.8,
ymajorgrids,
ylabel={$R$},
axis x line*=bottom,
axis y line*=left,
legend pos=north west,
legend style={draw=none,fill=none,legend cell align=left, font=\scriptsize}
]
\addplot+[mark=none,samples=200] {2*.72};
%\addlegendentry{$f_1(C,\rho)$};
\addplot+[mark=none,samples=200] {.72+1/2*ln(1+3*(1-x^2))/ln(2)};
%\addlegendentry{$f_2(C,\rho)$};
\addplot+[mark=none,samples=200] {1/2*ln(1+2*3*(1+x))/ln(2)};
%\addlegendentry{$f_3(C,\rho)$};
\addplot+[mark=none,samples=200] {max(sqrt(1+1/4/9)-1/2/3-x,0)/(sqrt(1+1/4/9)-1/2/3-x)*(.72+1/4*ln(1+2*3*(1+x))/ln(2)-1/4*ln(1/(1-x^2))/ln(2))+max(-sqrt(1+1/4/9)+1/2/3+x,0)/(-sqrt(1+1/4/9)+1/2/3+x)*(.72+1/2*ln(1+3*(1-x^2))/ln(2))};
%\addlegendentry{$f^\prime_4(C,\rho)$};

\draw[fill=none] (axis cs:.1,1.71)  node[above right] {$f_2(C,\rho)$};
\draw[fill=none] (axis cs:.8,1.45)  node[above right] {$f_1(C)$};
\draw[fill=none] (axis cs:.17,1.575)  node[above right] {$f_3(\rho)$};
\draw[fill=none] (axis cs:.3,1.475)  node[above right] {$f^\prime_4(C,\rho)$};

\draw[fill] (axis cs:{0.307220227007870,1}) circle [radius=3pt] node[above left] {$\rho^{(1)}$};
\draw[fill] (axis cs:{0.847127088383037,1}) circle [radius=3pt] node[above left] {$\rho^{(2)}$};
\end{axis}
\end{tikzpicture}

%% file: Gausscaseb.tex
\begin{tikzpicture}
\begin{axis}[scale=0.65,
width=4.5in,
height=3.5in,
scale only axis,
xmin=0,
xmax=1,
xmajorgrids,
xlabel={$\rho$},
ymin=1,
ymax=1.8,
ymajorgrids,
ylabel={$R$},
axis x line*=bottom,
axis y line*=left,
legend pos=north west,
legend style={draw=none,fill=none,legend cell align=left, font=\scriptsize}
]
\addplot+[mark=none,samples=200] {2*.76};
%\addlegendentry{$f_1(C,\rho)$};
\addplot+[mark=none,samples=200] {.76+1/2*ln(1+3*(1-x^2))/ln(2)};
%\addlegendentry{$f_2(C,\rho)$};
\addplot+[mark=none,samples=200] {1/2*ln(1+2*3*(1+x))/ln(2)};
%\addlegendentry{$f_3(C,\rho)$};
\addplot+[mark=none,samples=200] {max(sqrt(1+1/4/9)-1/2/3-x,0)/(sqrt(1+1/4/9)-1/2/3-x)*(.76+1/4*ln(1+2*3*(1+x))/ln(2)-1/4*ln(1/(1-x^2))/ln(2))+max(-sqrt(1+1/4/9)+1/2/3+x,0)/(-sqrt(1+1/4/9)+1/2/3+x)*(.76+1/2*ln(1+3*(1-x^2))/ln(2))};
%\addlegendentry{$f^\prime_4(C,\rho)$};

\draw[fill=none] (axis cs:.23,1.71)  node[above right] {$f_2(C,\rho)$};
\draw[fill=none] (axis cs:.8,1.525)  node[above right] {$f_1(C)$};
\draw[fill=none] (axis cs:.17,1.575)  node[above right] {$f_3(\rho)$};
\draw[fill=none] (axis cs:.2,1.4)  node[above right] {$f^\prime_4(C,\rho)$};

\draw[fill] (axis cs:{0.307220227007870,1}) circle [radius=3pt] node[above left] {$\rho^{(1)}$};
\draw[fill] (axis cs:{0.847127088383037,1}) circle [radius=3pt] node[above left] {$\rho^{(2)}$};
\end{axis}
\end{tikzpicture}

%% file: Gausscasec.tex
\begin{tikzpicture}
\begin{axis}[scale=0.65,
width=4.5in,
height=3.5in,
scale only axis,
xmin=0,
xmax=1,
xmajorgrids,
xlabel={$\rho$},
ymin=1.2,
ymax=3,
ymajorgrids,
ylabel={$R$},
axis x line*=bottom,
axis y line*=left,
legend pos=north west,
legend style={draw=none,fill=none,legend cell align=left, font=\scriptsize}
]
\addplot+[mark=none,samples=200] {2*1.2};
%\addlegendentry{$f_1(C,\rho)$};
\addplot+[mark=none,samples=200] {1.2+1/2*ln(1+3*(1-x^2))/ln(2)};
%\addlegendentry{$f_2(C,\rho)$};
\addplot+[mark=none,samples=200] {1/2*ln(1+2*3*(1+x))/ln(2)};
%\addlegendentry{$f_3(C,\rho)$};
\addplot+[mark=none,samples=200] {max(sqrt(1+1/4/9)-1/2/3-x,0)/(sqrt(1+1/4/9)-1/2/3-x)*(1.2+1/4*ln(1+2*3*(1+x))/ln(2)-1/4*ln(1/(1-x^2))/ln(2))+max(-sqrt(1+1/4/9)+1/2/3+x,0)/(-sqrt(1+1/4/9)+1/2/3+x)*(1.2+1/2*ln(1+3*(1-x^2))/ln(2))};
%\addlegendentry{$f^\prime_4(C,\rho)$};

\draw[fill=none] (axis cs:.2,2.16)  node[above right] {$f_2(C,\rho)$};
\draw[fill=none] (axis cs:.2,2.41)  node[above right] {$f_1(C)$};
\draw[fill=none] (axis cs:.17,1.565)  node[above right] {$f_3(\rho)$};
\draw[fill=none] (axis cs:.2,1.93)  node[above right] {$f^\prime_4(C,\rho)$};

\draw[fill] (axis cs:{0.307220227007870,1.2}) circle [radius=3pt] node[above left] {$\rho^{(1)}$};
\draw[fill] (axis cs:{0.847127088383037,1.2}) circle [radius=3pt] node[above left] {$\rho^{(2)}$};

\end{axis}
\end{tikzpicture}

%% file: lambdafig.tex
\begin{tikzpicture}
\begin{axis}[scale=0.65,
width=6.5in,
height=4in,
scale only axis,
xmin=0,
xmax=1,
xmajorgrids,
xlabel={$\rho$},
ymin=1,
ymax=2,
ymajorgrids,
ylabel={$R$},
axis x line*=bottom,
axis y line*=left,
legend pos=north west,
legend style={draw=none,fill=none,legend cell align=left, font=\scriptsize}
]
\addplot+[mark=none,samples=200] {2*.82};
%\addlegendentry{$f_1(C,\rho)$};
\addplot+[mark=none,samples=200] {.82+1/2*ln(1+3*(1-x^2))/ln(2)};
%\addlegendentry{$f_2(C,\rho)$};
\addplot+[mark=none,samples=200] {1/2*ln(1+2*3*(1+x))/ln(2)};
%\addlegendentry{$f_3(C,\rho)$};
\addplot+[mark=none,samples=200] {2*.82-1/2*ln(1/(1-x^2))/ln(2)};
%\addlegendentry{$f^\prime_4(C,\rho)$};
\addplot+[green!50!black,mark=none,samples=200,dashed] {2*.82-1/2*ln(1+7.9807+6*(1+.303317))/ln(2)-1/2*ln(1+7.9807)/ln(2)+ln(1+7.9807+3*(1-x^2))/ln(2)};

\draw[fill=none] (axis cs:.1,1.82)  node[above right] {$f_2(C,\rho)$};
\draw[fill=none] (axis cs:.8,1.63)  node[above right] {$f_1(C)$};
\draw[fill=none] (axis cs:.05,1.475)  node[above right] {$f_3(\rho)$};
\draw[fill=none] (axis cs:.425,1.515)  node[above right] {$f_5(C,\rho,N_\lambda)$};
\draw[fill=none] (axis cs:.475,1.2)  node[above right] {$f_0(C,\rho)$};
\draw[fill] (axis cs:{.303317,1}) circle [radius=3pt] node[above left] {$\lambda$};
\end{axis}
\end{tikzpicture}

%% file: binarycasea1.tex
\begin{tikzpicture}
\begin{axis}[scale=0.65,
width=4.5in,
height=3.5in,
scale only axis,
xmin=0,
xmax=.5,
xmajorgrids,
xlabel={$p$},
ymin=1,
ymax=1.8,
ymajorgrids,
ylabel={$R$},
axis x line*=bottom,
axis y line*=left,
legend pos=north west,
legend style={draw=none,fill=none,legend cell align=left, font=\scriptsize}
]
\addplot+[mark=none,samples=500] {2*.76};
%\addlegendentry{$f_1(C,\rho)$};
\addplot+[mark=none,samples=500] {1-x-x*ln(x)/ln(2)-(1-x)*ln(1-x)/ln(2)};
%\addlegendentry{$f_2(C,\rho)$};
\addplot+[mark=none,samples=500] {.76-x*ln(x)/ln(2)-(1-x)*ln(1-x)/ln(2)-x/2};
%\addlegendentry{$f_3(C,\rho)$};

%\addlegendentry{$f^\prime_4(C,\rho)$};

\draw[fill=none] (axis cs:.06,1.4)  node[above right] {$g_3(p)$};
\draw[fill=none] (axis cs:.09,1.52)  node[above right] {$g_1(C)$};
\draw[fill=none] (axis cs:.14,1.2)  node[above right] {$g_4(C,p)$};

\draw[fill] (axis cs:{0.3333,1}) circle [radius=3pt] node[above left] {$\frac{1}{3}$};
\draw[fill] (axis cs:{0.4142,1}) circle [radius=3pt] node[above left] {$p^{(1)}$};
\draw[fill] (axis cs:{0.48,1}) circle [radius=3pt] node[above left] {$p^{(3)}$};

\end{axis}
\end{tikzpicture}

%% file: binarycaseb.tex
\begin{tikzpicture}
\begin{axis}[scale=0.65,
width=4.5in,
height=3.5in,
scale only axis,
xmin=0,
xmax=.5,
xmajorgrids,
xlabel={$p$},
ymin=1,
ymax=1.8,
ymajorgrids,
ylabel={$R$},
axis x line*=bottom,
axis y line*=left,
legend pos=north west,
legend style={draw=none,fill=none,legend cell align=left, font=\scriptsize}
]
\addplot+[mark=none,samples=500] {2*.78};
%\addlegendentry{$f_1(C,\rho)$};
\addplot+[mark=none,samples=500] {1-x-x*ln(x)/ln(2)-(1-x)*ln(1-x)/ln(2)};
%\addlegendentry{$f_2(C,\rho)$};
\addplot+[mark=none,samples=500] {.78-x*ln(x)/ln(2)-(1-x)*ln(1-x)/ln(2)-x/2};
%\addlegendentry{$f_3(C,\rho)$};

%\addlegendentry{$f^\prime_4(C,\rho)$};

\draw[fill=none] (axis cs:.06,1.4)  node[above right] {$g_3(p)$};
\draw[fill=none] (axis cs:.09,1.56)  node[above right] {$g_1(C)$};
\draw[fill=none] (axis cs:.13,1.2)  node[above right] {$g_4(C,p)$};

\draw[fill] (axis cs:{0.3333,1}) circle [radius=3pt] node[above left] {$\frac{1}{3}$};
\draw[fill] (axis cs:{0.4142,1}) circle [radius=3pt] node[above left] {$p^{(1)}$};
\draw[fill] (axis cs:{0.44,1}) circle [radius=3pt] node[above right] {$p^{(3)}$};

\end{axis}
\end{tikzpicture}

%% file: binarycasec.tex
\begin{tikzpicture}
\begin{axis}[scale=0.65,
width=4.5in,
height=3.5in,
scale only axis,
xmin=0,
xmax=.5,
xmajorgrids,
xlabel={$p$},
ymin=1,
ymax=1.8,
ymajorgrids,
ylabel={$R$},
axis x line*=bottom,
axis y line*=left,
legend pos=north west,
legend style={draw=none,fill=none,legend cell align=left, font=\scriptsize}
]
\addplot+[mark=none,samples=500] {2*.8};
%\addlegendentry{$f_1(C,\rho)$};
\addplot+[mark=none,samples=500] {1-x-x*ln(x)/ln(2)-(1-x)*ln(1-x)/ln(2)};
%\addlegendentry{$f_2(C,\rho)$};
\addplot+[mark=none,samples=500] {.8-x*ln(x)/ln(2)-(1-x)*ln(1-x)/ln(2)-x/2};
%\addlegendentry{$f_3(C,\rho)$};

%\addlegendentry{$f^\prime_4(C,\rho)$};

\draw[fill=none] (axis cs:.06,1.4)  node[above right] {$g_3(p)$};
\draw[fill=none] (axis cs:.09,1.6)  node[above right] {$g_1(C)$};
\draw[fill=none] (axis cs:.12,1.2)  node[above right] {$g_4(C,p)$};

\draw[fill] (axis cs:{0.3333,1}) circle [radius=3pt] node[above left] {$\frac{1}{3}$};
\draw[fill] (axis cs:{0.4142,1}) circle [radius=3pt] node[above right] {$p^{(1)}$};
\draw[fill] (axis cs:{0.4,1}) circle [radius=3pt] node[above] {$p^{(3)}$};

\end{axis}
\end{tikzpicture}

%% file: binarycased.tex
\begin{tikzpicture}
\begin{axis}[scale=0.65,
width=4.5in,
height=3.5in,
scale only axis,
xmin=0,
xmax=.5,
xmajorgrids,
xlabel={$p$},
ymin=1,
ymax=1.8,
ymajorgrids,
ylabel={$R$},
axis x line*=bottom,
axis y line*=left,
legend pos=north west,
legend style={draw=none,fill=none,legend cell align=left, font=\scriptsize}
]
\addplot+[mark=none,samples=500] {1.795};
%\addlegendentry{$f_1(C,\rho)$};
\addplot+[mark=none,samples=500] {1-x-x*ln(x)/ln(2)-(1-x)*ln(1-x)/ln(2)};
%\addlegendentry{$f_2(C,\rho)$};
\addplot+[mark=none,samples=500] {.9-x*ln(x)/ln(2)-(1-x)*ln(1-x)/ln(2)-x/2};
%\addlegendentry{$f_3(C,\rho)$};

%\addlegendentry{$f^\prime_4(C,\rho)$};

\draw[fill=none] (axis cs:.06,1.4)  node[above right] {$g_3(p)$};
\draw[fill=none] (axis cs:.09,1.7)  node[above right] {$g_1(C)$};
\draw[fill=none] (axis cs:.1,1.2)  node[above right] {$g_4(C,p)$};

\draw[fill] (axis cs:{0.3333,1}) circle [radius=3pt] node[above left] {$\frac{1}{3}$};
\draw[fill] (axis cs:{0.4142,1}) circle [radius=3pt] node[above left] {$p^{(1)}$};
\draw[fill] (axis cs:{0.2,1}) circle [radius=3pt] node[above left] {$p^{(3)}$};

\end{axis}
\end{tikzpicture}

%% file: functionscross.tex
\begin{tikzpicture}
\begin{axis}[scale=0.65,
width=4.5in,
height=3.5in,
scale only axis,
xmin=0,
xmax=.5,
xmajorgrids,
xlabel={$q$},
ymin=1,
ymax=1.8,
ymajorgrids,
ylabel={$R$},
axis x line*=bottom,
axis y line*=left,
legend pos=north west,
legend style={draw=none,fill=none,legend cell align=left, font=\scriptsize}
]
\addplot+[mark=none,samples=500] {1.6};
%\addplot+[mark=none,samples=500] {.8-x*ln(x)/ln(2)-(1-x)*ln(1-x)/ln(2)};

\addplot+[mark=none,samples=500] {.8+.8-1-x*ln(x)/ln(2)-(1-x)*ln(1-x)/ln(2)};
%\addlegendentry{$f_1(C,\rho)$};
\addplot+[mark=none,samples=500] {-x*ln(x)/ln(2)-(1-x)*ln(1-x)/ln(2)+1-x};
%\addlegendentry{$f_2(C,\rho)$};
%\addplot+[mark=none,samples=500] {.77+.77-1-(1-x)*(-(0.238108599560538)*ln(0.238108599560538)/ln(2)-(1-(0.238108599560538))*ln(1-(0.238108599560538))/ln(2))-x+2*(-(1/2*x+(1-x)*(1-0.238108599560538))*ln(1/2*x+(1-x)*(1-0.238108599560538))/ln(2)-(1-(1/2*x+(1-x)*(1-0.238108599560538)))*ln(1-(1/2*x+(1-x)*(1-0.238108599560538)))/ln(2))};
%\addlegendentry{$f_3(C,\rho)$};
  \addplot[color=black,solid,line width=1.0pt,]
 table[row sep=crcr]{
                       0   1.150047759582758\\
   0.010000000000000   1.166067675773409\\
   0.020000000000000   1.181735630834823\\
   0.030000000000000   1.197059885380733\\
   0.040000000000000   1.212048242897978\\
   0.050000000000000   1.226708085539892\\
   0.060000000000000   1.241046406259174\\
   0.070000000000000   1.255069837736192\\
   0.080000000000000   1.268784678492128\\
   0.090000000000000   1.282196916520977\\
   0.100000000000000   1.295312250727874\\
   0.110000000000000   1.308136110422217\\
   0.120000000000000   1.320673673080961\\
   0.130000000000000   1.332929880569494\\
   0.140000000000000   1.344909453983621\\
   0.150000000000000   1.356616907255805\\
   0.160000000000000   1.368056559651349\\
   0.170000000000000   1.379232547265126\\
   0.180000000000000   1.390148833616533\\
   0.190000000000000   1.400809219429059\\
   0.200000000000000   1.411217351671129\\
   0.210000000000000   1.421376731926374\\
   0.220000000000000   1.431290724154078\\
   0.230000000000000   1.440962561894021\\
   0.240000000000000   1.450395354964307\\
   0.250000000000000   1.459592095695662\\
   0.260000000000000   1.468555664741367\\
   0.270000000000000   1.477497285106324\\
   0.280000000000000   1.487592516807782\\
   0.290000000000000   1.497134079636488\\
   0.300000000000000   1.506143427334505\\
   0.310000000000000   1.514640860645825\\
   0.320000000000000   1.522645607976059\\
   0.330000000000000   1.530175896791584\\
   0.340000000000000   1.537249017197465\\
   0.350000000000000   1.543881378878258\\
   0.360000000000000   1.550088562381492\\
   0.370000000000000   1.555885365558830\\
   0.380000000000000   1.561285845846141\\
   0.390000000000000   1.566303358954537\\
   0.400000000000000   1.570950594454669\\
   0.410000000000000   1.576700308302004\\
   0.420000000000000   1.582251285667985\\
   0.430000000000000   1.587605115373997\\
   0.440000000000000   1.592763334517395\\
   0.450000000000000   1.597727430292778\\
   0.460000000000000   1.602498841717794\\
   0.470000000000000   1.607078961269345\\
   0.480000000000000   1.611469136435629\\
   0.490000000000000   1.615670671189026\\
   0.500000000000000   1.619684827384522\\};
%\addlegendentry{$f^\prime_4(C,\rho)$};
  \addplot[color=black,dotted,line width=.5pt,]
 table[row sep=crcr]{
 0  1.570950594454669\\
   %                    .4 1.570950594454669\\
                       .5 1.570950594454669\\
                       };
  \addplot[color=black,dotted,line width=.5pt,]
 table[row sep=crcr]{
                       .2688   1.570950594454669\\
                       .2688 0\\
                       };
                         \addplot[color=black,dotted,line width=.5pt,]
 table[row sep=crcr]{
                       .4   1.570950594454669\\
                       .4 0\\
                       };

\draw[fill=none] (axis cs:.05,1.6)  node[above right] {$g_1(C_1,C_2)$};
\draw[fill=none] (axis cs:0.015,1.32)  node[above right] {$g_3(q)$};
\draw[fill=none] (axis cs:.1,1.325)  node[above right] {\rotatebox{25}{$g_5(C_1,C_2,q,\alpha_\eta)$}};
\draw[fill=none] (axis cs:.18,1.2)  node[above right] {$g_0(C_1,C_2,q)$};
%\node at (axis cs:4,16) [pin={-10:$x^2$},inner sep=0pt] {};
%\draw[fill] (axis cs:{0.3333,1}) circle [radius=3pt] node[above left] {$\frac{1}{3}$};
\draw[fill] (axis cs:{0.4,1}) circle [radius=3pt] node[above left] {$\eta$};
\draw[fill] (axis cs:{.2688,1}) circle [radius=3pt] node[above left] {$\tilde{\eta}$};
\draw[fill] (axis cs:{0,1.570950594454669}) circle [radius=3pt] node[below right] {$R^{(l)}_{\max}$};

%\draw[fill] (axis cs:{0.48,1}) circle [radius=3pt] node[above left] {$p^{(3)}$};

\end{axis}
\end{tikzpicture}